\documentclass[11pt]{article}
\usepackage{amsfonts}
\usepackage{amssymb,amsmath,amsthm}
\usepackage{epsfig}
\usepackage{color}
\usepackage{latexsym}
\usepackage{enumerate}
\usepackage{algorithm}
\usepackage{algorithmic}

\usepackage{tweaklist}
\renewcommand{\enumhook}{\setlength{\topsep}{2pt}%
  \setlength{\itemsep}{-2pt}}
\renewcommand{\itemhook}{\setlength{\topsep}{2pt}%
  \setlength{\itemsep}{-2pt}}
\usepackage[dvips, paper=letterpaper, top=1in, bottom=.75in, left=1in, right=1in, nohead, includefoot, footskip=.25in]{geometry}

\newtheorem{theorem}{Theorem}
\newtheorem{corollary}{Corollary}
\newtheorem{lemma}{Lemma}

\newtheorem{proposition}{Proposition}

\newtheorem{definition}{Definition}

\newtheorem{observation}{Observation}

\newtheorem{remark}{Remark}

\newtheorem{informaltheorem}{Informal Theorem}
\newenvironment{prevproof}[2]{\noindent {\em {Proof of {#1}~\ref{#2}:}}}{$\Box$\vskip \belowdisplayskip}

\newcommand{\junk}[1]{}
\junk{

}

\newcommand{\poly}{\text{poly}}
\newcommand{\opt}{\text{OPT}}

\newcommand{\mattnote}[1]{{\color{blue}{#1}}}

\newcommand{\notshow}[1]{{}}

\DeclareMathOperator{\argmax}{argmax}

\definecolor{MyGray}{rgb}{0.8,0.8,0.8}

\begin{document}
\title{Optimal Multi-Dimensional Mechanism Design: Reducing Revenue to Welfare Maximization}
\author {Yang Cai\thanks{Supported by NSF Award CCF-0953960 (CAREER) and CCF-1101491.}\\
EECS, MIT \\
\tt{ycai@csail.mit.edu}
\and
Constantinos Daskalakis\thanks{Supported by a Sloan Foundation Fellowship and NSF Award CCF-0953960 (CAREER) and CCF-1101491.}\\
EECS, MIT \\
\tt{costis@mit.edu}
\and
S. Matthew Weinberg\thanks{Supported by a NSF Graduate Research Fellowship and NSF Award CCF-1101491.}\\
EECS, MIT\\
\tt{smw79@mit.edu}
}
\addtocounter{page}{-1}
\maketitle
\begin{abstract}
We provide a reduction from revenue maximization to welfare maximization in multi-dimensional {Bayesian auctions} with arbitrary {(possibly combinatorial)} feasibility constraints and independent bidders with arbitrary {(possibly combinatorial)} demand constraints, appropriately extending Myerson's single-dimensional result~\cite{Myerson81} to this setting. We also show that every feasible Bayesian auction can be implemented as a distribution over \emph{virtual VCG allocation rules}. A virtual VCG allocation rule has the following simple form: {Every bidder's type $t_i$} is transformed into a virtual type $f_i(t_i)$, {via a bidder-specific function}. Then, the allocation maximizing virtual welfare is chosen. Using this {characterization}, we show how to find and run the revenue-optimal auction given only black box access to an implementation of {the VCG allocation rule}. We generalize this result to arbitrarily correlated bidders, introducing the notion of a \emph{second-order} VCG allocation rule.

We obtain our reduction {from revenue to welfare optimization} via two algorithmic results on reduced form auctions in settings with arbitrary feasibility {and demand} constraints. First, we provide a separation oracle for determining feasibility of a reduced form auction. Second, we provide a geometric algorithm to decompose any feasible reduced form into a distribution over virtual VCG allocation rules. In addition, we show how to execute both algorithms given only black box access to an implementation of {the VCG allocation rule}. 

Our results are computationally efficient for all multi-dimensional settings where the bidders are additive (or can be efficiently mapped to be additive). In this case, {our mechanisms run in time polynomial in the number of items and the total number of bidder types, but \emph{not} type profiles. This running time is polynomial in the number of items, the number of bidders, and {the cardinality of the support of each bidder's value distribution}. For generic correlated distributions, this is the natural description complexity of the problem.}
The running time can be further improved to polynomial in only the number of items and the number of bidders in item-symmetric settings by making use of techniques from~\cite{DW12}.
\end{abstract}
\thispagestyle{empty}

\newpage

\section{Introduction} \label{sec:intro}
The \emph{multi-dimensional mechanism design} problem has received much attention from the economics community, and recently the computer science community as well. The problem description is simple: a seller has a limited supply of several heterogenous items for sale and many interested buyers. The goal is for the seller to design an auction for the buyers to play that will maximize her revenue. In order to make this problem tractable (not just computationally, but at all), some assumptions must be made. First, we assume that the seller has some bayesian prior $\mathcal{D}$ on the \emph{types} of buyers that will show up to the auction. Second, we assume that the buyers have the same prior as the seller, and that they will play any auction at a \emph{Bayes-Nash Equilibrium}. We also assume that all buyers are \emph{quasi-linear} and \emph{risk-neutral}, terms that are defined formally in Section~\ref{sec:notation}. Finally, we say that the goal of the seller is to maximize her \emph{expected revenue} over all auctions when played at a Bayes-Nash Equilibrium. All of these assumptions have become standard with regards to this problem. Indeed, all were made in Myerson's seminal paper on revenue-maximizing mechanism design where this problem is solved for a single item and product distributions~\cite{Myerson81}. In addition, Myerson introduces the \emph{revelation principle}, showing that every auction played at a Bayes-Nash Equilibrium is strategically equivalent to a \emph{Bayesian Incentive Compatible} (BIC) \emph{direct revelation} mechanism. In a direct revelation mechanism, each bidder reports a bid for each possible subset of items they may receive. A direct revelation mechanism is called BIC if it is a Bayes-Nash Equilibrium for each bidder to bid exactly their value for each subset. In essence, Myerson's revelation principle says that one only need to consider BIC direct revelation mechanisms rather than arbitrary auctions played at a Bayes-Nash Equilibrium to maximize revenue (or any other objective for that matter).

As we depart from Myerson's single-item setting, the issue of \emph{feasibility} arises. With only a single item for sale, it is clear that the right feasibility constraints are simply that the item is always awarded to at most a single bidder. With many heterogenous items, there are many natural scenarios that we would like to model. Here are some examples:
\begin{enumerate}
\item Maybe the items are houses. In this case, a feasible allocation awards each house to at most one bidder, and to each bidder at most one house.
\item Maybe the items are appointment slots with doctors. In this case, a feasible allocation does not award the same slot to more than one bidder, and does not award a bidder more than one  slot with the same doctor, or overlapping slots with different doctors.
\item Maybe the items are bridges built at different locations. In this case, a feasible allocation awards each bridge to everyone or to no one.
\end{enumerate}

Sometimes, feasibility constraints are imposed by the supply side of the problem: a doctor cannot meet with two patients at once, and a bridge cannot be built for one bidder but not another. Other times, feasibility constraints are imposed by the demand side of the problem: no bidder wants two houses or two appointments with the same doctor. Without differentiating where feasibility constraints come from, we model them in the following way: let ${\cal A} = [m] \times [n]$ denote the space of assignments (where $(i,j)$ denotes that bidder $i$ is assigned item $j$), and let $\mathcal{F}$ be a set system on ${\cal A}$ (that is, a subset of $2^{\cal A}$). Then in a setting with feasibility constraints $\mathcal{F}$, it is possible for the seller to simultaneously {make} any subset of {assignments} in $\mathcal{F}$. $\mathcal{F}$ may be a truly arbitrary set system, \emph{it need not even be downward-closed.}

As we leave the single-dimensional setting, we also need to consider how a bidder values a bundle of multiple items. In general, a bidder may have arbitrarily complicated ways of evaluating bundles of items, and this information is encoded into the bidder's type. For the problem to be computationally meaningful, however, one would want to either assume that the auctioneer only has oracle access to a bidder's valuation, or impose some structure on the bidders' valuations allowing them to be succinctly described. Indeed, virtually every recent result in revenue-maximizing literature~\cite{AlaeiFHHM12,BhattacharyaGGM10,CaiD11,CaiDW12, CaiH12, ChawlaHK07,ChawlaHMS10,DW12,KleinbergW12} assumes that bidders are {\em capacitated-additive}.\footnote{A bidder is capacitated-additive if for some constant $C$ her value for any subset $S$ of at most $C$ goods is equal to the sum of her values for each item in $S$, and her value for any subset $S$ of more than $C$ goods is equal to her value for her favorite $S' \subseteq S$ of at most $C$ goods.} {In fact, most results are for unit-demand bidders.}  It is easy to see that, if we {are allowed to incorporate arbitrary} demand constraints into the definition of $\mathcal{F}$, such bidders can be described in our model as simply additive. In fact, far more complex bidders can be modeled as well, as demand constraints could instead be some arbitrary set system. Because $\mathcal{F}$ is already an arbitrary set system, we may model bidders as simply additive and still capture virtually every bidder model studied in recent results, and more general ones as well. In fact, we note that every multi-dimensional setting can be mapped to an additive one, albeit not necessarily computationally efficiently.\footnote{The generic transformation is to introduce a meta-item for every possible subset of the items, and have the feasibility constraints (which are allowed to be arbitrary) be such that an allocation is feasible if and only if each bidder receives at most one meta-item, and the corresponding allocation of real items (via replacing each meta-item with the subset of real items it represents) is feasible in the original setting. Many non-additive settings allow much more computationally efficient transformations than the generic one.} So while we focus our discussion to additive bidders throughout this paper, our results apply to every auction setting, without need for any additivity assumption. In particular, our characterization result (Informal Theorem~\ref{infthm:characterization}) of feasible allocation rules holds for any multi-dimensional setting, and our reduction from revenue to welfare optimization (Informal Theorem~\ref{infthm:PTAS}) also holds for any setting, and we show that it can be carried out computationally efficiently for any additive setting.

\paragraph{Optimal Multi-dimensional Mechanism Design.} With the above motivation in mind, we formally state the revenue optimization problem we solve. We remark that virtually every known result in the multi-dimensional mechanism design literature (see references above) tackles a special case of this problem, possibly with budget constraints on the bidders (which can be easily incorporated in all results presented in this paper as discussed in Appendix~\ref{app:budgets}), and possibly replacing BIC with IC. We explicitly assume in the definition of the problem that the bidders are additive, recalling that this is not a restriction if computational considerations are not in place.

\smallskip\noindent  \framebox{
\begin{minipage}{\hsize}
\textbf{Revenue-Maximizing Multi-Dimensional Mechanism Design Problem (MDMDP):} Given as input $m$ distributions (possibly correlated across items) $\mathcal{D}_1,\ldots,\mathcal{D}_m$ over valuation vectors for $n$ heterogenous items and feasibility constraints $\mathcal{F}$, output a BIC mechanism $M$ whose allocation is in $\mathcal{F}$ with probability $1$ and whose expected revenue is optimal relative to any other, possibly randomized, BIC mechanism when played by $m$ additive bidders whose valuation vectors are sampled from $\mathcal{D} = \times_i \mathcal{D}_i$.
\end{minipage}}

\smallskip We provide a poly-time black box reduction from the MDMDP with feasibility constraints $\mathcal{F}$ to implementing VCG with feasibility constraints $\mathcal{F}$ by introducing the notion of a \emph{virtual VCG allocation rule}. A virtual VCG allocation rule is defined by a collection of functions $f_i$ for each bidder $i$. $f_i$ takes as input bidder $i$'s reported bid vector and outputs a virtual bid vector. When the reported bid vectors are $\vec{v}_1,\ldots,\vec{v}_m$, the virtual VCG allocation rule with functions $\{f_i\}_{i\in[m]}$ simply implements the VCG allocation rule (with feasibility constraints $\mathcal{F}$) on the virtual bid vectors $f_1(\vec{v}_1),\ldots,f_m(\vec{v}_m)$. We also note here that implementing VCG for additive bidders is in general \emph{much} easier than implementing VCG for arbitrary bidders.\footnote{When bidders are additive, implementing VCG is simply solving the following problem, which is very well understood for a large class of feasibility constraints: every element of ${\cal A}$ has a weight. The weight of any subset of ${\cal A}$ is equal to the sum of the weights of its elements. Find the max-weight subset of ${\cal A}$ that is in $\mathcal{F}$.} Our solution to the MDMDP is informally stated below, and is formally given as Theorem~\ref{thm:general} of Section~\ref{sec:revenue}:
\begin{informaltheorem}\label{infthm:PTAS} Let $A_{\mathcal{F}}$ be an implementation of the VCG allocation rule with respect to $\mathcal{F}$ (i.e. $A_{\mathcal{F}}$ takes as input a profile of bid vectors and outputs the VCG allocation). Then {for all $\mathcal{D}_1,\ldots,\mathcal{D}_m$} with finite support and all $\mathcal{F}$, {given $\mathcal{D}_1,\ldots,\mathcal{D}_m$ and black box access to $A_{\mathcal{F}}$ (and without need of knowledge of ${\cal F}$)}, there exists a fully polynomial-time randomized approximation scheme\footnote{{This is often abbreviated as FPRAS, and we provide its formal definition in Section~\ref{sec:notation}.}} for the MDMDP whose runtime is polynomial in {$n$},  the number of bidder types (and \emph{not} type profiles), and the runtime of $A_{\mathcal{F}}$. Furthermore, the allocation rule {of the output mechanism} is a distribution over virtual VCG allocation rules. 
\end{informaltheorem}

We remark that the functions defining a virtual VCG allocation rule may map a bidder type to a vector with negative coordinates. Therefore, our given implementation of the VCG allocation rule should be able to handle negative weights. This is not a restriction for arbitrary downwards-closed $\mathcal{F}$ as any implementation of VCG that works for non-negative weights can easily be (in a black-box way) converted into an implementation of VCG allowing arbitary (possibly negative) inputs.\footnote{The following simple black-box transformation achieves this: first zero-out all negative coordinates in the input vectors; then run VCG; in the VCG allocation, un-allocate item $j$ from bidder $i$ if the corresponding coordinate is negative; this is still a feasible allocation as the setting is downwards-closed.} But this is not necessarily true for non downwards-closed ${\cal F}$'s. If the given $A_{\mathcal{F}}$ cannot accommodate negative weights, we need to replace it with an algorithm that can in order for our results to be applicable.

Several extensions are stated and discussed in Section~\ref{sec:revenue}, including solutions for distributions of infinite support, and improved runtimes in certain cases that make use of techniques from~\cite{DW12}. We also extend all our solutions to accommodate strong budget constraints by the bidders in Appendix~\ref{app:budgets}. So how does our solution compare to Myerson's single-dimensional result? One interpretation of Myerson's optimal auction is the following: First, he shows that the allocation rule used by the optimal auction is just the Vickrey allocation rule, but {on} virtual bids instead of true bids. Second, he provides a closed form for each virtual transformation using (ironed) virtual values. And finally, he provides a closed form pricing rule that makes the entire mechanism BIC (in fact, IC). In the multi-dimensional setting, it is known that randomness is necessary to achieve optimal revenue, even with a single bidder and two items~\cite{BriestCKW10,ChawlaMS10}, so we cannot possibly hope for a solution as clean as Myerson's. However, we have come quite close in a very general setting: Our allocation rule is just a distribution over virtual VCG allocation rules. And instead of a closed form for each virtual transformation and the pricing rule, we provide a computationally efficient algorithm to find them. 

\paragraph{Characterization of Feasible Interim Allocation Rules.} In addition to our solution of the MDMDP, we provide a characterization of feasible interim allocation rules of multi-dimensional mechanisms in all (not necessarily additive) settings.\footnote{For non-additive settings, the characterization is more usable for the purposes of mechanism design when applied to meta-items (see discussion above), although it still holds when directly applied to items as well.}  We show the following informal theorem, which is stated formally as Theorem~\ref{thm:characterization} in Section~\ref{sec:independent}. Recall that a virtual VCG allocation rule is associated with a collection of functions $f_i$ that map types $t_i$ to virtual types $f_i(t_i)$ for each bidder $i$, and allocates the items as follows: for a given type vector $(t_1,...,t_m)$, the bidders' types  are transformed into virtual types $(f_1(t_1),\ldots,f_m(t_m))$; then the virtual welfare optimizing allocation is chosen.

\begin{informaltheorem}\label{infthm:characterization}

Let $\mathcal{F}$ be any set system of feasibility constraints, and $\mathcal{D}$ any (possibly correlated) distribution over bidder types. Then the interim allocation rule of any feasible mechanism can be implemented as the interim rule of a distribution over virtual VCG allocation rules.

\end{informaltheorem}

\subsection{Related Work}
\subsubsection{Structural Results}
Some structural results are already known for special cases of the MDMDP and its extension to correlated bidders. As we have already discussed, Myerson showed that the revenue-optimal auction for selling a single item is a virtual Vickrey auction: bids are transformed to virtual bids, and the item is awarded to the bidder with the highest non-negative virtual value~\cite{Myerson81}. It was later shown that this approach also applies to all single-dimensional settings (i.e. when bidders can't tell the difference between different houses, appointment slots, bridges, etc) as long as bidders' values are independent. In this setting, bids are transformed to virtual bids (via Myerson's transformation), and the virtual-welfare-maximizing feasible allocation is chosen. These structural results are indeed strong, but hold only in the single-dimensional setting and are therefore of very limited applicability.

On the multi-dimensional front, it was recently shown that similar structure exists in restricted settings. It is shown in~\cite{CaiDW12} that when selling multiple heterogenous items to additive bidders with \emph{no} demand constraints (i.e. $\mathcal{F}$ only ensures that each item is awarded to at most one bidder), the optimal auction \emph{randomly} maps bids to virtual bids (according to some function that depends on the distributions from which bidders' values are drawn), then \emph{separately} allocates each item to the highest virtual bidder.\footnote{In fact, the allocation rule of~\cite{CaiDW12} has even stronger structure in that each item is independently allocated to the bidder whose virtual value for that item is the highest, and moreover the random mapping defining virtual values for each item simply irons a \emph{total ordering} of all bidder types that depends on the underlying distribution.} It is shown in~\cite{AlaeiFHHM12} that when there are many copies of the same customizable item and a matroid constraint on which bidders can simultaneously receive an item (i.e. $\mathcal{F}$ only ensures that at most $k$ items are awarded, subject to a matroid constraint on the served bidders), that the optimal auction randomly maps bids to virtual bids (according to some function that depends on the distributions from which bidders' values are drawn), then allocates the items to maximize virtual surplus (and customizes them after). We emphasize that both results, while quite strong for their corresponding settings, are extemely limited in the settings where they can be applied. In particular, neither says anything about the simple setting of selling houses to unit-demand bidders (i.e. $\mathcal{F}$ ensures that each house is awarded at most once and each bidder receives at most one house: Example~1, Section~\ref{sec:intro}). Selling houses to unit-demand bidders is on the easy side of the settings considered in this paper, as we provide a solution in multi-dimensional settings with arbitrary feasibility constraints. We do not even assume that $\mathcal{F}$ is downward-closed.

For correlated bidders, the series of results by Cremer and McLean~\cite{CM85,CM88} and McAfee and Reny~\cite{MR92} solve for arbitrary feasibility constraints subject to a non-degeneracy condition on the bidder correlation (that is not met when bidders are independent). Under this assumption, they show that the optimal auction extracts full surplus (i.e. has expected revenue equal to expected welfare) and simply uses the VCG allocation rule (the prices charged are not the VCG prices, but a specially designed pricing menu based on the bidder correlation). Our structural results for correlated bidders apply to \emph{arbitrary} feasibility constraints as well as \emph{arbitrary} bidder correlation, removing the non-degeneracy assumption. Of course, the expected revenue extracted by our mechanisms cannot possibly always be as high as the expected maximum social welfare (as it happens in Cremer-McLean and McAfee-Reny) as they also apply to independent bidders, but our characterization is still quite simple: the optimal auction randomly maps pairs of actual bids and possible alternative bids to second-order bids. Then, the second-order bids are combined (based on the underlying bidder correlation) to form virtual bids, and the virtual-welfare-maximizing allocation is chosen.
\subsubsection{Algorithmic Results}

The computer science community has contributed computationally efficient solutions to special cases of the MDMDP in recent years. Many are constant factor approximations~\cite{Alaei11,BhattacharyaGGM10,ChawlaHK07,ChawlaHMS10,KleinbergW12}. These results cover settings where the bidders are unit-demand (or capacitated-additive) and the seller has matroid or matroid-intersection constraints on which bidders can simultaneously receive which items. All these settings are special cases of the MDMDP framework solved in this paper.\footnote{Again, in some of these results~\cite{Alaei11,BhattacharyaGGM10} bidders may also have budget constraints, which can be easily incorporated to the MDMDP framework without any loss, as is shown in Appendix~\ref{app:budgets}, and some replace BIC with IC~\cite{Alaei11,ChawlaHK07,ChawlaHMS10,KleinbergW12}.} In even more restricted cases near-optimal solutions have already been provided. Tools are developed in~\cite{CaiD11,CaiH12,DW12} that yield solutions for simple cases with one or few bidders. Cases with many asymmetric independent bidders are considered in~\cite{CaiDW12} and~\cite{AlaeiFHHM12}. As discussed above, in~\cite{CaiDW12}, the case where $\mathcal{F}$ ensures that each item is awarded at most once is solved. In~\cite{AlaeiFHHM12}, the case where $\mathcal{F}$ ensures that at most $k$ items are awarded, subject to a matroid constraint on the served bidders is solved. Our computational results push far beyond existing results, providing a computationally efficient solution in multi-dimensional settings with arbitrary feasibility constraints.

\subsubsection{Additional Discussion on Multi-Dimensionality}
In light of similarities between our results and those of~\cite{AlaeiFHHM12,CaiDW12}, we provide a short discussion to properly compare the settings where each result applies. We begin with what makes a setting single- or multi-dimensional. In a setting with multiple items, two aspects come into play. First, there are bidder preferences. Bidder preferences are single-dimensional if bidders cannot tell the difference between different items. Specifically, each bidder can be completely described by a single value, $v$, and their value for a feasible allocation that awards them $k$ items is $k\cdot v$. Next, there are feasibility constraints. Feasibility constraints are single-dimensional if they only enforce which bidders can simultaneously receive how many items, but \emph{not} which items they receive. Formally, this means $\mathcal{F}$ is such that if an allocation $\{(i_1,j_1),\ldots,(i_k,j_k)\}$ is feasible, so is $\{(i_1,j'_1),\ldots,(i_k,j'_k)\}$ for all (not necessarily distinct) items $j'_1,\ldots,j'_k$. 

The results of~\cite{AlaeiFHHM12,CaiDW12} solve the MDMDP for very special types of feasibility constraints. Specifically,~\cite{CaiDW12} only covers the following setting: the seller has a single copy of each of $n$ heterogenous items. Bidders have multi-dimensional preferences and are additive with no demand constraints. In~\cite{AlaeiFHHM12}, the feasibililty constraints are single-dimensional and a matroid, and bidders have multi-dimensional preferences and are unit-demand. In this paper, the feasibility constraints may be arbitrary and multi-dimensional. Below are some specific examples to emphasize the differences. We again remark that the examples below are on the ``easy'' side of settings covered in this paper, as our results allow $\mathcal{F}$ to be arbitrary.
\begin{enumerate}
\item There is one copy of a single painting for sale. Then both the feasibility constraints and bidder preferences are single-dimensional. This is solved by Myerson~\cite{Myerson81}, Alaei et. al.~\cite{AlaeiFHHM12}, Cai et al.~\cite{CaiDW12}, and this paper.
\item There is one copy of each of several (different) paintings for sale, bidders are additive with no demand constraints. Then both the feasibility constraints and bidder preferences are multi-dimensional. This is solved by Cai et al.~\cite{CaiDW12} and this paper.
\item There are several copies of the same painting for sale, a matroid constraint on which bidders may simultaneously receive the painting, and bidders are unit-demand. Then both the feasibility constraints and bidder preferences are single-dimensional. This is solved by Myerson~\cite{Myerson81}, Alaei et al.~\cite{AlaeiFHHM12}, and this paper.
\item There are several copies of the same car for sale, and a matroid constraint on which bidders may simultaneously receive a car. The seller can freely customize each car in several different ways, and bidders are unit-demand. Then the feasibility constraints are single-dimensional, and bidder preferences are multi-dimensional.\footnote{To see why the feasibility constraints of this setting are single-dimensional, introduce an item for each possible customization of the car. Then, if an allocation of customized cars to bidders is feasible, any re-customization of these cars is still feasible.} This is solved by Alaei et al.~\cite{AlaeiFHHM12} and this paper.
\item There are several houses for sale to unit-demand bidders (Example~1, Section~\ref{sec:intro}). Then both the feasibility constraints and bidder preferences are multi-dimensional. This is previously unsolved, and solved in this paper.
\item There are appointment slots with several doctors available (Example~2, Section~\ref{sec:intro}). Then both the feasibility constraints and bidder preferences are multi-dimensional. This is previously unsolved, and solved in this paper.
\end{enumerate}

In view of the above discussion, the present paper is the first to obtain solutions in a multi-dimensional setting where both preferences and feasibility constraints are truly multi-dimensional. The settings considered by Myerson~\cite{Myerson81} are single-dimensional, and the settings considered by Alaei et al.~\cite{AlaeiFHHM12} have single-dimensional feasibility constraints. While the settings considered by Cai et al.~\cite{CaiDW12} do technically have multi-dimensional feasibility constraints, their results are enabled by a simple reduction to a setting with single-dimensional feasibility constraints (namely, allocating each item independently of the others). Our results are the first of their kind that apply in a truly multi-dimensional setting, without single-dimensional feasibility constraints or reductions to such settings.

\subsection{Our Approach and Intermediate results}
Since receiving attention from computer scientists, several special cases of the MDMDP have been solved computationally efficiently by linear programming~\cite{CaiH12,DW12}. Simply put, these algorithms explictly store a variable for every possible bidder profile denoting the probability that bidder $i$ receives item $j$ on that profile, and write a linear program to maximize expected revenue subject to feasibility and BIC constraints. Unfortunately, the number of variables required for such a program is exponential in the number of bidders, making such an explicit description prohibitive. More recent solutions have used the \emph{reduced form} of an auction~\cite{AlaeiFHHM12,CaiDW12} to sidestep this curse. The reduced form of an auction was first studied in~\cite{Border91,MR84,Matthews84} and contains, for every bidder $i$, for every type $A$ of bidder $i$, and every item $j$, the probability that bidder $i$ receives item $j$ when truthfully reporting type $A$ over the randomness of the auction \emph{and the randomness in the other bidders' types, assuming they report truthfully}. Indeed, the reduced form auction contains all the necessary information to verify that an auction is BIC when bidders are independent, although verifying its feasibility {(i.e. verifying whether a feasible mechanism exists matching these probabilities)} appears to be difficult. Despite this difficulty, computationally efficient separation oracles were discovered for independent bidders and a single item~\cite{AlaeiFHHM12,CaiDW12}. {The techniques of~\cite{AlaeiFHHM12} also accomodate many copies of the same item and a matroid constraint on which bidders may simultaneously be served}. In this paper, we step far beyond both existing results and consider reduced forms in settings with arbitrary feasibility constraints. Surprisingly, we are able to provide a simple proof of a strong characterization result: for arbitrary feasibility constraints, every feasible reduced form can be implemented by a distribution over virtual VCG allocation rules. Our proof is in Section~\ref{sec:independent} and follows the spirit of~\cite{Border91,CaiDW12}: we examine the region of feasible reduced forms {(we show it is always a polytope)} and identify special structure in the extreme points of this region. 

In Section~\ref{sec:algorithms} we provide a separation oracle {for feasible reduced forms}, as well as a decomposition algorithm to explicitly write any feasible reduced form as a distribution over virtual VCG allocation rules in all settings, given only black box access to an implementation of VCG with respect to $\mathcal{F}$. In order to make these algorithms exact in all settings, we must use time polynomial in $|\mathcal{D}|$, making them {practically} unusable. In Section~\ref{sec:approximations}, we show how to {$\epsilon$-implement} both the separation oracle and decomposition algorithm in time polynomial in $\sum_{i=1}^{m} |\mathcal{D}_i|$ {and $1/\epsilon$} with high probability. {By {$\epsilon$-implementing} a separation oracle and decomposition algorithm for polytope $P$, we mean computing a polytope $P'$ such that every point in $P$ is within $\epsilon/\poly(n\sum_{i=1}^{m} |\mathcal{D}_i|)$ (in $\ell_\infty$ distance) of a point in $P'$ and vice versa, and exactly implementing a separation oracle and decomposition algorithm for $P'$. We show that this is sufficient for computationally efficiently deciding whether a reduced form that is $\epsilon$-far (in $\ell_\infty$) from the boundary of $P$ lies inside $P$, as well as for computing a distribution over virtual VCG allocation rules that is within $\epsilon/\poly(n\sum_{i=1}^{m} |\mathcal{D}_i|)$ (in $\ell_\infty$ distance) of any given feasible reduced form that is $\epsilon$-far from the boundary of $P$.}

In Section~\ref{sec:revenue}, we show how to combine the linear programs from~\cite{CaiDW12,DW12} with our algorithms for reduced forms to obtain an {FPRAS} for MDMDP using only black box access to an implementation of VCG for $\mathcal{F}$. In generic cases, the runtime is polynomial in the number of items and $\sum_{i=1}^{m} |\mathcal{D}_i|$ (but \emph{not} $|\mathcal{D}|$). In many settings {(e.g. when there is correlation among item values, or when the value distributions have sparse supports)} this is {the natural description complexity of the problem}, and several recent algorithms~\cite{Alaei11,AlaeiFHHM12,BhattacharyaGGM10,CaiDW12,DobzinskiFK11} have the same computational complexity, {namely polynomial in the number of bidders, {the number of items and the cardinality of the support of each bidder's value distribution}. Additionally, by using results of~\cite{DW12} we can reduce the runtime to polynomial in only the number of items and number of bidders in item-symmetric settings, as well as extend our solution to distributions with infinite support. {Our mechanisms can be made interim or ex-post individually rational without any difference in revenue. We are also able to naturally accommodate hard budget constraints in our solutions.} The simple modification that is necessary is shown in Appendix~\ref{app:budgets}.

Finally, in Section~\ref{sec:correlated}, we prove a characterization result for correlated bidders. To do this, we introduce the notion of a second-order reduced form, and show that every second-order reduced form can be implemented as a distribution over second-order VCG allocation rules. With this modification, all the related techniques of Section~\ref{sec:algorithms} also apply to correlated bidders. 

\section{Preliminaries and notation}\label{sec:notation}

We denote the number of bidders by $m$, the number of items by $n$, and the type space of bidder $i$ by $T_i$. To ease notation, we sometimes use $A$ ($B$, $C$, etc.) to denote the \emph{type} of a bidder, without emphasizing whether it is a vector or a scalar. The elements of $\times_i T_i$ are called {\em type profiles}, and specify a type for every bidder. We assume type profiles are sampled from a distribution ${\cal D}$ over $\times_i T_i$. We denote by ${\cal D}_i$ the marginal of this distribution on bidder $i$'s type. For independent bidders, we use ${\cal D}_{-i}$ to denote the marginal of ${\cal D}$ over the types of all bidders, except bidder $i$. For correlated bidders, we use $\mathcal{D}_{-i}(\vec{v}_i)$ to denote the conditional distribution over the types of all bidders except for $i$, conditioned on bidder $i$'s type being $\vec{v}_i$. We use $t_i$ for the random variable representing the type of bidder $i$. So when we write $\Pr[t_i = A]$, we mean the probability that bidder $i$'s type is $A$. In Appendix~\ref{sec:input distribution}, we discuss how our algorithms access distribution ${\cal D}$.

We let ${\cal A} = [m]\times[n]$ denote the set of possible \emph{assignments} (i.e. the element $(i,j)$ denotes that bidder $i$ was awarded item $j$). {We call {(distributions over)} subsets of ${{\cal A}}$ {(randomized)} \emph{allocations}, and functions mapping {type} profiles to {(possibly randomized)} allocations \emph{allocation rules}. We call an allocation combined with a price charged to each bidder an \emph{outcome}, and an allocation rule combined with a pricing rule a {(direct revelation)} \emph{mechanism}.} {As discussed in Section~\ref{sec:intro}, we may also have a set system $\mathcal{F}$ on ${\cal A}$ (that is, a subset of $2^{\cal A}$), encoding constraints on what assignments can be made simultaneously by the mechanism. ${\cal F}$ may be incorporating arbitrary demand constraints imposed by each bidder, and  supply constraints imposed by the seller, and will be referred to as our {\em feasibility constraints}. In this case, we restrict all allocation rules to be supported on ${\cal F}$.}

The {\em reduced form} of an allocation rule (also called the \emph{interim allocation rule}) is a vector function $\pi(\cdot)$, specifying values $\pi_{ij}(A)$, for all items $j$, bidders $i$ and types $A \in T_i$. $\pi_{ij}(A)$ is the probability that bidder $i$ receives item $j$ when truthfully reporting type $A$, where the probability is over the randomness of all other bidders' types {(drawn from ${\cal D}_{-i}$ in the case of independent bidders, and ${\cal D}_{-i}(A)$ in the case of correlated bidders)} and the internal randomness of the allocation rule, assuming that the other bidders report truthfully their types. Sometimes, we will want to think of the reduced form as a $n\sum_{i=1}^{m} |T_i|$-dimensional vector, and may write $\vec{\pi}$ to emphasize this view.

Given a reduced form $\pi$, we will be interested in whether the form {is ``feasible'', or can be ``implemented.''}  By this we mean designing a feasible allocation rule $M$ (i.e. one that respects feasibility constraints ${\cal F}$ on every type profile {with probability $1$ over the randomness of the allocation rule}) such that the probability $M_{ij}(A)$ that bidder $i$ receives item $j$ when truthfully reporting type $A$ is exactly $\pi_{ij}(A)$, where the probability is computed with respect to the randomness in the allocation rule and the randomness in the types of the other bidders, assuming that the other bidders report truthfully.  While viewing reduced forms as vectors, we will denote by $F(\mathcal{F},\mathcal{D})$ the set of feasible reduced forms when the feasibility constraints are $\mathcal{F}$ and consumers are sampled from $\mathcal{D}$.

{A bidder is {\em additive} if her value for a bundle of items is the sum of her values for the items in that bundle. If bidders are additive, to specify the preferences of bidder $i$, we can provide a valuation vector $\vec{v}_i$,} with the convention that $v_{ij}$ represents her value for item $j$.  {Even in the presence of arbitrary demand constraints, the \emph{value} of additive bidder $i$ of type $\vec{v}_i$ for a randomized allocation that respects the bidder's demand constraints with probability $1$, and whose expected probability of allocating item $j$ to the bidder is $\pi_{ij}$, is just the bidder's expected value, namely $\sum_j v_{ij} \cdot \pi_{ij}$. The \emph{utility} of bidder $i$ for the same allocation when paying price $p_i$ is just $\sum_j v_{ij} \cdot \pi_{ij} - p_i$.} Such bidders whose value for a distribution of allocations is their expected value for the sampled allocation are called \emph{risk-neutral}. Bidders subtracting price from expected value are called \emph{quasi-linear}. 

Throughout this paper, we denote by $\opt$ the expected revenue of an optimal solution to MDMDP. Also, most of our results for this problem construct a {\em fully polynomial-time randomized approximation scheme}, or FPRAS. This is an algorithm that takes as input two additional parameters $\epsilon, \eta >0$ and outputs a mechanism {(or succinct description thereof)} whose revenue is at least $\opt - \epsilon$, with probability at least $1-\eta$ {(over the coin tosses of the algorithm)}, in time polynomial in $n\sum_i |T_i|,1/\epsilon$, and $\log (1/\eta)$. 

Finally, some arguments will involve reasoning about the \emph{bit complexity} of a rational number. We say that a rational number has bit complexity $b$ if it can be written with a binary numerator and denominator that each have at most $b$ bits. Also, for completeness, we define in Appendix~\ref{app:prelims} the standard notion of Bayesian Incentive Compatibility (BIC) and Individual Rationality (IR) of mechanisms for independent bidders, and state a well-known property of the Ellipsoid Algorithm for linear programs. 
\section{Characterization of Feasible Reduced Forms}\label{sec:independent}
In this section, we provide our characterization result, showing that every feasible reduced form can be implemented as a distribution over virtual VCG allocation rules. {For space considerations, all proofs of this section are in Appendix~\ref{app:characterization}}. {In the following definition, $VCG_{\mathcal{F}}$ denotes the allocation rule of VCG with feasibility constraints $\mathcal{F}$. That is, on input $\vec{v} =(\vec{v}_1,\ldots,\vec{v}_m)$, $VCG_{\mathcal{F}}$ outputs the allocation that VCG selects when the reported types are $\vec{v}$.

\begin{definition} {A \textbf{virtual VCG} allocation rule is defined by a collection of weight functions, $f_{i}:T_i \rightarrow \mathbb{R}^n$. $f_{i}$ maps a type of bidder $i$ to a \emph{virtual type} of bidder $i$. On any type profile $\vec{v}$, the virtual VCG allocation rule with functions $\{f_i\}_{i\in[m]}$ runs $VCG_{\mathcal{F}}$ on input $(f_1(\vec{v}_1),\ldots,f_m(\vec{v}_m))$.}\footnote{If there are multiple VCG allocations, break ties arbitrarily, {but consistently}. A consistent lexicographic tie-breaking rule is discussed in Section~\ref{sec:tiebreaking}. {For concreteness, the reader can use this rule for all results of this section.}} {$VVCG_{\mathcal{F}}(\{f_i\}_{i\in[m]})$ denotes the virtual VCG allocation rule with feasibility constraints $\mathcal{F}$ and weight functions $\{f_i\}_{i\in[m]}$.} \end{definition}

In other words, a virtual VCG allocation rule is simply a VCG allocation rule, but maximizing virtual welfare instead of true welfare. {It will be convenient to introduce the following notation, viewing the weight functions as a (scaled) $n\sum_{i=1}^{m} |T_i|$-dimensional vector. Below, $f_{ij}$ denotes the $j^{th}$ component of $f_i$.

\begin{definition}\label{def:vvcg} Let $\vec{w} \in \mathbb{R}^{n\sum_{i=1}^{m} |T_i|}$. Define $f_i$ so that $f_{ij}(A) = \frac{w_{ij}(A)}{\Pr[t_i = A]}$. Then $VVCG_{\mathcal{F}}(\vec{w})$ is the virtual VCG allocation rule $VVCG_{\mathcal{F}}(\{f_i\}_{i\in[m]})$.
\end{definition}

It is easy to see that every virtual VCG allocation rule can be defined using the notation of Definition~\ref{def:vvcg} by simply setting $w_{ij}(A) = f_{ij}(A)\cdot \Pr[t_i = A]$. We scale the weights this way only for notational convenience (which first becomes useful in Lemma~\ref{lem:vvcg}). We say that a virtual VCG allocation rule is simple iff, {for all $\vec{v}_1,\ldots,\vec{v}_m$, $VCG_{\mathcal{F}}(f_1(\vec{v}_1),\ldots,f_m(\vec{v}_m))$ has a unique max-weight allocation.}} We now state the main theorem of this section, which completely characterizes all feasible reduced forms.

\begin{theorem}\label{thm:characterization} Let $\mathcal{F}$ be any set system of feasibility constraints, and $\mathcal{D}$ be any (possibly correlated) distribution over bidder types with finite support. Then every feasible reduced form (with respect to $\mathcal{F}$ and $\mathcal{D}$) can be implemented as a distribution over at most $n\sum_{i=1}^{m}|T_{i}|+1$ simple virtual VCG allocation rules. 
\end{theorem}

{Before outlining the proof, we provide a brief example illustrating the content of Theorem~\ref{thm:characterization}. We are \emph{not} claiming that every feasible allocation rule can be implemented as a distribution over virtual VCG allocation rules. This is not true. What we are claiming is that every feasible allocation rule has the same reduced form as some distribution over virtual VCG allocation rules. Consider a scenario with a single item and two bidders each with two types, $A$ and $B$ that are sampled independently and uniformly at random. If $M$ is the allocation rule that awards bidder $1$ the item when the types match, and bidder $2$ the item when they don't, then $M$ cannot be implemented as a distribution over simple virtual VCG allocation rules. Because bidder $1$ gets the item when both types match, we must always have $w_{11}(A) > w_{21}(A)$ and $w_{11}(B) > w_{21}(B)$. Similarly, because bidder $2$ gets the item when the types don't match we must have $w_{21}(A) > w_{11}(B)$ and $w_{21}(B) > w_{11}(A)$. Clearly, no weights can simultaneously satisfy all four inequalities. However, there is a distribution over simple virtual VCG allocation rules with the same reduced form.\footnote{{Specifically, the reduced form of $M$ is $\frac{1}{2}\cdot \vec{1}$. If we define $w^{(1)}_{11}(A) = w^{(1)}_{11}(B) = 1$, $w^{(1)}_{21}(A) = w^{(1)}_{21}(B) = 0$, and $w^{(2)}_{11}(A) = w^{(2)}_{11}(B) = 0$, $w^{(2)}_{21}(A) = w^{(2)}_{21}(B) = 1$, then the allocation rule that chooses uniformly at random between $VVCG_{\mathcal{F}}(\vec{w}^{(1)})$ and $VVCG_{\mathcal{F}}(\vec{w}^{(2)})$ also has reduced form $\frac{1}{2}\cdot \vec{1}$.}} The proof of Theorem~\ref{thm:characterization} begins with a simple observation and proposition, whose proofs are in Appendix~\ref{app:characterization}.}
\begin{observation}\label{obs:deterministic}
An allocation rule is feasible if and only if it is a distribution over feasible deterministic allocation rules.
\end{observation}

\begin{proposition}\label{prop:convex polytope}
If $|\mathcal{D}|$ is finite, $F(\mathcal{F},\mathcal{D})$ is a convex polytope.
\end{proposition}

{Now that we know that $F(\mathcal{F},\mathcal{D})$ is a convex polytope, we want to look at the extreme points by examining, for any $\vec{w}$, the allocation rule of $F(\mathcal{F},\mathcal{D})$ whose reduced form maximizes $\vec{\pi} \cdot \vec{w}$. Lemma~\ref{lem:vvcg} and Proposition~\ref{prop:VCG} characterize the extreme points of $F(\mathcal{F},\mathcal{D})$, which allows us to prove Theorem~\ref{thm:characterization}. All three proofs are simple, and provided in Appendix~\ref{app:characterization}.}
\begin{lemma}\label{lem:vvcg} Let $\vec{\pi}$ be the reduced form of $VVCG_{\mathcal{F}}(\vec{w})$ {(with an arbitrary tie-breaking rule)} when bidders are sampled from $\mathcal{D}$. Then, for all $\vec{\pi}' \in F(\mathcal{F},\mathcal{D})$, $\vec{\pi} \cdot \vec{w} \geq \vec{\pi}' \cdot \vec{w}.$
\end{lemma}

\begin{proposition}\label{prop:VCG} Every corner of $F(\mathcal{F},\mathcal{D})$ can be implemented by a simple virtual VCG allocation rule, and the reduced form of any simple virtual VCG allocation rule is a corner of $F(\mathcal{F},\mathcal{D})$.
\end{proposition}

We conclude this section by providing necessary and sufficient conditions for feasibility of a reduced form. The proof is again simple and given in Appendix~\ref{app:characterization}.} For the following statement, for any weight vector $\vec{w} \in \mathbb{R}^{n\sum_{i=1}^{m} |T_i|}$, $W_{\mathcal{F}}(\vec{w})$ denotes the total expected weight of items awarded by $VVCG_{\mathcal{F}}(\vec{w})$ (where we assume that the weight of giving item $j$ to bidder $i$ of type $A$ is $f_{ij}(A) = w_{ij}(A)/\Pr[t_i = A]$).  {The proof of Lemma~\ref{lem:vvcg} implies that the tie-breaking rule used in $VVCG_{\cal F}(\vec{w})$ does not affect the value of  $W_{\mathcal{F}}(\vec{w})$, and that no feasible allocation rule can possibly exceed $W_{\mathcal{F}}(\vec{w})$}. The content of the next corollary is that this condition is also sufficient.

\begin{corollary}\label{cor:independent}
A reduced form $\vec{\pi}$ is feasible (with respect to $\mathcal{F}$ and $\mathcal{D}$) if and only if, for all $\vec{w} \in [-1,1]^{n\sum_{i=1}^{m} |T_i|}$, $\vec{\pi} \cdot \vec{w} \leq W_{\mathcal{F}}(\vec{w}).$
\end{corollary}

\subsection{Tie-breaking}\label{sec:tiebreaking}
{Here we discuss tie-breaking. This is important in later sections because we will want to argue that any virtual VCG allocation rule we use is simple. Because we only have black-box access to $A_{\mathcal{F}}$, we do not necessarily have any control over the tie-breaking rule used, which could be problematic. Instead, we would like to enforce a particularly simple tie-breaking rule by changing $\vec{w}$ to $\vec{w}'$ such that $VVCG_{\mathcal{F}}(\vec{w}')$ also maximizes $\vec{\pi} \cdot \vec{w}$ over all reduced forms in $F(\mathcal{F},\mathcal{D})$, and $VVCG_{\mathcal{F}}(\vec{w}')$ is simple. Additionally, we would like the bit complexity of coordinates of $\vec{w}'$ to be polynomial in the bit complexity of coordinates of $\vec{w}$. Lemma~\ref{lem:tiebreaking} in Appendix~\ref{app:characterization} states formally that a simple lexicographic tie-breaking rule can be implemented in the desired manner. From now on, whenever we use the term $VVCG_{\mathcal{F}}(\vec{w})$, we will implicitly assume that this tie-breaking rule has been applied. Sometimes we will explicitly state so, if we want to get our hands on $\vec{w}'$. }

\section{Algorithms for Reduced Forms}\label{sec:algorithms}
{The characterization result of Section~\ref{sec:independent} hinges on the realization that $F(\mathcal{F},\mathcal{D})$ is a polytope whose corners can be implemented by especially simple allocation rules, namely simple virtual VCG allocation rules. To compute the reduced form of an optimal mechanism, we would like to additionally optimize a linear objective (expected revenue) over $F(\mathcal{F},\mathcal{D})$, so we need a separation oracle for this polytope. Additionally, once we have found the revenue-optimal reduced form in $F(\mathcal{F},\mathcal{D})$, we need some way of implementing it. As we know that every corner of $F(\mathcal{F},\mathcal{D})$ can be implemented by an especially simple allocation rule, we would like a way to decompose a given feasible reduced form into an explicit convex combination of corners (which then corresponds to a distribution over simple virtual VCG allocation rules). In this section, we provide both algorithms.} For now, we will not worry about the running time of our algorithms, but just provide a generic framework that applies to all settings. In Section~\ref{sec:approximations} we will describe how to approximately implement these algorithms efficiently with high probability obtaining an {FPRAS} with only black box access to an implementation of the VCG allocation rule. {The same algorithms with the obvious modifications also apply to ``second-order reduced forms'', using the techniques of Section~\ref{sec:correlated}.}

\subsection{Separation Oracle}\label{sec:separation}

We know from Corollary~\ref{cor:independent} that if a reduced form $\vec{\pi}$ is infeasible, then there is some weight vector $\vec{w} \in [-1,1]^{n\sum_{i=1}^{m} |T_i|}$ such that $\vec{\pi} \cdot \vec{w} > W_{\mathcal{F}}(\vec{w}).$ Finding such a weight vector explicitly gives us a hyperplane separating $\vec{\pi}$ from $F(\mathcal{F},\mathcal{D})$, provided we can also compute $W_{\mathcal{F}}(\vec{w})$. So consider the function:

$$g_{\vec{\pi}}(\vec{w}) = W_{\mathcal{F}}(\vec{w}) - \vec{\pi} \cdot \vec{w}.$$

We know that $\vec{\pi}$ is feasible if and only if $g_{\vec{\pi}}(\vec{w}) \geq 0$ {for all $\vec{w}\in [-1,1]^{n\sum_{i=1}^{m} |T_i|}$. So the goal of our separation oracle {$SO$} is to minimize $g_{\vec{\pi}}(\vec{w})$ over the hypercube, and check if the minimum is negative. If negative, the reduced form is infeasible, and the minimizer bears witness. Otherwise, the reduced form is feasible. To write a linear program to minimize $g_{\vec{\pi}}(\vec{w})$, recall that $W_{\mathcal{F}}(\vec{w}) = \max_{\vec{x} \in F(\mathcal{F},\mathcal{D})} \{\vec{x} \cdot \vec{w}\}$, so $g_{\vec{\pi}}(\vec{w})$ is a piece-wise linear function. Using standard techniques, we could add a variable, $t$, for $W_{\mathcal{F}}(\vec{w})$, add constraints to guarantee that $t \geq \vec{x} \cdot \vec{w}$ for all $\vec{x} \in F(\mathcal{F},\mathcal{D})$, and minimize $t - \vec{\pi} \cdot \vec{w}$. As this is a burdensome number of constraints, we will use an internal separation oracle $\widehat{SO}$, whose job is simply to verify that $t \geq \vec{x} \cdot \vec{w}$ for all $\vec{x} \in F(\mathcal{F},\mathcal{D})$, and output a violating hyperplane otherwise.

\notshow{Since $F(\mathcal{F},\mathcal{D})$ is a convex polytope (Proposition~\ref{prop:convex polytope}), by the Fundamental Theorem of Linear Programming, we know that a maximum of any linear function over this region occurs at a corner. So let $C(\mathcal{F},\mathcal{D})$ denote the corners of $F(\mathcal{F},\mathcal{D})$. Then we can rewrite $W_{\mathcal{F}}(\vec{w})$ as $W_{\mathcal{F}}(\vec{w}) = \max_{\vec{x}\in C(\mathcal{F},\mathcal{D})} \{\vec{x} \cdot \vec{w}\}$. So the following linear program minimizes $g_{\vec{\pi}}$:

\begin{figure}[ht]
\colorbox{MyGray}{
\begin{minipage}{\textwidth} {
\noindent\textbf{Variables:}
\begin{itemize}
\item $t$, denoting the value of $W_{\mathcal{F}}(\vec{w})$.
\item $w_{ij}(A)$ for all bidders $i$, items $j$, and types $A \in T_i$. 
\end{itemize}
\textbf{Constraints:}
\begin{itemize}
\item $-1 \leq w_{ij}(A) \leq 1$ for all bidders $i$, items $j$, and types $A \in T_i$, guaranteeing that the weights lie in $[-1,1]^{n\sum_{i=1}^{m} |T_i|}$.
\item $\vec{x} \cdot \vec{w} \leq t$, for all {$\vec{x} \in C(\mathcal{F},\mathcal{D})$}, guaranteeing that $t \geq W_{\mathcal{F}}(\vec{w})$.
\end{itemize}
\textbf{Minimizing:}
\begin{itemize}
\item $t - \vec{\pi} \cdot \vec{w}$, this is $g_{\vec{\pi}}(\vec{w})$ provided $t = W_{\mathcal{F}}(\vec{w})$.\\
\end{itemize}}
\end{minipage}}
\end{figure}
}
{\notshow{
The obvious way to solve this linear program requires us to enumerate all $\vec{x} \in C(\mathcal{F},\mathcal{D})$. To avoid enumeration we replace the constraints $\vec{x} \cdot \vec{w} \leq t$ with a separation oracle $\widehat{SO}$.}} To implement $\widehat{SO}$, let {$R_{\mathcal{F}}(\vec{w})$} denote the reduced form of $VVCG_{\mathcal{F}}(\vec{w})$. Then we know that {$R_{\mathcal{F}}(\vec{w}) \cdot \vec{w} \geq \vec{x} \cdot \vec{w}$ for all $\vec{x} \in F(\mathcal{F},\mathcal{D})$}. So if any equation of the form $\vec{x} \cdot \vec{w} \leq t$ is violated, then certainly {$R_{\mathcal{F}}(\vec{w}) \cdot \vec{w} \leq t$} is violated. Therefore, for an input $\vec{w}, t$, we need only check a single constraint of this form. So let {$\widehat{SO}(\vec{w},t)$} output ``yes'' if {$R_{\mathcal{F}}(\vec{w}) \cdot \vec{w} \leq t$}, and output the violated hyperplane $R_{\mathcal{F}}(\vec{w})\cdot \vec{z} - y \leq 0$ otherwise. $\widehat{SO}$ allows us to reformulate a more efficient linear program to minimize $g_{\vec{\pi}}(\vec{w})$. This LP is explicitly shown in Figure~\ref{fig:separation oracle} of Appendix~\ref{app:algorithms}.

So our separation oracle $SO$ to check if $\vec{\pi} \in F(\mathcal{F},\mathcal{D})$ is as follows: run the linear program of Figure~\ref{fig:separation oracle} to minimize $g_{\vec{\pi}}(\vec{w})$. {Let the optimum output by the LP be $t^{*},\vec{w}^{*}$}. If the value of the LP is negative, we know that {$\vec{w}^{*} \cdot \vec{\pi} > t^{*} = W_{\mathcal{F}}(\vec{w}^*)$}, and we have our violated hyperplane. Otherwise, the reduced form is feasible, so we output ``yes.''

We conclude this section with a lemma relating the bit complexity of the corners of $F(\mathcal{F},\mathcal{D})$ to the bit complexity of the output of our separation oracle. This is handy for efficiently implementing our algorithms in later sections. The proof is simple, and provided in Appendix~\ref{app:algorithms}. We make use a standard property of the Ellipsoid algorithm (see Theorem~\ref{thm:ellipsoid}).

\begin{lemma}\label{lem:low bit SO}
If all coordinates of each corner of $F(\mathcal{F},\mathcal{D})$ are rational numbers of bit complexity $\ell$, then every coefficient of any hyperplane output by $SO$ is a rational number of bit complexity $\poly(n\sum_{i=1}^{m}|T_{i}|,\ell)$. 
\end{lemma}

\subsection{Decomposition Algorithm via a Corner Oracle}\label{sec:decomposition}
We provide an algorithm for writing a feasible reduced-form as a convex combination of corners of $F(\mathcal{F},\mathcal{D})$, i.e. reduced forms of simple virtual VCG allocation rules. A decomposition algorithm for arbitrary polytopes $P$ is already given in~\cite{CaiDW12}, and the only required ingredients for the algorithm is a separation oracle for $P$, \emph{corner oracle} for $P$, and bound $b$ on the bit complexity of the coefficients of any hyperplane that can possibly be output by the separation oracle. {The goal of this section is to define both oracles and determine $b$ for our setting.} But let us first recall the result of~\cite{CaiDW12}. Before stating the result, let us specify the required functionality of the corner oracle.

The corner oracle for polytope $P$ takes as input $k$ (where $k$ is at most the dimension, in our case $n \sum_i |T_i|$) hyperplanes $H_1,\ldots,H_k$ (whose coefficients are all rational numbers of bit complexity $b$) and has the following behavior: If no hyperplane intersects $P$ in its interior and there is a corner of $P$ that lies in all hyperplanes, then such a corner is output. Otherwise, the behavior may be arbitrary. Below is the theorem from~\cite{CaiDW12}.

\begin{theorem}\label{thm:geometric}(\cite{CaiDW12}) Let $P$ be a $d$-dimensional polytope with corner oracle $CO$ and separation oracle $SO$ such that each coefficient of every hyperplane ever output by $SO$ is a rational number of bit complexity $b$. Then there is an algorithm that decomposes any point $\vec{x} \in P$ into a convex combination of at most $d+1$ corners of $P$. Furthermore, if $\ell$ is the maximum number of bits needed to represent a coordinate of $\vec{x}$, then the runtime is polynomial in $d,b,\ell$ and the runtimes of $SO$ and $CO$ on inputs of bit complexity $\poly(d,b,\ell)$.
\end{theorem}

So all we need to do is define $CO$ and $SO$, and provide a bound on the bit complexity of the hyperplanes output by $SO$. {We've already defined $SO$ and bounded the bit complexity of hyperplanes output by it by $\poly(n\sum_{i=1}^{m} |T_i|,\ell)$, where $\ell$ is the maximum number of bits needed to represent a coordinate in a corner of $F(\mathcal{F},\mathcal{D})$ (see {Lemma~\ref{lem:low bit SO}} of Section~\ref{sec:separation}).} So now we define $CO$ and state its correctness in Theorem~\ref{thm:corneroracle} whose proof is in Appendix~\ref{app:algorithms}. In the last line, CO outputs the weights $\vec{w}'$ as well so that we can actually implement the reduced form that is output.

\begin{algorithm}[ht]
        \caption{Corner Oracle for $F(\mathcal{F},\mathcal{D})$}
    \begin{algorithmic}[1]\label{alg: corner}
        \STATE Input: Hyperplanes $(\vec{w}_1, h_{1}), \ldots,(\vec{w}_a$, $h_a)$, {$a \leq n\sum_{i=1}^{m} |T_i|$.}
        \STATE Set $\vec{w} = \sum_{j=1}^{a} \frac{1}{a} \vec{w}_j$.
        \STATE Use the tie-breaking rule of Section~\ref{sec:tiebreaking} (stated formally in Lemma~\ref{lem:tiebreaking} of Appendix~\ref{app:characterization}) on $\vec{w}$ to obtain $\vec{w}'$.
        \STATE Output the reduced form of $VVCG_{\mathcal{F}}(\vec{w}')$, as well as $\vec{w}'$.
     \end{algorithmic}
\end{algorithm}

\begin{theorem}\label{thm:corneroracle} The Corner Oracle of Algorithm~\ref{alg: corner} correctly outputs a corner of $F(\mathcal{F},\mathcal{D})$ contained in $\cap_{j=1}^{a} H_j$ whenever the hyperplanes {$H_1,\ldots,H_a$ are boundary hyperplanes of $F(\mathcal{F},\mathcal{D})$ and $\cap_{j=1}^{a} H_j$ contains a corner}. {Furthermore, if all coordinates of all $H_j$ are rational numbers of bit complexity $b$, and $\Pr[t_i = A]$ is a rational number of bit complexity $\ell$ for all $i,A \in T_i$, then every coordinate of the weight vector $\vec{w}'$ is a rational number of bit complexity $\poly(n\sum_{i=1}^{m}|T_i|,b,\ell)$.}
\end{theorem}

\section{Efficient Implementation of Algorithms for Reduced Forms}\label{sec:approximations}
In this section, we show how to approximately implement the separation oracle (SO) of Section~\ref{sec:separation} and the corner oracle (CO) of Section~\ref{sec:decomposition} efficiently with high probability, {thereby obtaining also an approximate decomposition algorithm for $F(\mathcal{F},\mathcal{D})$}. We begin by bounding the runtime of an exact implementation, showing that it is especially good when ${\cal D}$ is a uniform (possibly non-product) distribution of small support. As above, $A_{\mathcal{F}}$  denotes an algorithm that implements the VCG allocation rule with respect to feasibility constraints $\mathcal{F}$, and $rt_{\mathcal{F}}(b)$ denotes the runtime of $A_{\mathcal{F}}$ when each input weight has bit complexity $b$.
\subsection{Exact Implementation}\label{sec:exact}

The only tricky step in implementing $SO$ and $CO$ is computing $R_{\mathcal{F}}(\vec{w})$ for a given $\vec{w}$. A simple approach is to just enumerate every profile in the support of $\mathcal{D}$ and check if $VVCG_{\mathcal{F}}(\vec{w})$ awards bidder $i$ item $j$. This can be done in time polynomial in {the cardinality $|\mathcal{D}|$ of the support of ${\cal D}$, the bit complexity $\ell$ of the probabilities used by ${\cal D}$ and $rt_{\mathcal{F}}(\poly(b,\ell))$}, where $b$ is the bit complexity of $\vec{w}$'s coordinates. {So, if $b$ is an upper bound on the bit complexity of the coordinates of the weight vectors $\vec{w}$ for which $R_{\mathcal{F}}(\vec{w})$ is computed in an execution of $SO$ ($CO$), then $SO$ ($CO$) can be implemented in time polynomial in $n\sum_i |T_i|,|\mathcal{D}|$, $\ell$, $b$, $c$, and $rt_{\mathcal{F}}(\poly(b,\ell))$, where $c$ is the bit complexity of the numbers in the input of $SO$ ($CO$).} Alone, this result is not very helpful as we can do much more interesting computations in time polynomial in $|\mathcal{D}|$, including exactly solve MDMDP~\cite{DW12}. The interesting corollary is that when $\mathcal{D}$ is a (possibly correlated) uniform distribution over a collection of profiles (possibly with repetition) whose number is polynomial in $n\sum_i |T_i|$, the runtime of all algorithms of Section~\ref{sec:algorithms} becomes polynomial in $n\sum_i |T_i|$, $c$, and $rt_{\mathcal{F}}(\poly (n\sum_i |T_i|, c))$, where $c$ is the bit complexity of the numbers in the input to these algorithms. Corollaries~\ref{cor:uniform} and~\ref{cor:geometric alg for uniform} in Appendix~\ref{app:exact} quantify this statement precisely. It is these corollaries that enable an efficient approximation for arbitrary distributions in the next section.

\subsection{Approximate Implementation}\label{sec:approximate}
Now, we show how to ``approximately implement'' both algorithms in time polynomial in only $\sum_{i=1}^{m} |\mathcal{D}_i|$, where $|{\cal D}_i|$ is the cardinality of the support of ${\cal D}_i$, using the results of Section~\ref{sec:exact}. But we need to use the right notion of approximation. Simply implementing both algorithms approximately, e.g. separating out reduced forms that are not even approximately feasible and decomposing reduced forms that are approximately feasible, might not get us very far, as we could lose the necessary linear algebra to solve LPs. So we use a different notion of approximation. We compute a polytope $P'$ that, with high probability, is a ``good approximation'' to $F(\mathcal{F},\mathcal{D})$ in the sense that instead of optimizing over $F(\mathcal{F},\mathcal{D})$ we can optimize over $P'$ instead. Then we implement both the separation and the decomposition algorithms for $P'$ exactly so that their running time is polynomial in $n$, $\sum_{i=1}^{m}|T_i|$, $c$ and $rt_{\mathcal{F}}(\poly(n\sum_{i=1}^{m}|T_i|, c))$, where $c$ is the number of bits needed to describe a coordinate of the input to these algorithms\notshow{ and $rt_{\mathcal{F}}(\cdot)$ denotes, as above, the running time of the algorithm $A_{\mathcal{F}}$ that implements the VCG allocation rule with respect to feasibility constraints $\mathcal{F}$}.

\smallskip \noindent {\bf Approach:} So how can we compute an approximating polytope? Our starting point is a natural idea: Given an arbitrary distribution $\mathcal{D}$, we can sample profiles $P_1,\ldots,P_k$ from $\mathcal{D}$ independently at random and define a new distribution $\mathcal{D}'$ that samples a profile uniformly at random from $P_1,\ldots,P_k$ (i.e. chooses each $P_i$ with probability $1/k$). Clearly as $k \rightarrow \infty$ the polytope $F(\mathcal{F},\mathcal{D}')$ should approximate $F(\mathcal{F},\mathcal{D})$ better and better. The question is how large $k$ should be taken for a good approximation. If taking $k$ to be polynomial in $n\sum_{i=1}^{m} |T_i|$ suffices, then {Section~\ref{sec:exact} (specfically, Corollaries~\ref{cor:uniform} and \ref{cor:geometric alg for uniform} of Appendix~\ref{app:exact})} also implies that we can implement both the separation and the decomposition algorithms for $F(\mathcal{F},\mathcal{D}')$ in the desired running time. 

However this approach fails, as some types may very well have $\Pr[t_i = A] << \frac{1}{\poly(n\sum_{i=1}^{m} |T_i|)}$. Such types likely wouldn't even appear in the support of $\mathcal{D}'$ if $k$ scales polynomially in $n\sum_{i=1}^{m} |T_i|$. So how would then the proxy polytope $F(\mathcal{F},\mathcal{D}')$ inform us about $F(\mathcal{F},\mathcal{D})$ in the corresponding dimensions? To cope with this, for each bidder $i$ and type $A \in T_i$, we take an additional $k'$ samples from $\mathcal{D}_{-i}$ and set $t_i = A$. $\mathcal{D}'$ still picks uniformly at random from all $k + k'\sum_{i=1}^{m} |T_i|$ profiles. 

Now here is what we can guarantee. In Corollary~\ref{cor:oneway} (stated and proved in Appendix~\ref{app:oneway}), we show that with high probability every $\vec{\pi}$ in $F(\mathcal{F},\mathcal{D})$ has some $\vec{\pi}' \in F(\mathcal{F},\mathcal{D}')$ with $|\vec{\pi}-\vec{\pi}'|_{\infty}$ small. This is done by taking careful concentration and union bounds. In Corollary~\ref{cor:otherway} (stated and proved in Appendix~\ref{app:otherway}), we show the converse: that with high probability every $\vec{\pi}' \in F(\mathcal{F},\mathcal{D}')$ has some $\vec{\pi} \in F(\mathcal{F},\mathcal{D})$ with $|\vec{\pi} - \vec{\pi}'|_{\infty}$ small. This requires a little more care as the elements of $F(\mathcal{F},\mathcal{D}')$ are not fixed a priori (i.e. before taking samples from ${\cal D}$ to define ${\cal D}'$), but depend on the choice of $\mathcal{D}'$, which is precisely the object with respect to which we want to use the probabilistic method. We resolve this apparent circularity by appealing to some properties of the algorithms of Section~\ref{sec:algorithms} {(namely,  bounds on the bit complexity of any output of $SO$ and $CO$)}. Finally, in Theorems~\ref{thm:appxSO} and~\ref{thm:appxdecomp} (stated below and proved in Appendix~\ref{app:final}), we put our results together to prove that our approximations behave as desired while taking $k$ and $k'$ both polynomial in $n\sum_{i=1}^{m} |T_i|$, thereby achieving the desired runtime.

In the following theorems, Algorithm~\ref{alg:preprocess} refers to a pre-processing algorithm (see Appendix~\ref{app:final}) that explicitly chooses $k$ and $k'$, both polynomial in $n\sum_i |T_i|$, so that the polytopes $F(\mathcal{F},\mathcal{D})$ and $F(\mathcal{F},\mathcal{D}')$ are close with high probability. Algorithm~\ref{alg:decomposition} refers to a decomposition algorithm (Appendix~\ref{app:final}) combining the geometric algorithm of~\cite{CaiDW12} with some bookkeeping to decompose any reduced form in $F(\mathcal{F},\mathcal{D}')$ into an explicit distribution over simple virtual VCG allocation rules. Also in Appendix~\ref{app:final} are Corollaries~\ref{cor:appxSO} and~\ref{cor:appxdecomp}, which show that Theorems~\ref{thm:appxSO} and~\ref{thm:appxdecomp} imply that using $\mathcal{D}'$ as a proxy yields an approximate separation oracle and decomposition algorithm for $\mathcal{D}$.

\notshow{
Now, we \mattnote{use Corollaries~\ref{cor:oneway} and~\ref{cor:otherway}} to obtain ``approximate separation and decomposition oracles'' for arbitrary distributions $\mathcal{D}$. By this we mean generating a proxy distribution ${\cal D}'$ for ${\cal D}$, whose support is polynomial in the approximation parameter $1/\epsilon$ and the dimension $n \sum_i |T_i|$, so that for the purposes of approximately optimizing over $F(\mathcal{F},\mathcal{D})$ we can use instead the separation and decomposition oracles for $F(\mathcal{F},\mathcal{D}')$, whose running time scales gracefully. (The application of this idea is exhibited in the next section for computing revenue optimal auctions.) We start with the following algorithm that fixes some parameters. \mattnote{Algorithm~\ref{alg:decomposition} is our approximate decomposition algorithm. Theorem~\ref{thm:appxSO} proves the correctness of the separation oracle with a proxy $\mathcal{D}'$ generated by Algorithm~\ref{alg:preprocess}, and Theorem~\ref{thm:appxdecomp} proves the correctness of the decomposition algorithm with a proxy $\mathcal{D}'$ generated by Algorithm~\ref{alg:preprocess}.}}

\begin{theorem}\label{thm:appxSO} 
Given the choice of $k$, $k'$ and ${\cal D'}$ in Algorithm~\ref{alg:preprocess}, the following are true with probability at least $1-e^{-\Omega(n\sum_{i=1}^{m}|T_i|/\epsilon)}$:

\begin{enumerate}
\item For all $\vec{\pi} \in F(\mathcal{F},\mathcal{D})$, there is a $\vec{\pi}' \in F(\mathcal{F},\mathcal{D}')$ with $|\vec{\pi}-\vec{\pi}'|_\infty \leq \epsilon$.
\item For all $\vec{\pi}' \in F(\mathcal{F},\mathcal{D}')$, there is a $\vec{\pi} \in F(\mathcal{F},\mathcal{D})$ with $|\vec{\pi}-\vec{\pi}'|_\infty \leq \epsilon$.
\end{enumerate}
Moreover the separation oracle of Section~\ref{sec:separation} for feasibility set ${\cal F}$ and distribution ${\cal D}'$ runs in time polynomial in $n$, $\sum_{i=1}^{m} |T_i|$, $1/\epsilon$, $c$, and $rt_{\mathcal{F}}(\poly(n\sum_{i=1}^{m} |T_i|,\log{1/\epsilon},c))$, where $c$ is the bit complexity of the coordinates of its input.
\end{theorem}

\begin{theorem}\label{thm:appxdecomp} Algorithm~\ref{alg:decomposition} has the following property on input $\vec{\pi}' \in  F(\mathcal{F},\mathcal{D}')$ with probability at least $1-e^{-O(n\sum_{i=1}^{m}|T_i|/\epsilon)}$: Let $\vec{\pi}$ denote the reduced form of the output allocation rule when consumers are sampled from $\mathcal{D}$. Then $|\vec{\pi}-\vec{\pi}'|_\infty \leq \epsilon$. Furthermore, the running time of the algorithm is polynomial in $n,\sum_{i=1}^{m} |T_i|, 1/\epsilon,$ $c$, and $rt_{\mathcal{F}}(\poly(n\sum_{i=1}^{m} |T_i|, \log1/\epsilon, c))$, where $c$ is the bit complexity of the coordinates of its input. 
\end{theorem}

\begin{remark}\label{costasremark:FPRAS?}
Observe that, despite the non-canonical way Theorems~\ref{thm:appxSO} and~\ref{thm:appxdecomp} are stated for FPRASs, the dependence of the running time on the approximation error and the failure probability is the typical one. Namely, using the stated results as black box we can simultaneously achieve error probability at most $\eta$ and approximation error at most $\epsilon$ in time polynomial in $\log 1/\eta$, $1/\epsilon$, $n$, $\sum_i |T_i|$, $c$ and $rt_{\mathcal{F}}(\poly(n\sum_{i=1}^{m} |T_i|, \log1/\epsilon, \log \log (1/\eta), c))$, where $c$ is the bit complexity of the coordinates of the input to our algorithms. 
\end{remark}

\section{Revenue-Maximizing Mechanisms}\label{sec:revenue}
In this section we describe how to use the results of Section~\ref{sec:approximations} to obtain computationally efficient nearly-optimal solutions to MDMDP using only black box access to an implementation of the VCG allocation rule. {As our notable contribution to obtain these results is the techniques of Sections~\ref{sec:independent} through~\ref{sec:approximations}, we only state our results here. In Appendix~\ref{app:revenue}, we provide a high-level overview of how to combine our results with the LPs of~\cite{CaiDW12,DW12} to solve the MDMDP, followed by proofs}. In all {theorem statements}, the allocation rule of the mechanism output by our algorithm is a distribution over simple virtual VCG allocation rules. There is no special structure in the pricing rule, it is just the output of a linear program. As usual, we denote by $A_{\mathcal{F}}$ an algorithm that implements the VCG allocation rule with feasibility constraints $\mathcal{F}$, and denote by $rt_{\mathcal{F}}(b)$ the runtime of $A_{\mathcal{F}}$ when each input weight has bit complexity $b$. {We note that the mechanisms output by the following theorems can be made interim or ex-post individually rational without any difference in revenue. We are also able to accommodate bidders with hard budget constraints in our solutions. The proofs presented in Appendix~\ref{app:revenue} provide interim individually rational mechanisms without budget constraints. Ex-post individual rationality and budgets are discussed in Appendix~\ref{app:budgets}.}

\begin{theorem}\label{thm:general} For all $\epsilon, \eta>0$, all $\mathcal{D}$ of finite support in $[0,1]^{nm}$, and all $\mathcal{F}$, given {$\mathcal{D}$ and} black box access to $A_{\mathcal{F}}$ there is an additive FPRAS for MDMDP. In particular, the FPRAS obtains expected revenue $\opt - \epsilon$, {with probability at least $1-\eta$}, in time polynomial in {$\ell$}, $m, n, \max_{i\in[m]}\{|T_i|\},1/\epsilon, \log (1/\eta)$ and {$rt_{\mathcal{F}}(\poly(n\sum_{i=1}^{m} |T_i|, \log1/\epsilon, \log \log (1/\eta), \ell))$, where $\ell$ is an upper bound on the bit complexity of the coordinates of the points in the support of ${\cal D}$, as well as of the probabilities assigned by ${\cal D}_1,\ldots,{\cal D}_m$ to the points in their support. The output mechanism is $\epsilon$-BIC}, its allocation rule is a distribution over simple virtual VCG allocation rules, and it can be implemented in the afore-stated running time.
\end{theorem}

\begin{theorem}\label{thm:itemsym} For all $\epsilon, \eta>0$, all item-symmetric $\mathcal{D}$ of finite support in $[0,1]^{nm}$, and all item-symmetric $\mathcal{F}$,\footnote{Distributions and feasibility constraints are item-symmetric if they are invariant under every item permutation.} given {$\mathcal{D}$ and} black box access to $A_{\mathcal{F}}$, there is an additive FPRAS for MDMDP. The FPRAS obtains expected revenue $\opt - \epsilon$, {with probability at least $1-\eta$}, in time polynomial in ${\ell}, m,n^c,1/\epsilon, \log 1/\eta$ and $rt_{\mathcal{F}}(\poly(n^{c}m, \log1/\epsilon, \log \log (1/\eta), \ell))$, where $c = \max_{i,j}|\mathcal{D}_{ij}|$, where $|\mathcal{D}_{ij}|$ is the cardinality of the support of the marginal of $\mathcal{D}$ on bidder $i$ and item $j$, and $\ell$ is as in the statement of Theorem~\ref{thm:general}. The output mechanism  is $\epsilon$-BIC, its allocation rule is a distribution over simple virtual VCG allocation rules, and it can be implemented in the afore-stated running time.
\end{theorem}

\begin{theorem}\label{thm:bounded} For all $\epsilon, \eta, \delta>0$, all item-symmetric $\mathcal{D}$ supported on $[0,1]^{nm}$ and all item-symmetric $\mathcal{F}$, given {$\mathcal{D}$ and} black box access to $A_{\mathcal{F}}$, there is an additive bi-criterion PRAS algorithm for MDMDP with the following {guarantee}: If $C$ is the maximum number of items that are allowed to be allocated simultaneously by $\mathcal{F}$, the algorithm obtains expected revenue $\opt - {(\sqrt{\epsilon} + \sqrt{\delta}) C}$, 
{with probability $1-\eta$}, in time polynomial in $m, n^{1/\delta}, 1/\epsilon, \log (1/\eta),$ and $rt_{\mathcal{F}}(\poly({n^{1/\delta}m},{\log 1/\epsilon}, \log \log 1/\eta))$. In particular, the runtime does \emph{not} depend on $|\mathcal{D}|$ at all). The output mechanism is $\epsilon$-BIC, and can be implemented in the afore-stated running time.
\end{theorem}
\begin{remark} The assumption that $\mathcal{D}$ is supported in $[0,1]^{mn}$ as opposed to some other bounded set is w.l.o.g., as we could just scale the values down by a multiplicative $v_{\max}$. This would cause the additive approximation error to be $\epsilon v_{\max}$. In addition, the point of the additive error in the revenue of Theorem~\ref{thm:bounded} is \emph{not} to set $\epsilon,\delta$ so small that they cancel out the factor of $C$, but rather to accept the factor of $C$ as lost revenue. For ``reasonable'' distributions, the optimal revenue scales with $C$, so it is natural to expect that the additive loss should scale with $C$ as well.
\end{remark}

\section{Characterization for Correlated Biddders}\label{sec:correlated}
Here we provide the analogue of Theorem~\ref{thm:characterization} for correlated bidders. We begin by observing that, in fact, Theorem~\ref{thm:characterization} already holds for correlated bidders. The reduced form is still well-defined, and nothing about the proofs in Section~\ref{sec:independent} requires independence across bidders. What's wrong is that the information contained in the reduced form is not sufficient for correlated bidders. Indeed, for independent bidders, the information contained in the reduced form is sufficient to verify both feasibility (due to Section~\ref{sec:separation}) and bayesian incentive compatibility. For correlated bidders, while feasibility can still be verified, bayesian incentive compatibility cannot. This is because in order for bidder $i$ to decide whether she wishes to report type $A$ or type $B$ when her true type is $B$, she needs to know the probability of receiving each item if she reports $A$, conditioned on the fact that the remaining bidders are sampled according to the conditional distribution induced by $t_i = B$. This information is simply not contained in the reduced form. To cope with this issue, we first propose an extended definition of reduced form for the case of correlated bidders, and a proper analogue of a virtual VCG allocation rule.

\begin{definition} A \textbf{second-order reduced form} is a vector valued function $\pi( \cdot )$ such that $\pi_{ij}(A,B)$ denotes the probability that bidder $i$ receives item $j$ when reporting type $A$, where the probability is taken over the randomness of the mechanism and the other bidders' types, assuming they are sampled from $\mathcal{D}_{-i}(B)$ and bid truthfully.
\end{definition}

\begin{definition} A \textbf{second-order VCG allocation rule} is defined by a collection of second-order weight functions $w_{ij}: T_i \times T_i \rightarrow \mathbb{R}$. $w_{ij}$ maps a reported type of bidder $i$ and true type of bidder $i$ to a second-order bid for item $j$. On any profile $\vec{v}$, the second-order VCG allocation rule with weights $\vec{w}$ (denoted $SOVCG_{\mathcal{F}}(\vec{w})$) on input $\vec{v}$ selects the max-weight feasible allocation using the weights:
{
$$f_{ij}(\vec{v}) = \sum_{B \in T_i} w_{ij}(\vec{v}_i,B) \Pr[\vec{v}_{-i} \leftarrow \mathcal{D}_{-i}(B)]$$}
\end{definition}

{We say that a second-order VCG allocation rule with weights $\vec{w}$ is simple, if on every profile $\vec{v}$, there is a unique max-weight allocation, where the weight for allocating item $j$ to bidder $i$ is $f_{ij}(\vec{v})$.}
We now quickly observe a connection between second-order VCG allocation rules and virtual VCG allocation rules when bidders are independent, and follow with a statement of the analogue of Theorem~\ref{thm:characterization} for second-order reduced forms.

\begin{observation} If bidders are independent, then for any second-order weight vector $\vec{w}$ and virtual weight vector $\vec{w}'$ with $w'_{ij}(A) = \sum_{B \in T_i} w_{ij}(A,B)$, the allocation rules $VVCG_{\mathcal{F}}(\vec{w}')$ and $SOVCG_{\mathcal{F}}(\vec{w})$ are identical for all $\mathcal{F}$.
\end{observation}

\begin{proof}
When bidders are independent, $\Pr[\vec{v}_{-i} \leftarrow \mathcal{D}_{-i}(B)] = \Pr[\vec{v}_{-i} \leftarrow \mathcal{D}_{-i}]$ for all $B$. In addition, $\Pr[\vec{v}_{-i} \leftarrow \mathcal{D}_{-i}] = \Pr[\vec{v}\leftarrow \mathcal{D}]/\Pr[t_i = \vec{v}_i]$. Therefore, the weight $f_{ij}$ used by $SOVCG_{\mathcal{F}}(\vec{w})$ on bid vector $\vec{v}$ is just:
{
$$f_{ij}(\vec{v}) = \frac{\Pr[\vec{v} \leftarrow \mathcal{D}]}{\Pr[t_i = \vec{v}_i]} \sum_{B \in T_i} w_{ij}(\vec{v}_i,B)$$}

The weight $f'_{ij}$ used by $VVCG_{\mathcal{F}}(\vec{w}')$ is:
\begin{align*}
f'_{ij} &= w'_{ij}(\vec{v}_i)/\Pr[t_i = \vec{v}_i]\\
&= \frac{1}{\Pr[t_i = \vec{v}_i]} \sum_{B \in T_i} w_{ij}(\vec{v}_i,B)\\
&=\frac{f_{ij}(\vec{v})}{\Pr[\vec{v}\leftarrow \mathcal{D}]}
\end{align*}

So the weights used by $VVCG_{\mathcal{F}}(\vec{w}')$ are proportional to the weights used by $SOVCG_{\mathcal{F}}(\vec{w})$ and they will choose the same allocation on every profile.
\end{proof}

\begin{theorem}\label{thm:correlated} Let $\mathcal{F}$ be any set system of feasibility constraints, and $\mathcal{D}$ any arbitrarily correlated distribution over consumer types with finite support. Then every feasible second-order reduced form (with respect to $\mathcal{F}$ and $\mathcal{D}$) can be implemented by a distribution over {at most $\sum_{i=1}^{m}|T_{i}|^{2}+1$} simple second-order VCG allocation rules.
\end{theorem}

The proof of Theorem~\ref{thm:correlated} parallels that of Theorem~\ref{thm:characterization}. We begin by observing that Observation~\ref{obs:deterministic} and Proposition~\ref{prop:convex polytope} also hold in this setting and the proofs are identical. We denote the polytope of feasible second-order reduced forms with respect to $\mathcal{F}$ and $\mathcal{D}$ by $SO(\mathcal{F},\mathcal{D})$. We now characterize the corners of $SO(\mathcal{F},\mathcal{D})$, beginning with an analogue of Lemma~\ref{lem:vvcg}:

\begin{lemma}\label{lem:sovcg}
Let $\vec{\pi}$ be the second-order reduced form of $SOVCG_{\mathcal{F}}(\vec{w})$ with respect to $\mathcal{D}$. Then for all $\vec{\pi}' \in SO(\mathcal{F},\mathcal{D})$:

$$\vec{\pi} \cdot \vec{w} \geq \vec{\pi}' \cdot \vec{w}$$
\end{lemma}

\begin{proof}
Consider any allocation rule $M$ with second-order reduced form {$\vec{\pi}''$} and denote by $M_{ij}(\vec{v})$ the probability that $M$ awards item $j$ to bidder $i$ on profile $\vec{v}$. Then we can expand {$\vec{\pi}'' \cdot \vec{w}$} as:
{
\begin{align*}
\vec{\pi}'' \cdot \vec{w} &= \sum_j \sum_i \sum_{A \in T_i} \sum_{B \in T_i} w_{ij}(A,B) \pi''_{ij}(A,B)\\
&= \sum_j \sum_i \sum_{A \in T_i} \sum_{B \in T_i} w_{ij}(A,B) \sum_{\vec{v}_{-i}} \Pr[\vec{v}_{-i} \leftarrow \mathcal{D}_{-i}(B)]M_{ij}(\vec{v}_{-i};A)\\
&= \sum_{\vec{v}} \sum_i \sum_j M_{ij}(\vec{v})\cdot \sum_{B \in T_i} w_{ij}(\vec{v}_i,B) \Pr[\vec{v}_{-i} \leftarrow \mathcal{D}_{-i}(B)]\\
&= \sum_{\vec{v}} \sum_i \sum_j M_{ij}(\vec{v})\cdot f_{ij}(\vec{v})\\
\end{align*}
}
The second line is derived by simply exanding $\pi''_{ij}(A,B)$. The third line is derived by determining the coefficient for each $M_{ij}(\vec{v})$ in the previous line. {The final line is derived by replacing $\sum_{B \in T_i} w_{ij}(\vec{v}_i,B) \Pr[\vec{v}_{-i} \leftarrow \mathcal{D}_{-i}(B)]$ with $f_{ij}(\vec{v})$.} One should interpret {$f_{ij}(\vec{v})$} to be the weight of awarding item $j$ to bidder $i$ on profile $\vec{v}$. Therefore, the allocation rule whose second-order reduced form maximizes $\vec{\pi}''\cdot \vec{w}$ over all feasible second order reduced forms is simply the one that selects the max-weight allocation on every profile, where the weight of awarding bidder $i$ item $j$ on profile $\vec{v}$ is {$f_{ij}(\vec{v})$}. This is exactly the allocation rule $SOVCG_{\mathcal{F}}(\vec{w})$.
\end{proof}

\begin{proposition}\label{prop:correlatedcorners}
Every corner in $SO(\mathcal{F},\mathcal{D})$ can be implemented by a simple second-order VCG allocation rule, and the reduced form of any second-order VCG allocation rule is a corner in $SO(\mathcal{F},\mathcal{D})$.
\end{proposition}

\begin{proof} The proof is truly identical to that of Proposition~\ref{prop:VCG} after replacing Lemma~\ref{lem:vvcg} with Lemma~\ref{lem:sovcg}.
\end{proof}

\begin{prevproof}{Theorem}{thm:correlated}
Again, the proof is identical to that of Theorem~\ref{thm:characterization} after replacing Proposition~\ref{prop:VCG} with Proposition~\ref{prop:correlatedcorners}.
\end{prevproof}

We conclude this section with a discussion on the content of Theorems~\ref{thm:characterization} and~\ref{thm:correlated}. Again, we are \emph{not} claiming that every allocation rule can be implemented as a distribution over second-order VCG allocation rules. This is again not true, and the same example from Section~\ref{sec:independent} bears witness. Let's take a step back and view virtual and second-order VCG allocation rules as special cases of a more generic type of allocation rule:

\begin{definition}A \textbf{weight-scaling} allocation rule is defined by a dimension, $k$, and a collection of functions $\{V_{ij},W_{ij}\}_{(i,j) \in [m]\times[n]}$. $V_{ij}$ maps a type of bidder $i$ to a $k$-dimensional weight vector ($T_i \rightarrow \mathbb{R}^k$) and $W_{ij}$ maps the remaining types to a $k$-dimensional scaling vector ($\times_{i' \neq i} T_{i'} \rightarrow \mathbb{R}^k$). On any profile $\vec{v}$, the weight-scaling allocation rule with functions $\{V_{ij},W_{ij}\}_{i,j}$ selects the max-weight allocation with weights:

$$f_{ij}(\vec{v}) = V_{ij}(\vec{v}_i) \cdot W_{ij}(\vec{v}_{-i})$$
\end{definition}

In other words, the weight of awarding bidder $i$ item $j$ is the dot product of two vectors, one contributed by bidder $i$'s type, and the other contributed by the rest of the profile. It is not hard to see that every deterministic allocation rule can be implemented by a weight-scaling allocation rule of dimension $\prod_i |T_i|$.~\footnote{Specifically, to implement any deterministic allocation rule $M$, index the possible profiles as $P_1,\ldots,P_k$. If on profile $P_a$, $t_i \neq \vec{v}_i$, set $(V_{ij})_a(\vec{v}_i) = 0$ for all $j$. Similarly, if $t_{-i} \neq \vec{v}_{-i}$, set $(W_{ij})_a(\vec{v}_{-i}) = 0$ for all $j$. If $t_{-i} = \vec{v}_{-i}$ on profile $P_a$, set $(W_{ij})_a(\vec{v}_{-i}) = 1$. If $t_i = \vec{v}_i$, and $M$ awards bidder $i$ item $j$ on profile $P_a$, set $(V_{ij})_a(\vec{v}_i) = 1$. If $t_i = \vec{v}_i$ and $M$ doesn't award bidder $i$ item $j$ on profile $P_a$, set $(V_{ij})_a(\vec{v}_i) = -1$. Then on any profile, we will have $f_{ij}(\vec{v}) = 1$ iff $M$ awards bidder $i$ item $j$ on $\vec{v}$, and $-1$ otherwise.} It is also not hard to imagine that in order to specify arbitrary deterministic mechanisms as a weight-scaling allocation rule, dimension $\prod_i |T_i|$ might be necessary. However, virtual VCG allocation rules are weight-scaling allocation rules of dimension $1$ (and furthermore, $W_{ij}(\vec{v}_{-i}) = 1$ for all $i,\vec{v}_{-i}$), and second-order VCG allocation rules are weight-scaling allocation rules of dimension $\max_i |T_i|$. By restricting ourselves to only care about the reduced form or second-order reduced form, we have drastically simplified the space of allocation rules. Our characterization theorems show that every allocation rule has the same reduced form as a distribution over weight-scaling allocation rule of dimension $1$, and the same second-order reduced form as a distribution over weight-scaling allocation rule of dimension $\max_i |T_i|$. 
\appendix
\section{Details from Preliminaries}\label{app:prelims}
We provide a formal definition {of Bayesian Incentive Compatibility and Individual Rationality of a mechanism for independent bidders} and state a well-known property of the Ellipsoid Algorithm. For completeness, we provide an additional proposition showing a standard trick that can force the Ellipsoid algorithm to always output a corner.

\begin{definition}~\cite{DW12}(BIC/$\epsilon$-BIC Mechanism)\label{def:BIC} A mechanism $M$ is called $\epsilon$-BIC iff the following inequality holds for all bidders $i$ and types $\tau_i,\tau_i' \in T_i$:
$$\mathbb{E}_{t_{-i} \sim {\cal D}_{-i}}\left[U_i(\tau_i,M_i(\tau_i~;~{t}_{-i}))\right] \ge \mathbb{E}_{t_{-i} \sim {\cal D}_{-i}}\left[ U_i(\tau_i,M_i(\tau_i'~;~{t}_{-i})) \right] - \epsilon v_{\max} \cdot \max\left\{1,\sum_j \pi^{M}_{ij}(\tau_i')\right\},$$
where: 
\begin{itemize}
\item $U_i(A,M_i(B~;~{t}_{-i}))$ denotes the utility of bidder $i$ for the outcome of mechanism $M$ if his true type is $A$, he reports $B$ to the mechanism, and the other bidders report $t_{-i}$;
\item $v_{\max}$ is the maximum possible value of any bidder for any item in the support of the value distribution; and

\item $\pi^{M}_{ij}(A)$ is the probability that item $j$ is allocated to bidder $i$ by mechanism $M$ if bidder $i$ reports type $A$ to the mechanism, in expectation over the types of the other bidders, assuming they report truthfully, and the mechanism's internal randomness.
\end{itemize}
In other words, $M$ is $\epsilon$-BIC iff when a bidder $i$ lies by reporting $\tau_i'$ instead of his true type $\tau_i$, she does not expect to gain more than $\epsilon v_{\max}$ times the maximum of $1$ and the expected number of items that $\tau_i'$ receives. A mechanism is called BIC iff it is $0$-BIC.\footnote{Strictly speaking, the definition of BIC in~\cite{DW12} is the same but without taking a max with $1$. We are still correct in applying their results with this definition because any mechanism that is considered $\epsilon$-BIC by~\cite{DW12} is certainly considered $\epsilon$-BIC by this definition. We basically call a mechanism $\epsilon$-BIC if either the definition in~\cite{BeiH11,HartlineKM11,HartlineL10} ($\epsilon v_{\max}$) or~\cite{DW12} ($\epsilon v_{\max} \sum_j \pi_{ij}(\vec{w}_i)$) holds.}
\end{definition}

{We also define individual rationality of BIC/$\epsilon$-BIC mechanisms:

\begin{definition}
A BIC/$\epsilon$-BIC mechanism $M$ is called {\em interim individually rational (interim IR)} iff for all bidders $i$ and types $\tau_i \in T_i$:
$$\mathbb{E}_{t_{-i} \sim {\cal D}_{-i}}\left[U_i(\tau_i,M_i(\tau_i~;~t_{-i}))\right] \ge 0,$$
where $U_i(A,M_i(B~;~{t}_{-i}))$ denotes the utility of bidder $i$ for the outcome of mechanism $M$ if his true type is $A$, he reports $B$ to the mechanism, and the other bidders report $t_{-i}$. The mechanism is called {\em ex-post individually rational (ex-post IR)} iff for all $i$, $\tau_i$ and $t_{-i}$,  $U_i(\tau_i,M_i(  \tau_i~;~t_{-i}  ) \ge 0$ with probability $1$ (over the randomness in the mechanism).
\end{definition}
}

\begin{theorem}\label{thm:ellipsoid}[Ellipsoid Algorithm for Linear Programming]
Let $P$ be a convex polytope in $\mathbb{R}^d$ specified via a separation oracle $SO$, and $\vec{c}\cdot\vec{x}$ be a linear function. Assume that all coordinates of $\vec{a}$ and $b$, for all separation hyperplanes $\vec{a}\cdot\vec{x}\leq b$ possibly output by $SO$, and all coordinates of $\vec{c}$ are rational numbers of bit complexity $\ell$. Then we can run the ellipsoid algorithm to optimize $\vec{c}\cdot\vec{x}$ over $P$, maintaining the following properties:
\begin{enumerate}
\item The algorithm will only query $SO$ on rational points with bit complexity $\poly(d,\ell)$.
\item The ellipsoid algorithm will solve the Linear Program in time polynomial in $d$, $\ell$ and the runtime of $SO$ when the input query is a rational point of bit complexity $\poly(d,\ell)$.
\item The output optimal solution is a corner of $P$.\footnote{This is well-known, but also proved in Proposition~\ref{prop:uniqueoptimal} for completeness.}
\end{enumerate}
\end{theorem}

\begin{proposition}\label{prop:uniqueoptimal}
Let $\vec{a}$ be a $d$-dimensional vector, whose coordinates are rational numbers of bit complexity $\ell_{1}$, $P$ be a $d$-dimensional convex polytope, in which all coordinates of all corners are rational numbers of bit complexity $\ell_{2}$. Then we can transform $\vec{a}$ into a new $d$-dimensional vector $\vec{b}$, whose coordinates are all rational numbers of bit complexity $d(\ell_{1}+1)+(2d^{2}+1)\ell_{2}+1$, such that $\vec{x}^* = \argmax_{\vec{x}\in P} \vec{b}\cdot \vec{x}$ is unique. Furthermore, $\vec{x}^{*}$ is also an element of $\argmax_{\vec{x}\in P} \vec{a}\cdot \vec{x}$.
\end{proposition}

\begin{proof}
Let $a_{i}=p_{i}/q_{i}$, where both $p_{i}$ and $q_{i}$ are integers with at most $\ell_{1}$ bits. Now change the $a_{i}$'s to have the same denominator $Q=\Pi_{i} q_{i}$. So $a_{i}=p'_{i}/Q$, where $p'_{i}=p_{i}\Pi_{j\neq i} q_{j}$. Both $Q$ and $p'_{i}$ have at most $d\ell_{1}$ bits. Let now $b_{i}=(p'_{i}+2^{-(1+\ell_{2}+(2d\ell_{2}+1)\cdot i})/Q$. So $b_{i}$ can be described with $d(\ell_{1}+1)+(2d^{2}+1)\ell_{2}+1$ bits.

Now we will argue that, for any rational vector $\vec{z}\neq \vec{0}$, whose coordinates can be described with at most $2\ell_{2}$ bits, $\vec{b}\cdot\vec{z}\neq 0$. Let $z_{i}=r_{i}/s_{i}$, where  both $r_{i}$ and $s_{i}$ are integers with at most $2\ell_{2}$ bits. Now modify $z_{i}$ to be $r'_{i}/S$, where $S=\Pi_{i} s_{i}$ and $r'_{i}=r_{i}\Pi_{j\neq i}s_{i}$. Both $S$ and $r_{i}'$ have at most $2d\ell_{2}$ bits. Now consider the fractional parts of $Q\cdot S\cdot(\vec{b}\cdot\vec{z})$, which is 
$$\sum_{i=1}^{d}2^{-(1+\ell_{2}+(2d\ell_{2}+1)\cdot i)}\cdot r'_{i}.$$
But this equals to $0$ only when $r_{i}=0$ for all $i$. Thus, if $\vec{z}\neq \vec{0}$, $\vec{b}\cdot\vec{z}\neq 0$.

Next, we argue that, if $\vec{x}$ and $\vec{y}$ are two different vectors, whose coordinates can be described with $\ell_{2}$ bits, $\vec{b}\cdot\vec{x}\neq\vec{b}\cdot\vec{y}$. This is implied by the above argument, since all coordinates of $\vec{x}-\vec{y}$ can be described with at most $2\ell_{2}$ bits. So there is a unique optimal solution to $\max_{\vec{x}\in P} \vec{b}\cdot \vec{x}$. Call that solution $\vec{x}^{*}$.

Now we show that $\vec{x}^{*}$ is also an optimal solution for $\max_{\vec{x}\in P} \vec{a}\cdot \vec{x}$. We only need to argue that if corner $\vec{x}$ is not optimal for $\vec{a}$, it will not be optimal for $\vec{b}$. First, it is not hard to see that for corners $\vec{x}$ and $\vec{y}$, if $\vec{a}\cdot(\vec{x}-\vec{y})\neq 0$, $\vec{a}\cdot(\vec{x}-\vec{y})\geq \frac{1}{2^{2d\ell_{2}}Q}$. Second, for any corner $\vec{x}$, $$|(\vec{b}-\vec{a})\cdot\vec{x}|\leq \sum_{i=1}^{d}\left|\frac{2^{-(1+\ell_{2}+(2d\ell_{2}+1)\cdot i)}}{Q}\right|\cdot 2^{\ell_{2}}<\frac{1}{2^{1+2d\ell_{2}}Q}.$$ So if $\vec{a}\cdot\vec{x}>\vec{a}\cdot\vec{y}$, $\vec{b}\cdot\vec{x}$ is still strictly greater than $\vec{b}\cdot\vec{y}$. Thus, $\vec{x}^{*}$ must be an optimal solution for $\max_{\vec{x}\in P} \vec{a}\cdot \vec{x}$.
\end{proof}

\section{Input Model} \label{sec:input distribution}
We discuss two models for accessing a value distribution $\mathcal{D}$, as well as what modifications are necessary, if any, to our algorithms to work with each model:
\begin{itemize}
\item \textbf{Exact Access:} We are given access to a sampling oracle as well as an oracle that exactly integrates the pdf of the distribution over a specified region.
\item \textbf{Sample-Only Access:} We are given access to a sampling oracle and nothing else.
\end{itemize}
The presentation of the paper focuses on the first model. In this case, we can exactly evaluate the probabilities of events without any special care. If we have sample-only access to the distribution, some care is required. Contained in Appendix A of~\cite{DW12} is a sketch of the modifications necessary for all our results to apply with sample-only access. {The sketch is given for the item-symmetric case, but the same approach will work in the asymmetric case.} Simply put, repeated sampling will yield some distribution $\mathcal{D}'$ that is very close to $\mathcal{D}$ with high probability. If the distributions are close enough, then a solution to the MDMDP for $\mathcal{D}'$ is an approximate solution for $\mathcal{D}$. The error in approximating $\mathcal{D}$ is absorbed into the additive error in both revenue and truthfulness.
\section{Omitted Proofs from Section~\ref{sec:independent}}\label{app:characterization}
This appendix contains the missing proofs from Section~\ref{sec:independent}, and our tie-breaking lemma (Lemma~\ref{lem:tiebreaking}).

\begin{prevproof}{Observation}{obs:deterministic}
For any feasible allocation rule $M$, and any type profile $\vec{v}$, the (possibly randomized) allocation $M(\vec{v})$ is a distribution over feasible deterministic allocations. So let $M(\vec{v})$ sample the deterministic allocation $A_i(\vec{v})$ with probability $p_i(\vec{v})$. Then $M(\vec{v})$ can be implemented by uniformly sampling $x$ from $[0,1]$ and selecting $A_i(\vec{v})$ iff $\sum_{j < i} p_j(\vec{v}) < x \leq \sum_{j \leq i} p_j(\vec{v})$. So for $y \in [0,1]$ let $M^{(y)}$ denote the deterministic allocation rule that on profile $\vec{v}$ implements the deterministic allocation selected by $M(\vec{v})$ when $x = y$, then $M$ is exactly the allocation rule that samples $x$ uniformly at random from $[0,1]$ and implements the deterministic allocation rule $M^{(x)}$. So every feasible allocation rule is a distribution over deterministic allocation rules. The other direction is straight-forward: any distribution over feasible deterministic allocation rules is still feasible.
\end{prevproof}

\begin{prevproof}{Proposition}{prop:convex polytope}
It is clear that there are only finitely many deterministic allocation rules: there are finitely many choices per profile, and finitely many profiles. So consider the set $S$ that contains the reduced form of every deterministic allocation rule that is feasible with respect to $\mathcal{F}$. We claim that $F(\mathcal{F},\mathcal{D})$ is exactly the convex hull of $S$. Consider any feasible reduced form $\vec{\pi}$. Then there is some feasible allocation rule $M$ that implements $\vec{\pi}$. By Observation~\ref{obs:deterministic}, $M$ is a distribution over deterministic allocation rules, sampling $M_i$ with probability $p_i$. Therefore, if $\vec{\pi}_i$ denotes the reduced form of $M_i$, we must have $\vec{\pi} = \sum_{i} p_i \vec{\pi}_i$, so $\vec{\pi}$ is in the convex hull of $S$. Similarly, if a reduced form $\vec{\pi}$ satisfies $\vec{\pi} = \sum_{i} p_i \vec{\pi}_i$, where $\vec{\pi}_i$ is the reduced form of a deterministic allocation rule $M_i$ for all $i$, the allocation rule that selects $M_i$ with probability $p_i$ implements $\vec{\pi}$. So the space of feasible reduced forms is exactly the convex hull of $S$, which is finite, and hence its convex hull is a polytope.
\end{prevproof}

\begin{prevproof}{Lemma}{lem:vvcg}
The proof is straight-forward once we correctly interpret $\vec{\pi} \cdot \vec{w}$. Expanding the dot product and using that $f_{ij}(A) = w_{ij}(A)/ \Pr[t_i = A]$, we see that:
\begin{align*}
\vec{\pi} \cdot \vec{w} &= \sum_i \sum_j \sum_{A \in T_i} \pi_{ij}(A) w_{ij}(A)\\
&= \sum_i \sum_j \sum_{A \in T_i} \pi_{ij}(A) f_{ij}(A) \Pr[t_i = A].
\end{align*}

If the ``weight'' of awarding item $j$ to bidder $i$ when her reported type is $A$ is $f_{ij}(A)$, then the last line is exactly the expected weight of items awarded by an allocation rule whose reduced form is $\vec{\pi}$. The feasible allocation rule that maximizes the expected weight of items awarded simply chooses {a} max-weight feasible allocation on every profile. {This is exactly what $VVCG_{\mathcal{F}}(\{f_i\}_{i\in[m]})$ does, i.e. exactly what $VVCG_{\mathcal{F}}(\vec{w})$ does}. So the reduced form of $VVCG_{\mathcal{F}}(\vec{w})$ exactly maximizes $\vec{x} \cdot \vec{w}$ over all $\vec{x}\in F(\mathcal{F},\mathcal{D})$.
\end{prevproof}

\begin{prevproof}{Proposition}{prop:VCG} We first prove that every corner of $F(\mathcal{F},\mathcal{D})$ can be implemented by a simple virtual VCG allocation rule. From Proposition~\ref{prop:convex polytope}, we know $F(\mathcal{F},\mathcal{D})$ is a convex polytope. So for every corner $\vec{\pi} \in F(\mathcal{F},\mathcal{D})$, there is a weight vector $\vec{w}$, such that $\forall\ \vec{\pi}'\in F(\mathcal{F},\mathcal{D})$ and $\vec{\pi}'\neq\vec{\pi}$,
$$\vec{w}\cdot\vec{\pi}> \vec{w}\cdot\vec{\pi}'.$$

So by Lemma~\ref{lem:vvcg}, we know that $\vec{\pi}$ must be the reduced form of $VVCG_{\mathcal{F}}(\vec{w})$, as $\vec{\pi}$ maximizes $\vec{x} \cdot \vec{w}$ over all $\vec{x}\in F(\mathcal{F},\mathcal{D})$. To see that $VVCG_{\mathcal{F}}(\vec{w})$ is simple, assume for contradiction that there is some profile with multiple max-weight feasible allocations. Let $B$ denote the allocation rule that chooses the exact same allocation as $VVCG_{\mathcal{F}}(\vec{w})$ on every other profile, but chooses a different max-weight feasible allocation on this profile. Let $\vec{\pi}_B$ denote the reduced form of $B$. By the definition of $B$, we still have $\vec{\pi}_B \cdot \vec{w} = \vec{\pi} \cdot \vec{w}$. Yet, we also clearly have $\vec{\pi}_B \neq \vec{\pi}$, as they are reduced forms for allocation rules that are identical on all but one profile, where they differ. This contradicts the choice of $\vec{w}$, so $VVCG_{\mathcal{F}}(\vec{w})$ must be simple.

Now we show that the reduced form of any simple virtual VCG allocation rule is a corner of $F(\mathcal{F},\mathcal{D})$. Let $\vec{\pi}$ be the reduced form of $VVCG_{\mathcal{F}}(\vec{w})$. Then for any other $\vec{\pi}'\in F(\mathcal{F},\mathcal{D})$, we must have $\vec{\pi}\cdot\vec{w}>\vec{\pi}'\cdot\vec{w}$. Otherwise, let $\vec{\pi}'$ denote a feasible reduced form with $\vec{\pi}'\cdot \vec{w} \geq \vec{\pi} \cdot \vec{w}$, $\vec{\pi}' \neq \vec{\pi}$ and let $M'$ implement $\vec{\pi}'$. Then clearly, there is some profile where the allocation chosen by $M'$ differs from that chosen by $VVCG_{\mathcal{F}}(\vec{w})$ and its  weight with respect to $\{f_i\}_{i\in[m]}$ is at least as large as the weight of the allocation chosen by $VVCG_{\mathcal{F}}(\vec{w})$. As $VVCG_{\mathcal{F}}(\vec{w})$ is simple, this is a contradiction. Therefore, $\vec{\pi} \cdot \vec{w} > \vec{\pi}' \cdot \vec{w}$ for all $\vec{\pi}'\neq \vec{\pi} \in F(\mathcal{F},\mathcal{D})$ and $\vec{\pi}$ is a corner.
\end{prevproof}

\begin{prevproof}{Theorem}{thm:characterization}
Since $F(\mathcal{F},\mathcal{D})$ is a convex polytope (Proposition~\ref{prop:convex polytope}), by Carath\'eodory's Theorem, we know that every point in $F(\mathcal{F},\mathcal{D})$ can be written as a convex combination of at most $n\sum_{i=1}^{m} |T_i| + 1$ corners of $F(\mathcal{F},\mathcal{D})$. By Proposition~\ref{prop:VCG}, we know that every corner of $F(\mathcal{F},\mathcal{D})$ can be implemented by a simple virtual VCG allocation rule. Finally, we observe that if the allocation rules $M_i$ implement $\vec{\pi}_i$, then the allocation rule that samples $M_i$ with probability $p_i$ implements $\sum_i p_i \vec{\pi}_i$.
\end{prevproof}

\begin{prevproof}{Corollary}{cor:independent} As $F(\mathcal{F},\mathcal{D})$ is a convex polytope, we know that $\vec{\pi} \in F(\mathcal{F},\mathcal{D})$ if and only if for all $\vec{w} \in [-1,1]^{n\sum_{i=1}^{m} |T_i|}$:

$$\vec{\pi} \cdot \vec{w} \leq \max_{\vec{\pi}' \in F(\mathcal{F},\mathcal{D})} \vec{\pi}'\cdot \vec{w}.$$

By the definition of $VVCG_{\mathcal{F}}(\vec{w})$, the right hand side is exactly $W_{\mathcal{F}}(\vec{w})$.
\end{prevproof}

\begin{lemma}\label{lem:tiebreaking} Let $\vec{w}$ be a {weight vector whose coordinates are rational numbers of bit complexity $\ell_1$, and let $\ell_2$ be such that for all $i,A\in T_{i}$, $\Pr[t_{i}=A]$ is a rational number of bit complexity $\ell_{2}$.} Then {the lexicographic tie-breaking rule can be implemented by} a simple transformation that turns $\vec{w}$ into $\vec{w}'$ such that $VVCG_{\mathcal{F}}(\vec{w}')$ is simple, {$VVCG_{\mathcal{F}}(\vec{w}')$ selects a maximum weight allocation with respect to the scaled weights $\vec{w}$ on every profile}, and each coordinate of $\vec{w}'$ is a rational number of bit complexity {$n\ell_{1}\sum_{i=1}^{m} |T_i|+(n\sum_{i=1}^{m} |T_i|+1)\ell_{2}+mn+n+1$}.
\end{lemma}

\begin{prevproof}{Lemma}{lem:tiebreaking}
Let $\{f_i\}_i$ denote the weight functions used by $VVCG_{\mathcal{F}}(\vec{w})$ (i.e. $f_{ij}(A) = w_{ij}(A)/\Pr[t_i = A]$), and rewrite each value $f_{ij}(A)$ {with a common denominator. Before rewriting, each $f_{ij}(A)$ was a rational number of bit complexity $\ell_1 + \ell_2$. After rewriting, the numerator and denominator of each $f_{ij}(A)$ has at most $b = n{(\sum_{i=1}^{m} |T_i|)}(\ell_1+\ell_2)$ bits.}  Now define new weight functions $\{f'_i\}_i$ such that $f'_{ij}(A)$ is equal to $f_{ij}(A)$ except that $2^{-b-ni-j-1}$ is added to its numerator. In going from $\{f_i\}_i$ to $\{f'_i\}_i$, the numerator of the weight of any allocation goes up by at most $2^{-b-n-1}$, not enough to make an allocation optimal if it was suboptimal. Moreover, notice that  the last $mn$ bits of the numerator are in one-to-one correspondence with possible allocations. So even if there were any ties before, there can't be any ties after the transformation, and ties will be broken lexicographically. It is also obvious that we have only added $nm+n+1$ bits to each numerator of $f_{ij}(A)$, {and therefore $w'_{ij}(A)=f'_{ij}(A)\Pr[t_i = A]$ has bit complexity $b+nm+n+1+\ell_{2}=n\ell_1\sum_{i=1}^{m} |T_i|+(n\sum_{i=1}^{m} |T_i|+1)\ell_{2}+mn+n+1$}.
\end{prevproof}

\section{Omitted Details from Section~\ref{sec:algorithms}}\label{app:algorithms}

\subsection{Separation Oracle}
\text{ }
\begin{figure}[h!]
\colorbox{MyGray}{
\begin{minipage}{\textwidth} {
\noindent\textbf{Variables:}
\begin{itemize}
\item $t$, denoting the value of $W_{\mathcal{F}}(\vec{w})$.
\item $w_{ij}(A)$ for all bidders $i$, items $j$, and types $A \in T_i$. 
\end{itemize}
\textbf{Constraints:}
\begin{itemize}
\item $-1 \leq w_{ij}(A) \leq 1$ for all bidders $i$, items $j$, and types $A \in T_i$, guaranteeing that the weights lie in $[-1,1]^{n\sum_{i=1}^{m} |T_i|}$.
\item {$\widehat{SO}(\vec{w},t) = $} ``yes,'' guaranteeing that $t \geq W_{\mathcal{F}}(\vec{w})$. 
\end{itemize}
\textbf{Minimizing:}
\begin{itemize}
\item $t - \vec{\pi} \cdot \vec{w}$ (this is $g_{\vec{\pi}}(\vec{w})$ provided $t = W_{\mathcal{F}}(\vec{w})$).
\end{itemize}}
\end{minipage}} \caption{A Linear Program to minimize $g_{\vec{\pi}}(\vec{w})$.}\label{fig:separation oracle}
\end{figure}

\begin{prevproof}{Lemma}{lem:low bit SO}
The dimension of the LP (shown in Figure~\ref{fig:separation oracle}) used to run $SO$ is $n\sum_{i=1}^{m} |T_i|$. Every constraint of the linear program that is not part of $\widehat{SO}$ has bit complexity $O(1)$, and the coefficients of every hyperplane output by $\widehat{SO}$ have bit complexity $\ell$ by our hypothesis. (Recall our discussion in Section~\ref{sec:tiebreaking}.) So by the theory of Gaussian elimination, the coordinates of all corners of the LP of Figure~\ref{fig:separation oracle} are rational numbers of bit complexity $\poly(n\sum_{i=1}^{m}|T_{i}|,\ell)$. Now recall the third property of the Ellipsoid algorithm from Theorem~\ref{thm:ellipsoid}.
\end{prevproof}
\subsection{Decomposition Algorithm}
Before giving the proof of Theorem~\ref{thm:corneroracle} we present a useful lemma.
\begin{lemma}\label{lem:choosecorner} Let $P$ be a polytope and $H_1,\ldots,H_i$ be $i$ hyperplanes of the form $\vec{w}_j \cdot \vec{v} = h_j$ such that every point $\vec{x} \in P$ satisfies $\vec{x} \cdot \vec{w}_j \leq h_j$. Then for any $c_1,\ldots,c_i > 0$, any $\vec{a} \in P$ satisfying:

$$\vec{a} \cdot \left(\sum_{j=1}^{i} c_j \vec{w}_j \right) = \sum_{j=1}^i c_j h_j$$

is in $\cap_{j=1}^{i} H_j$.
\end{lemma}
\begin{proof}
Because all $\vec{a} \in P$ satisfy $\vec{a} \cdot \vec{w}_j \leq h_j$ for all $j$ and {$c_{j}>0\ \forall j$, the only way to have $\vec{a} \cdot \left( \sum_j c_j \vec{w}_j \right) = \sum_j c_j h_j$ is to have $\vec{a} \cdot c_j \vec{w}_j = c_j h_j$ for all $j$. Thus, we must have $\vec{a} \cdot \vec{w}_j = h_j$ for all $j$, meaning that $\vec{a} \in \cap_j H_j$}.
\end{proof}

\begin{prevproof}{Theorem}{thm:corneroracle} Let $\vec{\pi}, \vec{w}'$ denote the output of the corner oracle. First, we observe that if $H_1,\ldots,H_a$ intersect inside $F(\mathcal{F},\mathcal{D})$, there is some reduced form $\vec{\pi}'$ satisfying $\vec{\pi}' \cdot \vec{w}_j = h_j$ for all $j$. Therefore, such a reduced form must also satisfy $\vec{\pi}' \cdot \vec{w} = \frac{1}{a} \sum_{j=1}^{a} h_j$. Second, as no feasible reduced form can have $\vec{v} \cdot \vec{w}_j > h_j$, we also get that no feasible reduced form has $\vec{v} \cdot \vec{w} > \frac{1}{a} \sum_{j=1}^{a} h_j$. Putting these two observations together, we see that there exists a $\vec{\pi}'$ with $\vec{\pi}' \cdot \vec{w} = \frac{1}{a} \sum_{j=1}^{a} h_j$, and this is the maximum over all feasible reduced forms. {Therefore, the reduced form $R_{\mathcal{F}}(\vec{w})$ of $VVCG_{\mathcal{F}}(\vec{w})$ necessarily has $R_{\mathcal{F}}(\vec{w}) \cdot \vec{w} = \frac{1}{a} \sum_{j=1}^{a} h_j$.} Lemma~\ref{lem:tiebreaking} tells us that {$\vec{\pi}$ is the reduced form of a simple virtual VCG allocation rule $VVCG_{\mathcal{F}}(\vec{w}')$, which also maximizes $\vec{x} \cdot \vec{w}$ over all feasible reduced forms $\vec{x}$. Therefore, $\vec{\pi}\cdot\vec{w} =\frac{1}{a} \sum_{j=1}^{a} h_j$ and by Lemma~\ref{lem:choosecorner}, $\vec{\pi}$ is in $\cap_{j=1}^{a} H_j$. From Proposition~\ref{prop:VCG}, we know $\vec{\pi}$ is a corner. As each coordinate of $\vec{w}_j$ is a rational number of bit complexity $b$, and $a \leq n\sum_{i=1}^{m} |T_i|$, we see that each coefficient of $\vec{w}$ is a rational number of bit complexity $\poly(\log (n\sum_{i=1}^{m} |T_i|),b)$. Lemma~\ref{lem:tiebreaking} then guarantees that each coefficient of $\vec{w}'$ is a rational number of bit complexity $\poly(n\sum_{i=1}^{m} |T_i|,b,\ell)$.}
\end{prevproof}

\section{Proofs Omitted From Section~\ref{sec:exact}: Exact Implementation}\label{app:exact}

We bound the running time of the algorithms of Section~\ref{sec:algorithms} when ${\cal D}$ is a, possibly correlated, uniform distribution. Before doing this, we establish a useful lemma.
\begin{lemma}\label{lem:VVCG} 
For all $\mathcal{F}$ and $\mathcal{D}$, if every corner of $F(\mathcal{F},\mathcal{D})$ is a vector of rational numbers of bit complexity $b$, {the probabilities used by ${\cal D}$ have bit complexity $\ell$}, and $SO$'s input $\vec{\pi}$ is a vector of rational numbers of bit complexity $c$, then the following are true.
\begin{enumerate}
\item The separation oracle $SO$ of Section~\ref{sec:separation} can be implemented to run in time polynomial in $n\sum_{i=1}^{m} |T_i|$, $b$, $c$, {$\ell$}, $|\mathcal{D}|$, and $rt_{\mathcal{F}}\left(\poly(n\sum_{i=1}^{m} |T_i|, b,c, {\ell})\right)$. Furthermore, the coefficients of any hyperplane that can be possibly output by $SO$ have bit complexity $\poly(n\sum_{i=1}^{m} |T_i|,b)$. 
\item If the corner oracle $CO$ of Section~\ref{sec:decomposition} only takes as input hyperplanes output by $SO$, it can be implemented to run in time polynomial in $n\sum_{i=1}^{m} |T_i|$, $b$, $\ell$, $|\mathcal{D}|$, and $rt_{\mathcal{F}}\left(\poly(n\sum_{i=1}^{m} |T_i|,b,\ell)\right)$.
\end{enumerate}
\end{lemma}

\begin{proof} 
We first bound the runtime of $SO$, using Theorem~\ref{thm:ellipsoid}. The separation oracle is a linear program with $1+n\sum_{i=1}^{m} |T_i|$ variables, $2n\sum_{i=1}^{m} |T_i|$ constraints, and an internal separation oracle $\widehat{SO}$. $\widehat{SO}$ on input $(\vec{w},t)$ simply checks if $R_{\mathcal{F}}(\vec{w}') \cdot \vec{w}$, where $\vec{w}'$ is  the perturbation of $\vec{w}$ according to Lemma~\ref{lem:tiebreaking}, is smaller than or equal to $t$. If not, it outputs the separation hyperplane $(R_{\mathcal{F}}(\vec{w}') , -1) (\vec{w},y) \le 0$. Given that $R_{\mathcal{F}}(\vec{w}')$ is a corner of the polytope and corners have bit complexity $b$, Theorem~\ref{thm:ellipsoid} tells us that $\widehat{SO}$ will only be called on $\vec{w},t$ whose coordinates are rational numbers of bit complexity at most $\poly(n\sum_{i=1}^{m} |T_i|,\max\{b,c\})$. To compute $R_{\mathcal{F}}(\vec{w}')$ exactly we can enumerate every profile in the support of $\mathcal{D}$, run $VVCG_{\mathcal{F}}(\vec{w}')$, and see if bidder $i$ was awarded item $j$, for all $i,j$. As the coordinates of $\vec{w}'$ are rational numbers of bit complexity {$\poly(n\sum_{i=1}^{m} |T_i|,\max\{b,c, {\ell}\})$} (after Lemma~\ref{lem:tiebreaking} was applied to $\vec{w}$)), this method exactly computes $R_{\mathcal{F}}(\vec{w}')$ in time polynomial in $n\sum_{i=1}^{m} |T_i|, |\mathcal{D}|$, $b, c,$ {$\ell$} and $rt_{\mathcal{F}}(\poly(n\sum_{i=1}^{m}|T_i|,\max\{b,c,{\ell}\}))$. After computing $R_{\mathcal{F}}(\vec{w}')$, $\widehat{SO}$ simply takes a dot product and makes a comparison, so the total runtime of $SO$ is polynomial in $n\sum_{i=1}^{m} |T_i|,b, c, \ell, |\mathcal{D}|$ and {$rt_{\mathcal{F}}(\poly(n\sum_{i=1}^{m} |T_i|,\max\{b,c,\ell\}))$}. {Also, by Lemma~\ref{lem:low bit SO}, we know that all hyperplanes output by $SO$ have coefficients that are rational numbers of bit complexity $\poly(n\sum_{i=1}^{m} |T_i|,b)$, which is independent of $c$.}

The corner oracle of Section~\ref{sec:decomposition} has three steps. The first step is simply computing the average of at most $n\sum_{i=1}^{m} |T_i|$ vectors in $\mathbb{R}^{n\sum_{i=1}^{m} |T_i|}$, whose coordinates are rational numbers of bit complexity $\poly(n\sum_{i=1}^{m} |T_{i}|,b)$ (by the previous paragraph). The second step is applying Lemma~\ref{lem:tiebreaking} to the averaged weight vector to get $\vec{w}'$. So each weight of $\vec{w}'$ is a rational number of bit complexity $\poly(n\sum_{i=1}^{m} |T_i|,\ell,b)$. The last step is computing $R_{\mathcal{F}}(\vec{w}')$. It is clear that the first two steps can be implemented in the desired runtime. As the coordinates of $\vec{w}'$ are rational numbers of bit complexity $\poly(n\sum_{i=1}^{m} |T_i|,b,\ell)$, we can use the same method as in the previous paragraph to compute $R_{\mathcal{F}}(\vec{w}')$ in time polynomial in $n\sum_{i=1}^{m} |T_i|, |\mathcal{D}|$, $b, \ell$ and $rt_{\mathcal{F}}(\poly(n\sum_{i=1}^{m} |T_i|,b,\ell))$, to implement $CO$ in the desired runtime.
\end{proof}

\begin{corollary}\label{cor:uniform} For all $\mathcal{F}$, if $\mathcal{D}$ is a (possibly correlated) uniform distribution over $k$ profiles (possibly with repetitions), and $SO$'s input $\vec{\pi}$ is a vector of rational numbers of bit complexity $c$, then the following are true.
\begin{enumerate}
\item The separation oracle $SO$  of Section~\ref{sec:separation} can be implemented to run in time polynomial in $n\sum_{i=1}^{m}|T_i|$, $k$, $c$ and $rt_{\mathcal{F}}(\poly(n\sum_{i=1}^{m} |T_i|,\log k, c))$. Furthermore, the coefficients of any hyperplane that can be possibly output by $SO$ have bit complexity $\poly(n\sum_{i=1}^{m} |T_i|,\log k)$.

\item If the corner oracle $CO$ of Section~\ref{sec:decomposition} only takes as inputs hyperplanes output by $SO$ as input, it can be implemented in time polynomial in $n\sum_{i=1}^{m}|T_i|$, $k$, and $rt_{\mathcal{F}}(\poly(n\sum_{i=1}^{m} |T_i|,\log k))$.  \end{enumerate}
\end{corollary}

\begin{prevproof}{Corollary}{cor:uniform} Every corner of $F(\mathcal{F},\mathcal{D})$ is the reduced form of a deterministic mechanism. So let us bound the bit complexity of the reduced form $\pi$ of a deterministic mechanism $M$. We may let $n_{ij}(A)$ denote the number of profiles (with repetition) in the support of $\mathcal{D}$ where bidder $i$'s type is $A$, and $M$ awards item $j$ to $i$, and let $d_{ij}(A)$ denote the number of profiles where bidder $i$'s type is $A$. Then for all $i,j,A \in T_i$, $\pi_{ij}(A) = \frac{n_{ij}(A)}{d_{ij}(A)}$. As $n_{ij}(A)$ and $d_{ij}(A)$ are integral and at most $k$, $\pi_{ij}(A)$ has bit complexity $O(\log k)$. So we may take $b = O(\log k)$, $\ell = O(\log k)$, $|\mathcal{D}| = k$, and apply Lemma~\ref{lem:VVCG}.
\end{prevproof}

Next we bound the running time of the decomposition algorithm.

\begin{corollary}\label{cor:geometric alg for uniform}
For all $\mathcal{F}$, if $\mathcal{D}$ is a (possibly correlated) uniform distribution over $k$ profiles (possibly with repetitions), then given a reduced form $\vec{\pi}\in F(\mathcal{F},\mathcal{D})$, which is a vector of rational numbers with bit complexity $c$, we can rewrite $\vec{\pi}$ as a convex combination of corners of $F(\mathcal{F},\mathcal{D})$ using the geometric algorithm of Theorem~\ref{thm:geometric} with running time polynomial in $n\sum_{i=1}^{m}|T_i|$, $k$, $c$ and $rt_{\mathcal{F}}(\poly(n\sum_{i=1}^{m} |T_i|,\log k, c))$. 
\end{corollary}

\begin{prevproof}{Corollary}{cor:geometric alg for uniform}
From Corollary~\ref{cor:uniform} it follows that  the coefficients of any hyperplane that can be possibly output by $SO$ have bit complexity $\poly(n\sum_{i=1}^{m} |T_i|,\log k)$. So to apply Theorem~\ref{thm:geometric} it suffices to bound the running time of SO and CO on vectors of rational numbers of bit complexity $c'=\poly(n\sum_{i=1}^{m}|T_i|,\log k, c)$. Using Corollary~\ref{cor:uniform} this is polynomial in $n\sum_{i=1}^{m}|T_i|$, $k$, $c$ and $rt_{\mathcal{F}}(\poly(n\sum_{i=1}^{m} |T_i|,\log k, c))$. Combining this bound with Theorem~\ref{thm:geometric} finishes the proof.
\end{prevproof}

\section{Proofs Omitted From Section~\ref{sec:approximate}: Approximate Implementation}\label{app:approximations}
\notshow{In Section~\ref{sec:exact}, we showed that in time polynomial in $|\mathcal{D}|$, we can exactly implement the separation oracle and the decomposition algorithm. Alone, this result is entirely uninteresting, as we could do much more interesting computations in time polynomial in $|\mathcal{D}|$, including exactly solve the MDMDP~\cite{DW12}.}

\paragraph{Notation} We will use the following notation throughout this section: $k$ is the number of samples taken directly from $\mathcal{D}$, $k'$ is the number of samples taken from $\mathcal{D}_{-i}$ after fixing $t_i = A$ for all $i,A \in T_i$, $k''$ is the total number of samples taken (i.e. $k'' = k + k'\sum_i |T_i|$) and $\mathcal{D}'$ is the distribution that samples one of the $k''$ sampled profiles uniformly at random. We also make use of the following, standard Chernoff bound:
\begin{theorem}(Hoeffding~\cite{Hoeffding63}) Let $X_1,\ldots,X_n$ be independent random variables in $[0,1]$, and let $X = \sum_i X_i/n$. Then $\Pr[ |X - \mathbb{E}[X]| > t] \leq 2e^{-2t^2n}.$
\end{theorem}

\subsection{Every point in $F(\mathcal{F},\mathcal{D})$ is close to some point in $F(\mathcal{F},\mathcal{D}')$ }\label{app:oneway}
The desired claim is stated at the end of the section as Corollary~\ref{cor:oneway}, and the proof is obtained by a series of small technical lemmas. Throughout this section, we will be interested in whether a profile with $t_i = A$ in the support of ${\cal D}'$ is \emph{genuine} from the perspective of bidder $i$ (i.e. it was sampled from ${\cal D}$ without conditioning on anything except perhaps $t_i = A$) or \emph{biased} (i.e. it was sampled by conditioning on $t_{i'} = A'$ for some $i' \neq i$). Now let us fix an allocation rule $M$ with reduced form $\vec{\pi}$ if bidders are sampled from ${\cal D}$. What does the reduced form of $M$ look like for bidders sampled from ${\cal D'}$? The expected probability (over the randomness in the types of the other bidders) that bidder $i$ receives item $j$ conditioning on $t_i = A$ on a genuine from bidder $i$'s perspective profile is exactly $\pi_{ij}(A)$. However, if the profile is biased, the probability that bidder $i$ receives item $j$ might have nothing to do with $\pi_{ij}(A)$. So for a fixed $\mathcal{D}'$, we'll let $G_i(A)$ denote the (random) set of profiles in the support of $\mathcal{D}'$ with $t_i = A$ that were obtained genuinely from the perspective of bidder $i$, and $B_i(A)$ denote the set of profiles with $t_i = A$ that are biased from the perspective of bidder $i$.

\begin{lemma}\label{lem:finalclose} Fix $i$ and $A \in T_i$,  let $M$ be any allocation rule, and let $\mathcal{D}'$ be such that $|G_i(A)| = x$ and $|B_i(A)| = z$ (i.e. condition on $|G_i(A)| = x,|B_i(A)| = z$ and then sample $\mathcal{D}'$). Then over the randomness in generating $\mathcal{D}'$, for all items $j$ and all $t \leq 1$, if $\vec{\pi}'$ denotes the reduced form of $M$ when bidders are sampled from $\mathcal{D}'$ and $\vec{\pi}$ denotes the reduced form of $M$ when bidders are sampled from $\mathcal{D}$, we have:

$$\Pr\left[|\pi_{ij}(A) - \pi'_{ij}(A)| > t+\frac{z}{x}\right] \leq 2e^{-2t^2x}.$$
\end{lemma}

\begin{proof}
Label the profiles in $G_i(A)$ as $P_1,\ldots,P_x$ and $B_i(A)$ as $P_{x+1},\ldots,P_{x+z}$, and let $X_a$ be the random variable denoting the probability that $M$ awards item $j$ to bidder $i$ on profile $P_a$. Then for all $1 \leq a \leq x$ we have $\mathbb{E}[X_a] = \pi_{ij}(A)$. For all $a > x$, we have $0 \leq \mathbb{E}[X_a] \leq 1$. As $\pi'_{ij}(A) = \frac{1}{z+x}\sum_a X_a$, we see that:

\begin{align*}
\pi_{ij}(A) - \frac{z}{x+z} \leq \frac{x}{x+z}\pi_{ij}(A)&\leq
\\\frac{x}{x+z}\pi_{ij}(A)&+\frac{1}{x+z}\sum_{i=1}^{z}\mathbb{E}[X_{x+i}]=\mathbb{E}[\pi'_{ij}(A)] \le\\
 &~~~~~~~~~~~ \frac{x}{x+z}\pi_{ij}(A) + \frac{z}{x+z} \leq \pi_{ij}(A) + \frac{z}{x+z}.\\
\end{align*}
So $|\mathbb{E}[\pi'_{ij}(A)] - \pi_{ij}(A)| \leq \frac{z}{x+z}.$ Therefore, the triangle inequality tells us that in order to have $|\pi_{ij}(A) - \pi'_{ij}(A)| > t+\frac{z}{x}>t+\frac{z}{z+x}$, we must have $|\pi'_{ij}(A) - \mathbb{E}[\pi'_{ij}(A)]| > t$. As $\pi'_{ij}(A)$ is the average of $x+z$ independent trials, by the Hoeffding inequality, this happens with probability at most $2e^{-2t^2x}$.
\end{proof}

\begin{lemma}\label{lem:allfine} For any $i,A \in T_i$, $|G_i(A)| \geq k'$. Furthermore, if $k > k'\sum_{i' \neq i} |T_{i'}|$, for all $x \leq 1$ we have:

$$\Pr\left[|B_i(A)| > \left(2x+\frac{k'}{k}\right)\sum_{i' \neq i}|T_{i'}|\cdot |G_i(A)| \right] \leq 4e^{-2x^2k'n\sum_{i' \neq i} |T_{i'}|}.$$
\end{lemma}

\begin{prevproof}{Lemma}{lem:allfine}
The first claim is obvious, as we fix $t_i = A$ in exactly $k'$ profiles. For the second claim, there are $k'(\sum_{i' \neq i} |T_{i'}|)$ independent chances to get a profile in $B_i(A)$. Each chance occurs with probability $q = \Pr[t_i = A]$. There are $k$ independent chances to get additional profiles in $G_i(A)$, and each occurs with probability $q$. Therefore, we get that $\mathbb{E}[|B_i(A)|] = qk'\sum_{i' \neq i}|T_i|$ and $\mathbb{E}[|G_i(A)|] = k'+qk$. Applying the Hoeffding inequality, we get 
$$\Pr\left[~\Big|~ |B_i(A)| - qk'\sum_{i' \neq i} |T_{i'}|~\Big|~ > xk'\sum_{i' \neq i}|T_{i'}|\right]\leq 2e^{-x^{2}k'\sum_{i'\neq i}|T_{i'}|},$$
and
$$\Pr\left[ \Big| |G_{i}(A)|-(k'+qk) \Big|> xk\right]\leq 2e^{-x^{2}k}.$$
Then since $k\geq k'\sum_{i'\neq i}|T_{i'}|$, by union bound, we get that for any $x \leq 1$, with probability at least $1-4e^{-2x^2k'\sum_{i' \neq i} |T_{i'}|}$ we have the following two inequalities:
$$|B_i(A)| \leq qk'\sum_{i' \neq i} |T_{i'}| + xk'\sum_{i' \neq i}|T_{i'}|$$
\begin{equation}\label{eq:g}
|G_i(A)| \geq k' + \max\{0,(q-x)k\}
\end{equation}

\noindent which imply the following two inequalities by ignoring one of the positive terms on the right-hand side of Equation~\eqref{eq:g}:

\begin{align}
|B_i(A)| &\leq (q+x)\sum_{i' \neq i}|T_{i'}|\cdot |G_i(A)| \label{eq:first}\\
|B_i(A)| &\leq \frac{q+x}{q-x}\cdot  \frac{k'\sum_{i' \neq i}|T_{i'}|}{k}|G_i(A)|~~~\text{(we only use this when $q > x$)} \label{eq:second}
\end{align}
When $q \leq x+\frac{k'}{k}$, Equation~\eqref{eq:first} gives a better bound. Otherwise, Equation~\eqref{eq:second} gives a better bound. As $q$ decreases, the bound from Equation~\eqref{eq:first} only gets better. Likewise, as $q$ increases, the bound from Equation~\eqref{eq:second} only gets better. So for any $q$, one of the bounds will yield:
$$|B_i(A)| \leq \left(2x+\frac{k'}{k}\right)\sum_{i' \neq i}|T_{i'}|\cdot |G_i(A)|$$
as desired.\end{prevproof}

\begin{corollary}\label{cor:allfine} Let $M$ be any allocation rule and assume $k > k'\sum_{i' \neq i} |T_{i'}|$. Then for all items $j$, bidders $i$, types $A \in T_i$, and all $t \leq 1$, if $\vec{\pi}'$ denotes the reduced form of $M$ when bidders are sampled from $\mathcal{D}'$ and $\vec{\pi}$ denotes the reduced form of $M$ when bidders are sampled from $\mathcal{D}$, we have:

$$\Pr\left[|\pi_{ij}(A) - \pi'_{ij}(A)| > t + \left(2t+\frac{k'}{k}\right)\sum_{i' \neq i}|T_{i'}|\right] \leq 6e^{-2t^2k'}.$$
\end{corollary}

\begin{prevproof}{Corollary}{cor:allfine}
Lemma~\ref{lem:allfine} says that with probability at least $1-4e^{-2x^2k'\sum_{i' \neq i} |T_{i'}|} \geq 1-4e^{-2x^2k'}$, $\mathcal{D}'$ is such that $|B_i(A)| \leq (2x+\frac{k'}{k})\sum_{i' \neq i}|T_{i'}|\cdot |G_i(A)|$ and $|G_i(A)| \geq k'$. For such $\mathcal{D}'$, the bound given by Lemma~\ref{lem:finalclose} is:

$$\Pr\left[ |\pi_{ij}(A) - \pi'_{ij}(A)| > t + \left(2x+\frac{k'}{k}\right)\sum_{i' \neq i}|T_{i'}|\right] \leq 2e^{-2t^2k'}.$$

So after taking a union bound and setting $x = t$ we get the desired claim.
\end{prevproof}

\begin{corollary}\label{cor:wholeform} Let $M$ be any allocation rule and assume $k > k'\sum_{i' \neq i} |T_{i'}|$. Then if $\vec{\pi}'$ denotes the reduced form of $M$ when bidders are sampled from $\mathcal{D}'$ and $\vec{\pi}$ denotes the reduced form when bidders are sampled from $\mathcal{D}$, we have:

$$\Pr\left[|\vec{\pi} - \vec{\pi}'|_{\infty} > t + \left(2t+\frac{k'}{k}\right)\sum_{i}|T_{i}|\right] \leq 6n\sum_{i=1}^{m} |T_i|e^{-2t^2k'}.$$
\end{corollary}
\begin{prevproof}{Corollary}{cor:wholeform} Use Corollary~\ref{cor:allfine}, observe that $\sum_{i'} |T_{i'}| > \sum_{i' \neq i} |T_{i'}|$ for all $i$, and take a union bound over all $j,i,A \in T_i$.
\end{prevproof}

\begin{corollary}\label{cor:oneway} Assume $k > k'\sum_{i} |T_{i}|$. Then for all $t \leq 1$, with probability at least $$1-6n\sum_{i=1}^{m} |T_i|e^{-2t^2k'-n\sum_{i=1}^{m} |T_i|\ln t},$$ for every $\vec{\pi} \in F(\mathcal{F},\mathcal{D})$, there is a $\vec{\pi}' \in F(\mathcal{F},\mathcal{D}')$ with $|\vec{\pi}- \vec{\pi}'|_\infty \leq 2t+(2t+\frac{k'}{k})\sum_{i=1}^{m}|T_{i}|$.
\end{corollary}

\begin{prevproof}{Corollary}{cor:oneway}
Consider an $t$-$\ell_\infty$ cover of {$F(\mathcal{F},\mathcal{D})$} such that every point in  $F(\mathcal{F},\mathcal{D})$ is within $\ell_\infty$ distance $t$ of a point in the cover. There is certainly a cover that uses at most $\left(\frac{1}{t}\right)^{n\sum_{i=1}^{m} |T_i|}$ points, as there is a cover of the entire hypercube using this many points. If for every point $\vec{x}$ in the cover, there is a point $\vec{z}\in F(\mathcal{F},\mathcal{D}')$, such that $|\vec{x}- \vec{z}|_\infty \leq t+(2t+\frac{k'}{k})\sum_{i}|T_{i}|$, then clearly for every point $\vec{\pi}$ in $F(\mathcal{F},\mathcal{D})$, there is a point $\vec{\pi}'\in F(\mathcal{F},\mathcal{D}')$ such that $|\vec{\pi}- \vec{\pi}'|_\infty \leq 2t+(2t+\frac{k'}{k})\sum_{i}|T_{i}|$ by the triangle inequality. So we simply take a union bound over all $e^{-n\sum_{i=1}^{m} |T_i| \ln t}$ points in the cover and apply Corollary~\ref{cor:wholeform} to conclude the proof.
\end{prevproof}

\subsection{Every point in $F(\mathcal{F},\mathcal{D}')$ is close to some point in $F(\mathcal{F},\mathcal{D})$}\label{app:otherway}

In the previous section, we showed that for all $\vec{\pi} \in F(\mathcal{F},\mathcal{D})$ there is a nearby $\vec{\pi}' \in F(\mathcal{F},\mathcal{D}')$ using the probabilistic method over the choice of $\mathcal{D}'$. In this section, we want to show the other direction, namely that for any reduced form $\vec{\pi}' \in F(\mathcal{F},\mathcal{D}')$ there is a nearby reduced form $\vec{\pi} \in F(\mathcal{F},\mathcal{D})$. However,  it is not clear how to use the probabilistic method over the choice of $\mathcal{D}'$ to prove this, as for $\vec{\pi}' \in F(\mathcal{F},\mathcal{D}')$ the allocation rule that implements $\vec{\pi}'$ is heavily dependent on $\mathcal{D}'$, which is the object with respect to which we plan to apply the probabilistic method. To go around this circularity we show that after fixing $k$ and $k'$, but before selecting $\mathcal{D}'$, there are not too many allocation rules that could possibly implement a reduced form that is a corner of $F(\mathcal{F},\mathcal{D}')$. Specifically, we show that the reduced form with respect to ${\cal D}'$ of any simple virtual VCG allocation rule is equivalent to one whose functions only output rational numbers with bit complexity that only depends on $k$, $k'$ and the dimension. In particular, regardless of $\mathcal{D}'$, there is an a-priori fixed set of allocation rules that implement the corners of $F(\mathcal{F},\mathcal{D}')$ whose cardinality depends only on $k$, $k'$ and the dimension. So we can still use concentration of measure to argue that the reduced form $\pi'$ of every corner of $F(\mathcal{F},\mathcal{D}')$ has a nearby reduced form $\pi \in  F(\mathcal{F},\mathcal{D})$. And, as every point in $F(\mathcal{F},\mathcal{D}')$ is a convex combination of the corners of $F(\mathcal{F},\mathcal{D}')$, this suffices to prove the desired claim. Our starting point is the following observation.

\begin{lemma}\label{lem:lowbits} Let $\vec{\pi}$ be the reduced form of a simple virtual VCG allocation with respect to $\mathcal{D}'$. Then each $\pi_{ij}(A)$ is a rational number of bit complexity $O(\log k'')$.
\end{lemma}
\begin{prevproof}{Lemma}{lem:lowbits} In every simple virtual VCG allocation, the probability that bidder $i$ gets item $j$ on profile $P$ is always $1$ or $0$. Therefore, if $t_i = A$ in exactly $x$ profiles, and bidder $i$ receives item $j$ in exactly $x'$ of those profiles, $\pi_{ij}(A) = x'/x$. As $x' \leq x \leq k''$, this value clearly has bit complexity $O(\log k'')$.
\end{prevproof}

Given Lemma~\ref{lem:lowbits} the corners of $F(\mathcal{F},\mathcal{D}')$ belong to a set of at most $(k'')^{O(n \sum_{i}|T_{i}| )}$ reduced forms that is independent of ${\cal D}'$. Still the allocation rules that implement those reduced forms may vary depending on ${\cal D}'$. We show that this can be mitigated by appealing to the correctness of the decomposition algorithm of Section~\ref{sec:decomposition}.

\begin{lemma}\label{lem:costascorner} Suppose that the allocation rule $M$ implements a corner $\vec{\pi}$ of $F(\mathcal{F},\mathcal{D}')$. Then there is a virtual VCG allocation rule $VVCG(\{f_i\}_{i\in[m]})$ whose reduced form with respect to ${\cal D}'$ is exactly $\vec{\pi}$ and such that each $f_i$ only outputs rational numbers of bit complexity $f_c(n\sum_{i=1}^{m}|T_i|\log k'')$, where $f_c(\cdot)$ is a polynomial function. 

Moreover, for any input $\vec{\pi} \in F(\mathcal{F},\mathcal{D}')$ to the decomposition algorithm of Section~\ref{sec:decomposition} for ${\cal D}'$, the output decomposition uses simple virtual VCG allocation rules whose weight functions only output rational numbers of complexity $f_c(n\sum_{i=1}^{m}|T_i|\log k'')$. 
\end{lemma}

\begin{prevproof}{Lemma}{lem:costascorner}
For the first part of the theorem, suppose that a corner $\vec{\pi}$ of  $F(\mathcal{F},\mathcal{D}')$ is fed as input to the decomposition algorithm of Section~\ref{sec:decomposition}. From the correctness of this algorithm it follows that the output decomposition consists of a single reduced form, namely $\vec{\pi}$ itselft, as $\vec{\pi}$ is a corner of $F(\mathcal{F},\mathcal{D}')$. The decomposition algorithm of~\cite{CaiDW12} (which is at the heart of the decomposition algorithm of Section~\ref{sec:decomposition}) has the property that every corner used in the output decomposition will always be an output of the corner oracle. So let us try to argue that all reduced forms that are possibly output by the corner oracle can be implemented by a small set of allocation rules that does not depend on ${\cal D}'$. We use the following lemma:

\begin{lemma}\label{lem:sepbits} On any input $\vec{\pi}$, {the coefficients of the hyperplane} output by the separation oracle $SO$ of Section~\ref{sec:separation} using $\mathcal{D}'$ as the bidder-type distribution are rational numbers of bit complexity $f_{s}(n\sum_{i=1}^{m} |T_i|\log k'')$, where $f_{s}(\cdot)$ is a polynomial function. 
\end{lemma}
\begin{proof}
By Lemma~\ref{lem:lowbits}, we know that all coordinates of all corners of $F(\mathcal{F},\mathcal{D}')$ can be described using at most $O(\log k'')$ bits. Lemma~\ref{lem:low bit SO} now tells us that $SO$ will only output rational numbers of bit complexity $\poly(n\sum_{i=1}^{m} |T_i|,\log k'')$. \end{proof}

\noindent Now let's go back to the corner oracle. By Lemma~\ref{lem:sepbits} the weights input to the corner oracle will always be rational numbers of bit complexity $f_s(n\sum_{i=1}^{m} |T_i|\log k'') = \poly(n\sum_{i=1}^{m} |T_i|,\log k'')$. As the number of hyperplanes input to the corner oracle will never be more than $n\sum_{i=1}^{m} |T_i|$, the weights obtained by averaging in step 2 of the corner oracle require at most an additional 
$O(\log (n\sum_{i=1}^{m} |T_i|))$
bits. Given the above and that, for all $i, A$, $\Pr[t_i = A]$ is a multiple of $1/k''$ and hence has bit complexity 
$O( \log k'' )$, 
the tie-breaking rule in step 3 of the corner oracle results in a weight vector whose coordinates have bit complexity 
$\poly(n\sum_{i=1}^{m} |T_i|,\log k'')$. As $\Pr[t_i = A]$ is a multiple of $1/k''$ for all $i,A$, transforming from the weight vector representation of the simple VCG mechanism computed by the corner oracle to the weight function representation adds at most an additional $O(\log k'')$ bits per weight.

The second part of the lemma is already implied by the above discussion. As we noted above the decomposition algorithm of~\cite{CaiDW12} (which is at the heart of the decomposition algorithm of Section~\ref{sec:decomposition}) has the property that every corner used in the output decomposition will always be an output of the corner oracle. And we argued that the corner oracle for ${\cal D}'$ outputs simple virtual VCG allocation rules whose weight functions only output rationals of bit complexity as bounded above.
\end{prevproof}

Lemma~\ref{lem:costascorner} implies that, before we have sampled $\mathcal{D}'$ but after we have chosen $k$ and $k'$, there is a fixed set ${\cal K}$ of at most $4^{n\sum_{i=1}^{m}|T_i|f_{c}(n\sum_{i=1}^{m} |T_i|\log k'')}$ simple virtual VCG allocation rules (namely those whose weight functions only output rational numbers of bit complexity $f_{c}(n\sum_{i=1}^{m} |T_i|\log k'')$) such that, no matter what ${\cal D}'$ is sampled, all corners of $F(\mathcal{F},\mathcal{D}')$ can be implemented by a simple virtual VCG allocation rule in ${\cal K}$. Moreover, the decomposition algorithm of Section~\ref{sec:decomposition} only uses simple virtual VCG mechanisms from ${\cal K}$ in its support. This implies the following.

\begin{corollary}\label{cor:otherway}  Assume $k > k'\sum_{i} |T_{i}|$ and $t \leq 1$. Then, with probability at least $$1-6n\sum_{i=1}^{m} |T_i|e^{-2t^2k'+ n\sum_{i=1}^{m}|T_{i}| f_{c}(n\sum_{i=1}^{m}|T_{i}| \log k'')\ln 4},$$ the following hold, where $f_c(\cdot)$ is a polynomial function: 
\begin{enumerate}
\item every $\vec{\pi}' \in F(\mathcal{F},\mathcal{D}')$ has some $\vec{\pi} \in F(\mathcal{F},\mathcal{D})$ with $|\vec{\pi}- \vec{\pi}'|_\infty \leq t+ (2t+\frac{k'}{k})\sum_{i}|T_{i}|$; 
\item if $\vec{\pi}$ is the reduced form with respect to ${\cal D}$ of the distribution over simple virtual VCG allocation rules that is output on input $\vec{\pi}' \in F(\mathcal{F},\mathcal{D}')$ by the decomposition algorithm of Section~\ref{sec:decomposition} for ${\cal D'}$ then $|\vec{\pi}- \vec{\pi}'|_\infty \leq t+ (2t+\frac{k'}{k})\sum_{i}|T_{i}|$.
\end{enumerate}
\end{corollary}

\begin{prevproof}{Corollary}{cor:otherway}
For a fixed simple virtual VCG allocation rule $M \in {\cal K}$, Corollary~\ref{cor:wholeform} guarantees that the reduced form of $M$ when consumers are sampled from $\mathcal{D}$, $\vec{\pi}(M)$, and when consumers are sampled from $\mathcal{D}'$, $\vec{\pi}'(M)$, satisfy: $|\vec{\pi}(M) - \vec{\pi}'(M)|_\infty \leq t + (2t+\frac{k'}{k})\sum_{i=1}^{m}|T_{i}|$ with probability at least $1- 6n\sum_{i=1}^{m} |T_i|e^{-2t^2k'}$. In addition, Lemma~\ref{lem:costascorner} guarantees that $|{\cal K}|\le4^{f_c(n\sum_{i=1}^{m} |T_i|\log k'')n\sum_{i=1}^{m} |T_i|}$. Because this set is fixed a priori and independent of $\mathcal{D}'$, we may take a union bound over the elements of the set to get that the same claim holds for \emph{all} simple virtual VCG allocation rules in ${\cal K}$ with probability at least

$$1- 6n\sum_{i=1}^{m} |T_i|e^{-2t^2k'} 4^{n\sum_{i=1}^{m} |T_i|f_{c}(n\sum_{i=1}^{m}|T_{i}|\log k'')}.$$

We proceed to show (i) and (ii) conditioning on the above. For (i) we use the first part of Lemma~\ref{lem:costascorner} to get that for all $\vec{\pi}' \in F(\mathcal{F},\mathcal{D}')$, we can write $\vec{\pi}' = \sum_a p_a \vec{\pi}'(M_a)$, where for all $a$: $M_a \in {\cal K}$,  $p_a > 0$, and $\sum_a p_a = 1$. If we consider the exact same distribution over simple virtual VCG allocation rules when consumers are sampled from $\mathcal{D}$, the reduced form will be $\vec{\pi} = \sum_a p_a \vec{\pi}(M_a)$. Given that for all $M_a \in {\cal K}$ we have $|\vec{\pi}(M_a) - \vec{\pi}'(M_a)|_\infty \leq t + (2t+\frac{k'}{k})\sum_{i=1}^{m}|T_{i}|$, we have $|\vec{\pi} - \vec{\pi}'|_\infty \leq t + (2t+\frac{k'}{k})\sum_{i=1}^{m}|T_{i}|$ as well.

For (ii) the proof is virtually identical. By the second part of Lemma~\ref{lem:costascorner} the simple virtual VCG allocation rules in the support of the decomposition belong to the set ${\cal K}$. We proceed as above.
\end{prevproof}
 
\subsection{Putting Everything Together}\label{app:final} 

\subsubsection{Setting $k$ and $k'$, and sampling ${\cal D'}$}
Algorithm~\ref{alg:preprocess} is a preprocessing algorithm used to set the parameters $k$ and $k'$ that were left free in Sections~\ref{app:oneway} and~\ref{app:otherway}.

\begin{algorithm}[ht]
        \caption{Pre-processing to generate a proxy distribution $\mathcal{D}'$ for $\mathcal{D}$. The desired $\ell_\infty$ accuracy is $\epsilon$.}
    \begin{algorithmic}[1]\label{alg:preprocess}
        \STATE Input: $\mathcal{D}$. Denote by $T = \sum_{i=1}^{m} |T_i|$.
        \STATE Set $t = \frac{\epsilon}{6T}, k' = \frac{n^2T^2f_c(nT)}{t^3}, k = \frac{4k'T}{\epsilon}$
        \STATE Build $\mathcal{D}'$ by sampling $k$ profiles independently from $\mathcal{D}$. For each $i,A \in T_i$, fix $t_i = A$ and sample $k'$ profiles independently from $\mathcal{D}_{-i}$. $\mathcal{D}'$ picks one of the $k + k'\sum_{i=1}^{m} |T_i|$ sampled profiles uniformly at random.
        \STATE Output $t, k, k', \mathcal{D}'.$
     \end{algorithmic}
\end{algorithm}

\subsubsection{Separation Oracle for Approximating Polytope $F(\mathcal{F},\mathcal{D}')$}

We provide the proof Theorem~\ref{thm:appxSO}.\\

\begin{prevproof}{Theorem}{thm:appxSO} We use the shorthand $T=\sum_{i=1}^{m} |T_i|$. After plugging in our choice of $t,k',k$, we see that $2t + (2t+\frac{k'}{k})\sum_{i=1}^{m} |T_i| \leq \epsilon$. So we just have to verify that the probability bounds given by Corollaries~\ref{cor:oneway} and~\ref{cor:otherway} are as desired.

Before plugging in the choice of $t,k',k$ to Corollary~\ref{cor:oneway}, we get that the first claim is true with probability at least $1-6nTe^{-2t^2k'-nT\ln t}$. As $k' \geq nT/t^3$, this is at least $1-6nTe^{-(2/t-\ln(1/t))nT}$. As $1/t$ asymptotically dominates $\ln(1/ t)$ as $t \rightarrow 0$,  this probability becomes $1-e^{-\Omega(nT/\epsilon)}$ after plugging in our choice of $t$.

Before plugging in the choice of $t,k',k$ to Corollary~\ref{cor:otherway}, we get that the second claim is true with probability at least $1-6nTe^{-2t^2k' + nTf_c(nT\log (k+k')) \ln 4}$. Plugging in the choice of $k'$ and $k$ (and observing that $\log (k+ k')$ is $O(\log \left(nT/t \right)$) this becomes:

$$1-6nTe^{-2n^2T^2f_c(nT)/t ~+~nTf_c(nT \cdot O(\log (nT/t)))\ln 4}.$$

Therefore, the ratio of the absolute value of the negative term in the exponent to the value of the positive term is $\frac{nT/t}{\poly\log (nT/t)}$, so the negative term dominates asymptotically as $t \rightarrow 0$. Therefore, the entire probability is $1-e^{-\Omega(nT/\epsilon)}$ after plugging in our choice of $t$.

The bound on the running time follows directly from Corollary~\ref{cor:uniform} and our choice of parameters.
%
%
%
\end{prevproof}


\subsubsection{Decomposition Algorithm for Approximating Polytope $F(\mathcal{F},\mathcal{D}')$}
Algorithm~\ref{alg:decomposition} describes our decomposition algorithm for $F(\mathcal{F},\mathcal{D}')$. After stating it, we analyze it.

\begin{algorithm}[ht!]
        \caption{Algorithm for decomposing a reduced form $\vec{\pi}' \in  F(\mathcal{F},\mathcal{D}')$ into a distribution over simple virtual VCG allocations.}
    \begin{algorithmic}[1]\label{alg:decomposition}
        \STATE Input: $\mathcal{F}$, $\mathcal{D}'$, $\vec{\pi}'\in F(\mathcal{F},\mathcal{D}')$. 
			\STATE Run the geometric algorithm of~\cite{CaiDW12} on $F(\mathcal{F},\mathcal{D}')$ using the separation oracle of Section~\ref{sec:separation} and the corner oracle of Section~\ref{sec:decomposition}. The output will be a collection of at most $n\sum_{i=1}^{m} |T_i|+1$ corners output by the corner oracle. These will be simple virtual VCG allocation rules, whose weight functions we denote by $\{f_i\}_{i\in[m]}^{(1)},\ldots,\{f_i\}_{i\in[m]}^{(n\sum_{i=1}^{m} |T_i|+1)}$. We also denote by $p_j$ the probability placed on $VVCG\left(\{f_i\}_{i\in[m]}^{(j)}\right)$ in the output decomposition.\\
				\STATE Output the allocation rule $M$ that runs $VVCG\left(\{f_i\}_{i\in[m]}^{(j)}\right)$ with probability $p_j$. 
     \end{algorithmic}
\end{algorithm}

\begin{prevproof}{Theorem}{thm:appxdecomp}
It follows from the correctness of the decomposition algorithm of~\cite{CaiDW12} that the output allocation rule $M$ implements the input reduced form $\vec{\pi}'$ when bidders are sampled from ${\cal D'}$. Now it follows from Corollary~\ref{cor:otherway} that with probability at least $1-e^{-O(n\sum_{i=1}^{m}|T_i|/\epsilon)}$ (see the proof of Theorem~\ref{thm:appxSO} for why the probability guaranteed by Corollary~\ref{cor:otherway} is at least this large given our choice of parameters) it holds that $|\vec{\pi}- \vec{\pi}'|_\infty \leq \epsilon$ (again see the proof of Theorem~\ref{thm:appxSO} for why the guaranteed distance is at most $\epsilon$ given our choice of parameters). The bound on the running time follows directly from Corollary~\ref{cor:geometric alg for uniform} and our choice of parameters.
\end{prevproof}

\subsection{Approximate Algorithms for $F(\mathcal{F},\mathcal{D})$}

The results of this section are provided for completeness, but are not used elsewhere in the paper. Our goal is to use Theorems~\ref{thm:appxSO} and \ref{thm:appxdecomp} to obtain approximate algorithms for the original polytope $F(\mathcal{F},\mathcal{D})$. Namely, we want to obtain, with high probability, a separation oracle that is correct on all points that are $\epsilon$-away from the boundary of $F(\mathcal{F},\mathcal{D})$ (in $\ell_\infty$), and a decomposition algorithm that returns a distribution over virtual VCG allocation rules whose reduced form is within $\epsilon$ (in $\ell_\infty$) of any given feasible reduced form that is $\epsilon$-away from the boundary. Such algorithms are provided by the following two corollaries.

\begin{corollary}\label{cor:appxSO} Given our choice of $k,k'$ in Algorithm~\ref{alg:preprocess}, a separation oracle for $F(\mathcal{F},\mathcal{D}')$ is an approximate separation oracle for $F(\mathcal{F},\mathcal{D})$. Specifically, with probability $1-e^{-\Omega(n\sum_i |T_i|/\epsilon)}$, we have:
\begin{enumerate}
\item If the entire $\epsilon$-$\ell_{\infty}$ ball around $\vec{x}$ is outside $F(\mathcal{F},\mathcal{D})$, then the separation oracle outputs ``no.''
\item If the entire $\epsilon$-$\ell_{\infty}$ ball around $\vec{x}$ is inside $F(\mathcal{F},\mathcal{D})$, then the separation oracle outputs ``yes.''
\end{enumerate}
\end{corollary}

\begin{proof}
By Theorem~\ref{thm:appxSO}, with probability $1-e^{-\Omega(n\sum_i |T_i|/\epsilon)}$, any point $\vec{x} \in F(\mathcal{F},\mathcal{D}')$ has some $\vec{y} \in F(\mathcal{F},\mathcal{D})$ with $|\vec{x}-\vec{y}|_{\infty} \leq \epsilon$. Therefore, $\vec{y}$ is in $F(\mathcal{F},\mathcal{D})$ as well as the $\epsilon$-$\ell_{\infty}$ ball around $\vec{x}$, and the entire $\epsilon$-$\ell_{\infty}$ ball around $\vec{x}$ is not outside $F(\mathcal{F},\mathcal{D})$. Taking the contrapositive proves part 1).

For any point $\vec{x} \notin F(\mathcal{F},\mathcal{D}')$, there is a separation hyperplane $H$ (not going through $\vec{x}$) separating $\vec{x}$ from $F(\mathcal{F},\mathcal{D}')$. Grow an $\ell_{\infty}$ ball centered at $\vec{x}$ (i.e. a hypercube) until it hits $H$. Let $\vec{y}$ denote one point in the intersection of $H$ with the ball. The ball has some positive radius, $\delta$. So the distance from $\vec{x}$ to the closest point in $H$ ($\vec{y}$) is $\delta$. Let now $\vec{z} = \vec{x} + \frac{\epsilon}{\delta}(\vec{x} - \vec{y})$. It is clear that the closest point (in $\ell_{\infty}$) in $H$ to $\vec{z}$ is also $\vec{y}$, and the distance is $\epsilon + \delta$. Therefore, there is no point within $\epsilon$ of $\vec{z}$ in $F(\mathcal{F},\mathcal{D}')$. By Theorem~\ref{thm:appxSO}, with probability $1-e^{-\Omega(n\sum_i |T_i|/\epsilon)}$, this implies that $\vec{z} \notin F(\mathcal{F},\mathcal{D})$. As $|\vec{x} - \vec{z}|_{\infty} = \epsilon$, taking the contrapositive proves part 2).
\end{proof}

\begin{corollary}\label{cor:appxdecomp} With probability $1-e^{-\Omega(n\sum_i |T_i|/\epsilon)}$, the Algorithm~\ref{alg:decomposition} has the property that for any reduced form $\vec{\pi}$ such that the entire $\epsilon$-$\ell_{\infty}$ ball around $\vec{\pi}$ is inside $F(\mathcal{F},\mathcal{D})$, Algorithm~\ref{alg:decomposition} outputs a distribution over virtual VCG allocation rules whose reduced form $\vec{\pi}'$ satisfies $|\vec{\pi} - \vec{\pi}'|_{\infty} \leq \epsilon$.
\end{corollary}
\begin{proof}
Theorem~\ref{thm:appxdecomp} guarantees that with probability $1-e^{-\Omega(n\sum_i |T_i|/\epsilon)}$, Algorithm~\ref{alg:decomposition} is such that whenever $\vec{\pi}$ is in $F(\mathcal{F},\mathcal{D}')$, a $\vec{\pi}'$ with $|\vec{\pi} - \vec{\pi}'|_{\infty} \leq \epsilon$ is output. Corollary~\ref{cor:appxSO} guarantees that with probability $1-e^{-\Omega(n\sum_i |T_i|/\epsilon)}$, any $\vec{\pi}$ whose entire $\epsilon$-$\ell_{\infty}$ ball is contained in $F(\mathcal{F},\mathcal{D})$ is in $F(\mathcal{F},\mathcal{D}')$, which completes the proof.
\end{proof}
\section{Discussion and Proofs from Section~\ref{sec:revenue}}\label{app:revenue}

\paragraph{Approach.} Theorems~\ref{thm:general} through~\ref{thm:bounded} are obtained similarly to the corresponding Theorems 8 through 10 of~\cite{CaiDW12}. In~\cite{DW12}, linear programs are provided that exactly solve MDMDP in cases with finite support (and less general feasibility constraints, namely each bidder has an upper bound on the number of items she wants to receive, and every item should be allocated to at most one bidder). However, the proposed LPs maintain variables for every type profile $P$, denoting the probability that bidder $i$ receives item $j$ on profile $P$, resulting in LP size proportional to $|{\cal D}|$. In~\cite{CaiDW12} it is observed that these LPs can be made more efficient by making use of the reduced form, even if just a separation oracle is provided for the feasibility of the reduced form. Indeed, the reduced form of a mechanism contains sufficient information to verify truthfulness (given the additivity of the bidders), and a separation oracle for the feasibility of the reduced form is sufficient to optimize the expected revenue of the mechanism by solving an LP. Algorithm~\ref{alg: MDMD} and Figure~\ref{fig:LPMDMD} provide the details of how we apply this approach in our setting, culminating in a proof of Theorem~\ref{thm:general}. Simply put, we use the LP approach of~\cite{CaiDW12} with the following twist: (a) we use a separation oracle for the proxy polytope of feasible reduced forms $F(\mathcal{F},\mathcal{D}')$, obtained in Section~\ref{sec:approximations}, rather than the real polytope $F(\mathcal{F},\mathcal{D})$;  (b) still, we compute expected revenue for bidders sampled from the real distribution $\mathcal{D}$.

\begin{algorithm}[ht]
    \caption{FPRAS for solving MDMDP when $\mathcal{D}$ has finite support.}
    \begin{algorithmic}[1]\label{alg: MDMD}
        \STATE Input: $\mathcal{D}$, $\mathcal{F}$, $\epsilon$.
        \STATE Set $\delta = \frac{\epsilon}{2m}$. Run the pre-processing algorithm (Algorithm~\ref{alg:preprocess}) on input $\mathcal{D}$, with accuracy $\delta/2n$. Call the output distribution $\mathcal{D}'$.
        \STATE Let $SO(\vec{\pi})$ be the separation oracle that on input $\vec{\pi}$ executes the separation oracle of Section~\ref{sec:separation} on input $\vec{\pi}$ for distribution $\mathcal{D}'$ and feasibility constraints $\mathcal{F}$. 
        \STATE Using $SO$, solve the Linear Program of Figure~\ref{fig:LPMDMD}. Store the output as $\vec{\pi},\vec{p}$.\label{step: solve LP}
        \STATE Run the decomposition algorithm (Algorithm~\ref{alg:decomposition}) with input $\mathcal{F}, \mathcal{D}', \vec{\pi}$. Store the output as $M'$. $M'$ is a distribution over at most $n\sum_{i=1}^{m} |T_i| + 1$ simple virtual VCG allocations. \label{step:decompose}
        \STATE Output the allocation rule $M'$ and pricing rule $\vec{p} - \delta\cdot \vec{1}$ (i.e. when bidder $i$ reports type $A$, charge her $p_i(A)-\delta$). \label{step: last step}
     \end{algorithmic}
\end{algorithm}

\begin{figure}[ht]
\colorbox{MyGray}{
\begin{minipage}{\textwidth} {
\noindent\textbf{Variables:}
\begin{itemize}
\item $p_i(\vec{v}_i)$, for all bidders $i$ and types $\vec{v}_i \in T_i$, denoting the expected price paid by bidder $i$ when reporting type $\vec{v}_i$ over the randomness of the mechanism and the other bidders' types.
\item $\pi_{ij}(\vec{v}_i)$, for all bidders $i$, items $j$, and types $\vec{v}_i \in T_i$, denoting the probability that bidder $i$ receives item $j$ when reporting type $\vec{v}_i$ over the randomness of the mechanism and the other bidders' types.
\end{itemize}
\textbf{Constraints:}
\begin{itemize}
\item $\vec{\pi}_i(\vec{v}_i) \cdot \vec{v}_i - p_i(\vec{v}_i) \geq \vec{\pi}_i(\vec{w}_i)\cdot \vec{v}_i - p_i(\vec{w}_i) - \delta $, for all bidders $i$, and types $\vec{v}_i,\vec{w}_i \in T_i$, guaranteeing that the reduced form mechanism $(\vec{\pi},\vec{p})$ is $\delta$-BIC.
\item $\vec{\pi}_i(\vec{v}_i) \cdot \vec{v}_i - p_i(\vec{v}_i) \geq 0$, for all bidders $i$, and types $\vec{v}_i \in T_i$, guaranteeing that the reduced form mechanism $(\vec{\pi},\vec{p})$ is individually rational.
\item $SO(\vec{\pi}) = $``yes,'' guaranteeing that the reduced form $\vec{\pi}$ is in $F(\mathcal{F},\mathcal{D}')$.
\end{itemize}
\textbf{Maximizing:}
\begin{itemize}
\item $\sum_{i=1}^{m} \sum_{\vec{v}_i \in T_i} \Pr[t_i = \vec{v}_i]\cdot p_i(\vec{v}_i)$, the expected revenue \textbf{when played by bidders sampled from the true distribution $\mathcal{D}$}.\\
\end{itemize}}
\end{minipage}}
\caption{A linear programming formulation for MDMDP.}
\label{fig:LPMDMD}
\end{figure}

\medskip \begin{prevproof}{Theorem}{thm:general}
We use Algorithm~\ref{alg: MDMD}. Using the additivity of the bidders, it follows that Step~\ref{step: solve LP} of the algorithm outputs a reduced form/pricing rule pair $(\vec{\pi},\vec{p})$ that is revenue-optimal with respect to all $\delta$-BIC, IR reduced form/pricing rule pairs, except that the reduced forms that are searched over belong to $F(\mathcal{F},\mathcal{D}')$ and may be infeasible with respect to ${\cal D}$. We proceed to argue that the achieved revenue is nearly-optimal with respect to all BIC, IR reduced form/pricing rule pairs for which the reduced form lies inside $F(\mathcal{F},\mathcal{D})$. So let $(\vec{\pi}^*,\vec{p}^*)$ denote an optimal such pair, and let $\opt$ denote its expected revenue. By Theorem~\ref{thm:appxSO}, we know that, with high probability, there is some reduced form $\vec{\pi}' \in F(\mathcal{F},\mathcal{D}')$ satisfying {$|\vec{\pi}'-\vec{\pi}^{*}|_\infty \leq \delta/2n$}. So, if we let {$\vec{p}' = \vec{p}^* - \frac{\delta}{2}\cdot \vec{1}$,} it is obvious that the reduced form $(\vec{\pi}',\vec{p}')$ is $\delta$-BIC. It is also obvious that it is individually rational. Finally, it is clear that value of the LP achieved by $(\vec{\pi}{'},\vec{p}{'})$ is exactly {$m\delta/2  = \epsilon/4$} less than the value of $(\vec{\pi}^*,\vec{p}^*)$. So because $(\vec{\pi}',\vec{p}{'})$ is in the feasible region of the LP of Figure~\ref{fig:LPMDMD}, the reduced form/price rule pair output by Step~\ref{step: solve LP} of Algorithm~\ref{alg: MDMD} has expected revenue at least {$\opt - \epsilon/4$}. Noticing that we subtract an additional $\delta$ from the price charged to each bidder in Step~\ref{step: last step} of the algorithm, we get that the end price rule makes expected revenue at least $\opt - \epsilon$.

We argue next that the mechanism output by Algorithm~\ref{alg: MDMD} is $\epsilon$-BIC and IR. To see this let $M'$ be the allocation rule (computed in Step~\ref{step:decompose} of the algorithm), which implements the reduced form $\vec{\pi}$ (computed in Step~\ref{step: solve LP}) with respect to $\mathcal{D}'$.  Let also $\vec{\pi}'$ denote the reduced form of $M'$ with respect to $\mathcal{D}$. By Theorem~\ref{thm:appxdecomp}, we know that, with high probability, {$|\vec{\pi} - \vec{\pi}'|_\infty \leq\delta/2n$.} Therefore, given that $(\vec{\pi},\vec{p})$ is $\delta$-BIC, we immediately get that $(\vec{\pi}',{\vec{p}- \delta\cdot \vec{1}})$ is {$2\delta$}-BIC. {So allocation rule $M'$ with pricing rule $\vec{p}- \delta\cdot \vec{1}$  comprises an $\epsilon$-BIC mechanism. Also because $(\vec{\pi},\vec{p})$ is IR and $|\vec{\pi}-\vec{\pi}'|_\infty \leq \delta/2n$, we immediately get that $(\vec{\pi}',\vec{p} - \delta\cdot \vec{1})$ is IR, and hence that allocation rule $M'$ with pricing rule $\vec{p}- \delta\cdot \vec{1}$ is IR.}

Overall the above imply that the allocation rule $M'$ and the pricing rule $\vec{p}- \delta\cdot \vec{1}$ output in Step~\ref{step: last step} of our algorithm comprise an $\epsilon$-BIC and IR mechanism, whose pricing rule achieves revenue at least $OPT -\epsilon$.
Moreover, it is immediate from Theorems~\ref{thm:appxSO} and~\ref{thm:appxdecomp} that Algorithm~\ref{alg: MDMD} runs in time polynomial in {c}, $n$, $\sum_{i=1}^{m}|T_{i}|$, $1/\epsilon$, and $rt_{\mathcal{F}}(\poly(n\sum_{i=1}^{m} |T_i|, {\log 1/\epsilon, c}))$, {where $c$ is as in the statement of the theorem.}

{Finally, the way we chose our parameters in Algorithm~\ref{alg: MDMD} the probability of failure of the algorithm is $1-e^{-\Omega(m n^2 \sum_i |T_i|/\epsilon)}$ by Theorems ~\ref{thm:appxSO} and~\ref{thm:appxdecomp}. Trading off probability of error with $\epsilon$ (as per Remark~\ref{costasremark:FPRAS?}) we complete the proof of Theorem~\ref{thm:general}.}\end{prevproof}

{Theorems~\ref{thm:itemsym} and~\ref{thm:bounded} are obtained by combining Theorem~\ref{thm:general} with tools developed in~\cite{DW12}. Both theorems are based on the following observation, generalizing Theorem~2 of~\cite{DW12}: If ${\cal D}$ is item-symmetric and the feasibility constraints ${\cal F}$ are also item-symmetric then there exists an optimal mechanism that is:
\begin{enumerate}
\item item-symmetric, i.e. for all bidders $i$, all types $\vec{v}_i \in T_i$, and all item-permutations $\sigma$, the reduced form of the mechanism satisfies $\vec{\pi}_i(\sigma(\vec{v}_i)) = \sigma(\vec{\pi}_i(\vec{v}_i))$; this means that the reduced form on a permuted type of a bidder is the same permutation of the reduced form of the un-permuted type of the bidder.
\item strongly-monotone, i.e. for all bidders $i$, and items $j$ and $j'$, $v_{ij} \ge v_{ij'} \implies \pi_{ij}(\vec{v}_i) \ge \pi_{ij'}(\vec{v}_i)$.
\end{enumerate}
Using this structural observation for optimal mechanisms, we sketch the proofs of Theorems~\ref{thm:itemsym} and~\ref{thm:bounded}.

\smallskip \begin{prevproof}{Theorem}{thm:itemsym} (Sketch) Given our structural observation for optimal mechanisms, we can---without loss of generality---rewrite the LP of Figure~\ref{fig:LPMDMD}, while at the same time enforcing the above constraints, i.e. searching over the set of item-symmetric, strongly-monotone reduced forms. Indeed, to save on computation we can write a succinct LP on variables $\{\pi_{ij}(\vec{v}_i)\}_{i, j, \vec{v}_i \in E_i}$ where, for every bidder $i$, $E_i$ is a sufficient (based on the above symmetries) set of representative types, e.g. we can take $E_i = \{\vec{v}_i~|~v_{i1} \ge \ldots \ge v_{in} \}$. We refer the reader to~\cite{CaiDW12} for the explicit form of the succinct LP. The benefit of the succinct formulation is that $|E_i| \le n^c$, where $c$ is as in the statement of Theorem~\ref{thm:itemsym}, so the size of the succinct LP is polynomial in $m$, $n^c$ and $\ell$, where $\ell$ is as in the statement of the theorem. 

But we also need to come up with an efficient separation oracle for our setting. Checking violation of the strong-monotonicity property is easy to do in time linear in $O(m n)$ and the description of $\{\pi_{ij}(\vec{v}_i)\}_{i, j, \vec{v}_i \in E_i}$. So it remains to describe a separation oracle determining the feasibility of a succinct description $\{\pi_{ij}(\vec{v}_i)\}_{i, j, \vec{v}_i \in E_i}$ of an item-symmetric reduced form. One approach to this would be to expand out the succinct description of the item-symmetric reduced form to a full-fledged reduced-form and invoke the separation oracle developed in Sections~\ref{sec:algorithms} through~\ref{sec:approximations}. However, this would make us pay computation time polynomial in $\sum_i |T_i|$, and the whole point of using item-symmetries is to avoid this cost. To circumvent this, we take the following approach:
\begin{itemize}
\item First, let $F_S({\cal F},{\cal D})$ be the set of item-symmetric reduced forms that are feasible with respect to ${\cal F}$ and ${\cal D}$;

\item $F_S({\cal F},{\cal D})$ is a polytope, as it is the intersection of the polytope $F({\cal F},{\cal D})$ and the item-symmetry constraints; moreover, every point in $F_S({\cal F},{\cal D})$ has a succinct description of the form $\{\pi_{ij}(\vec{v}_i)\}_{i, j, \vec{v}_i \in E_i}$ where, for all $i$, $E_i$ is defined as above;

\item What are the corners of $F_S({\cal F},{\cal D})$? These can be implemented by item-symmetric virtual VCG allocation rules whose weight-functions are item-symmetric. The proof of this is identical to the proof of Proposition~\ref{prop:VCG} noticing that $F_S({\cal F},{\cal D})$ lies in the lower-dimensional space spanned by the item-symmetries. We note that we do not require the virtual VCG allocation rules to be {\em simple} in the same sense defined in Section~\ref{sec:independent}, as this could violate item-symmetry. 

\item However, given an  item-symmetric weight vector $\vec{w}$ how do we even run an item-symmetric virtual VCG allocation rule corresponding to $\vec{w}$? There are two issues with this: (a) how to enforce the item-symmetry of the virtual VCG allocation rule; and (b) there could be multiple item-symmetric virtual VCG allocation rules consistent with $\vec{w}$, e.g., if $\vec{w}$ is perpendicular to a facet of $F_S({\cal F},{\cal D})$. Here is how we resolve these issues: First, we apply Lemma~\ref{lem:tiebreaking} to the symmetric weight vector $\vec{w}$ to get a non-symmetric weight vector $\vec{w}'$ (we may do this transformation explicitly or do a lazy-evaluation of it---this is relevant only for our computational results three bullets down). When bidders submit their types, we {pick a permutation $\sigma$ uniformly at random}, and permute the names of the items. Then we use the simple virtual VCG allocation rule corresponding to $\vec{w}'$. Finally, we un-permute the names of the items {in the allocation}. {We denote this allocation rule by $S.VVCG_{\mathcal{F}}(\vec{w})$.} It is clear that $S.VVCG_{\mathcal{F}}(\vec{w})$ is well-defined (i.e. no tie-breaking will ever be required), is item-symmetric, and defines a virtual VCG allocation rule w.r.t. the original weight vector $\vec{w}$. Given the above discussion, every item-symmetric weight vector $\vec{w}$ has a succinct description of the form $\{w_{ij}(\vec{v}_i)\}_{i, j, \vec{v}_i \in E_i}$, which defines uniquely an item-symmetric virtual VCG allocation rule w.r.t. $\vec{w}$ {(namely, $S.VVCG_{\mathcal{F}}(\vec{w})$)}; 

\item Given the above definitions and interpretations, we can generalize the results of Sections~\ref{sec:independent} and~\ref{sec:algorithms} to the polytope $F_S({\cal F},{\cal D})$.
\item Next we discuss how to extend the computationally-friendly results of Section~\ref{sec:approximations} to the item-symmetric setting, while maintaining the computational complexity of all algorithms polynomial in  $m$, $\max_i |E_i| = O(n^c)$ and $\ell$, where $\ell$ is as in the statement of the theorem. We define an item-symmetric distribution ${\cal D'}$ as follows: We draw $k''=k+k' \sum_{i=1}^m |E_i|$ profiles of bidders $P_1,\ldots,P_ {k''}$ from ${\cal D}$ as in Section~\ref{sec:approximations}, except that, for each bidder $i$, we draw $k'$ profiles conditioning on the type of the bidder being each element of $E_i$ and not $T_i$. Then we define (without explicitly writing down) ${\cal D'}$ to be the {two-stage} distribution that {in the first stage} draws a random profile from $P_1,\ldots,P_{k''}$ and {in the second-stage} permutes the items using a uniformly random item-permutation. {We claim that using $T= \poly(m n^c)$ in Algorithm~\ref{alg:preprocess} suffices to obtain an analog of Theorems~\ref{thm:appxSO} and~\ref{thm:appxdecomp} for our setting with probability of success $1-e^{-1/\epsilon}$. (We address the running time shortly.) The reason we can save on the number of samples is that we are working with $F_S({\cal F},{\cal D})$ and hence all reduced forms are forced to be item-symmetric. So we need to prove concentration of measure for a smaller number of marginal allocation probabilities.}

\item Unfortunately, we cannot afford to {compute reduced forms with respect to ${\cal D}'$, as we can't in general} evaluate the reduced form of $S.VVCG_{\mathcal{F}}(\vec{w})$ {on a given type profile without making prohibitively many queries to $A_{\cal F}$}. Instead, we will also independently sample item permutations $\sigma_1,\ldots,\sigma_{k''}$, and {associate the permutation $\sigma_{\alpha}$ with $P_{\alpha}$ in the following sense. If a given type profile was sampled by ${\cal D}'$ after sampling $P_{\alpha}$ in the first stage of ${\cal D}'$, we will permute the items by (just) $\sigma_{\alpha}$ when evaluating $S.VVCG_{\mathcal{F}}(\vec{w})$ on that profile, instead of taking a uniformly random permutation. In other words, we have removed the randomness in evaluating $S.VVCG_{\mathcal{F}}(\vec{w})$ and fixed the applied item-permutation to some $\sigma_{\alpha}$, which was chosen uniformly at random. Doing so, we still have the same expectations as in the previous bullet, and we can deduce that the reduced form of $S.VVCG_{\mathcal{F}}(\vec{w})$ when consumers are sampled from $\mathcal{D}$ is very close to the reduced form of $S.VVCG_{\mathcal{F}}(\vec{w})$ when consumers are sampled from $\mathcal{D}'$ (while only using item-permutation $\sigma_{\alpha}$ to run $S.VVCG_{\mathcal{F}}(\vec{w})$ on all profiles sampled from ${\cal D}'$ after sampling $P_{\alpha}$ in the first stage of ${\cal D}'$, as explained above). }

\item With the above modifications, for both Theorems~\ref{thm:appxSO} and~\ref{thm:appxdecomp} the running time of the corresponding algorithm can be made polynomial in $\ell'$, $m$, $n^c$, $1/\epsilon$ and $rt_{\mathcal{F}}(\poly(n^c, m,\log{1/\epsilon},\ell'))$, where $\ell'$ is the max of $\ell$ (see statement) and the bit complexity of the coordinates of the input to the algorithms. To achieve this we only do computations with succinct descriptions of item-symmetric reduced forms and weight vectors. What we need to justify further is that we can do exact computations on these objects with respect to the distribution ${\cal D'}$ given oracle access to $A_{\cal F}$ in the afore-stated running time. For this it suffices to be able compute the succinct description of the reduced form $\vec{\pi}$ of {$S.VVCG_{\mathcal{F}}(\vec{w})$ (using permutation $\sigma_{\alpha}$ on profiles coming from $P_{\alpha}$ as explained above).} The small obstacle is that the support of ${\cal D}'$ is not polynomial in the required running time, but it suffices to do the following. For each profile $P_{\alpha}$ find the allocation output by the (non item-symmetric) virtual VCG allocation rule corresponding to the perturbed vector $\vec{w}'$, {after relabeling the items according to $\sigma_{\alpha}$.} This we can do in the allotted running time with a lazy evaluation of the perturbation. Then, for all $i$, $\vec{v}_i$ and $j$, to compute ${\pi}_{ij}(\vec{v}_i)$ do the following: for all profiles $P_{\alpha}$, let $x_{\alpha}(i,\vec{v}_i)$ denote the number of $\tau$ such that $\tau(\vec{v}_i)$ matches the type $t_i(\alpha)$ of bidder $i$ in $P_{\alpha}$. Let $y_{\alpha}(i,\vec{v}_i)$ denote the number of $\tau$ such that $\tau(\vec{v}_i)$ matches the type of bidder $i$ in $P_{\alpha}$, and item $\tau(j)$ is awarded to bidder $i$. It is easy to compute $x_{\alpha}(i,\vec{v}_i)$: let $J^1_v = \{j | v_{ij} = v\}$, and $J^2_v$ be the set of items that bidder $i$ values at $v$ in profile $P_{\alpha}$. Then if $|J^1_v| \neq |J^2_v|$ for any $v$, $x_{\alpha}(i,\vec{v}_i) = 0$. Otherwise, $x_{\alpha}(i,\vec{v}_i) = \prod_v |J^1_v|!$, because $\tau(\vec{v}_i)$ matches the type of bidder $i$ iff $\tau$ maps all of $J^1_v$ to $J^2_v$ for all $v$. Computing $y_{\alpha}(i,\vec{v}_i)$ is also easy: simply break up $J^2_v$ into two sets: $J^2_v(W)$ of items that bidder $i$ values at $v$ and wins, and $J^2_v(L)$ of items that bidder $i$ values at $v$ and loses. Then if $x_{\alpha}(i,\vec{v}_i) \neq 0$, $y_{\alpha}(i,\vec{v}_i) = |J^2_{v_{ij}}(W)|\cdot (|J^2_{v_{ij}}|-1)! \cdot \prod_{v \neq v_{ij}} |J^2_v|!$, because bidder $i$ is awarded item $\tau(j)$ and $\tau(\vec{v}_i)$ matches the type of bidder $i$ iff $\tau(j) \in J^2_{v_{ij}}(W)$, and $\tau$ maps $J^1_v$ to $J^2_v$ for all $v$. Once we've computed $x_{\alpha}(i,\vec{v}_i)$ and $y_{\alpha}(i,\vec{v}_i)$, it is easy to see that:

$$\pi_{ij}(\vec{v}_i) = \frac{1}{|\{\alpha | x_{\alpha}(i,\vec{v}_i) > 0\}|}\sum_{\alpha | x_{\alpha}(i,\vec{v}_i) > 0} \frac{y_{\alpha}(i,\vec{v}_i)}{x_{\alpha}(i,\vec{v}_i)}.$$

\notshow{
\item Next we discuss how to extend the computationally-friendly results of Section~\ref{sec:approximations} to the item-symmetric setting, while maintaining the computational complexity of all algorithms polynomial in  $m$, $\max_i |E_i| = O(n^c)$ and $\ell$, where $\ell$ is as in the statement of the theorem. We define an item-symmetric distribution ${\cal D'}$ as follows: We draw $k''=k+k' \sum_{i=1}^m |E_i|$ profiles of bidders $P_1,\ldots,P_ {k''}$ from ${\cal D}$ as in Section~\ref{sec:approximations}, except that, for each bidder $i$, we draw $k'$ profiles conditioning on the type of the bidder being each element of $E_i$ and not $T_i$. Then we define (without explicitly writing down) ${\cal D'}$ to be the distribution that first draws a random profile from $P_1,\ldots,P_{k''}$ and then permutes the items using a uniformly random item-permutation. We claim that using $T= \poly(m n^c)$ in Algorithm~\ref{alg:preprocess} suffices to obtain an analog of Theorems~\ref{thm:appxSO} and~\ref{thm:appxdecomp} for our setting with probability of success $1-e^{-1/\epsilon}$ and better running time discussed in the next bullet. The reason we can save on the number of samples is that we are working on $F_S({\cal F},{\cal D})$ and hence all reduced forms are forced to be item-symmetric. So we need to prove concentration of measure for a smaller number of marginal allocation probabilities.

\item For both theorems the running time of the corresponding algorithm can be polynomial in $\ell'$, $m$, $n^c$, $1/\epsilon$ and $rt_{\mathcal{F}}(\poly(n^c, m,\log{1/\epsilon},\ell'))$, where $\ell'$ is the max of $\ell$ (see statement) and the bit complexity of the coordinates of the input to the algorithms. To do this we only do computations with succinct descriptions of item-symmetric reduced forms and weight vectors. What we need to justify is that we can do exact computations on these objects with respect to the distribution ${\cal D'}$ given oracle access to ${\cal A}_{\cal F}$ in the afore-stated running time. For this it suffices to be able compute the succinct description of the reduced form $\vec{\pi}$ of the item-symmetric virtual VCG allocation rule corresponding to a succinctly-specified item-symmetric $\vec{w}$ in the allotted time. The small obstacle is that the support of ${\cal D}'$ is not polynomial in this time. But it suffices to do the following. For each profile $P_{\alpha}$ find the allocation output by the (non item-symmetric) virtual VCG allocation rule corresponding to the perturbed vector $\vec{w}'$, as per the perturbation defined earlier. This we can do in the allotted running time with a lazy evaluation of the perturbation. Then, for all $i$, $\vec{v}_i$ and $j$, to compute ${\pi}_{ij}(\vec{v}_i)$ do the following: for all profiles $P_{\alpha}$, sum over all item permutations $\sigma$ the indicator function of whether the type of bidder $i$ in $P_{\alpha}$ is $\sigma(\vec{v}_i)$ and whether item $\sigma(j)$ was given to bidder $i$. Then sum these numbers over all profiles and divide out by $k'' n!$.}
\end{itemize} 
The above bullet points explain briefly how our ideas from the previous sections are modified for item-symmetric distributions. The complete details are omitted.\end{prevproof}}

\begin{prevproof}{Theorem}{thm:bounded} We combine Theorem~\ref{thm:itemsym} with (i) a discretization of the hypercube so that every $\vec{v}$ in the support of the distribution has $v_{ij} = k\delta$, $k \in \mathbb{N}$, for all $i,j$; and (ii) the approximately-BIC to BIC reduction of Section 6 of~\cite{DW12}, informally stated below.
\begin{informaltheorem}(Reworded from~\cite{DW12}) Let $\mathcal{C} = \times_i \mathcal{C}_i$ and $\mathcal{C}' = \times_i \mathcal{C}'_i$ be product distributions sampling every additive bidder independently from $[0,1]^n$. Suppose that, for all $i$, $\mathcal{C}_i$ and $\mathcal{C}'_i$ can be coupled so that with probability $1$, $\vec{v}_i$ sampled from $\mathcal{C}_i$ and $\vec{v}'_i$ sampled from $\mathcal{C}'_i$ satisfy $v_{ij} \geq v'_{ij} \geq v_{ij} - \delta$ for all $j$. If $M'$ is any $\epsilon$-BIC mechanism for $\mathcal{C}'$, then with \emph{exact} knowledge of the reduced form of $M'$ with respect to $\mathcal{C}'$, we can transform $M'$ into a BIC mechanism for $\mathcal{C}$ while only losing {$O(C(\sqrt{\delta} + \sqrt{\epsilon}))$} revenue, {where $C$ is the maximum number of items that are allowed to be allocated simultaneously.} In item symmetric settings, the reduction runs in time polynomial in $n^{1/\delta},m$.
\end{informaltheorem}
 While doing the discretization is straihgtforward, there is an issue with applying the aforementioned approximately-BIC to BIC reduction. To directly apply the reduction in our setting, one might try to take $\mathcal{C}'$ to be the $\mathcal{D}'$ from the sampling procedure. Unfortunately, this doesn't work because $\mathcal{D}'$ is correlated across bidders. Instead, we might try to take $\mathcal{C}'$ to be $\mathcal{D}$. This too doesn't work because we can't exactly compute the reduced form of a mechanism with respect to $\mathcal{D}$. We can, however, compute the reduced form of a mechanism with respect to $\mathcal{D}$ with quite good accuracy. So we will prove a quick lemma about the quality of the reduction proposed in~\cite{DW12} in our setting. Virtually the same lemma is used in~\cite{HartlineKM11} where the ideas behind this reduction originated, but in a different setting. 

\begin{lemma}\label{lem:reduction} Let $\mathcal{C}$ and $\mathcal{C}'$ satisfy the hypotheses of Theorem 4 in~\cite{DW12}, and let $M'$ be a $\gamma$-BIC mechanism whose reduced form with respect to $\mathcal{C}'$ is $\vec{\pi}'$ (which is possibly unknown). Then with knowledge of some $\vec{\pi}$ such that $|\vec{\pi}-\vec{\pi}'|_1 \leq \epsilon$, the reduction of~\cite{DW12} transforms $M'$ into a $2\epsilon$-BIC mechanism for $\mathcal{C}$ while only losing {$O(C(\sqrt{\delta} + \sqrt{\gamma}))$} revenue, where $C$ is as above.
\end{lemma}

\begin{prevproof}{Lemma}{lem:reduction} We avoid repeating a complete description of the reduction and refer the reader to~\cite{DW12} for more details. At a high level, the reduction is the following. The new mechanism $M$ is a two-stage mechanism. First, each bidder $i$, independently from the other bidders, plays a VCG auction against make-believe replicas of herself, drawn independently from ${\cal C}_i$, to purchase a surrogate from a collection of surrogates drawn independently from ${\cal C}'_i$.  (The items of the per-bidder VCG auction are the surrogates and, in particular, they have nothing to do with the items that $M$ is selling. Moreover, the feasibility constraints of the VCG auction are just unit-demand constraints on the bidder-side and unit-supply constraints on the item-side.) After each bidder buys a surrogate, the purchased surrogates play $M'$ against each other. Each bidder receives the items their surrogate is awarded and pays the price the surrogate pays in $M'$, as well as a little extra in order to buy the surrogate in the per-bidder VCG auction. The truthfulness of the two-stage mechanism $M$ boils down to the truthfulness of VCG: If we can exactly evaluate the value of bidder $i$ of type $\vec{v}_i$ for being represented by each surrogate $\vec{s}_i$, and we use these values in computing the VCG allocation in the per-bidder VCG auction, then $M$ is  BIC. Moreover, it is shown in~\cite{DW12} that the distribution of surrogates that play $M'$ is exactly $\mathcal{C}'$. So, this implies that, if we know exactly the reduced form of $M'$ with respect to $\mathcal{C}'$, we can exactly compute the value of the bidder for each surrogate, and $M$ will be BIC.

If we only know the reduced form of $M'$ with respect to $\mathcal{C}'$ within $\epsilon$ in $\ell_1$-distance, then we cannot exactly evaluate the value of bidder $\vec{v}_i$ for being represented by surrogate $\vec{s}_i$, but we can evaluate it within $\epsilon$. Suppose that we run per-bidder VCG auctions using our estimates. In the VCG auction corresponding to bidder $i$, suppose $\vec{v}_i$ is the true type of the bidder, let $\vec{\pi}_1$ be our estimate of the reduced form of the surrogate that was sold to the bidder, let $p_1$ be the price that surrogate pays in $M'$, and let $q_1$ denote the price paid for that surrogate in VCG. Let $(\vec{\pi}_2,p_2,q_2)$ be the corresponding triplet for another surrogate that the bidder would win by misreporting her type. By the truthfulness of VCG, 
$$\vec{v}_i \cdot \vec{\pi}_1 - p_1 - q_1 \geq \vec{v}_i \cdot \vec{\pi}_2 - p_2 - q_2.$$
The question is how much the bidder regrets not misreporting to the VCG auction given that the true reduced form of surrogate $\vec{s}_i$ in $M'$ is some $\vec{\pi}'_i$, which was false-advertised in the VCG auction as $\vec{\pi}_i$, $i=1,2$.  Given that $|\vec{\pi}'_i - \vec{\pi}_i|_1 \le \epsilon$ and remembering that each $\vec{v}_{i} \in [0,1]^n$ we get:
\begin{align*}
\vec{v}_i \cdot \vec{\pi}'_1 &\geq \vec{v}_i \cdot \vec{\pi}_1 - \epsilon \\
\vec{v}_i \cdot \vec{\pi}_2 &\geq \vec{v}_i \cdot \vec{\pi}'_2- \epsilon.
\end{align*}
Therefore, $$\vec{v}_i \cdot \vec{\pi}'_1 - p_1 - q_1 \geq \vec{v}_i \cdot \vec{\pi}'_2 - p_2-q_2 -2\epsilon.$$
This means that if bidder $i$ were to misreport her type to get surrogate $\vec{s}_2$, her true utility would increase by at most $2\epsilon$. So $M$ is $2\epsilon$-BIC.

Finally, we can use the same argument as in \cite{DW12} to show that the reduction loses at most {$O(C(\sqrt{\delta}+\sqrt{\gamma}))$ in revenue.}\end{prevproof}

\noindent Coming back to the proof of Theorem~\ref{thm:bounded}, if we just discretized the value distribution without running the approximately-BIC to BIC reduction of Lemma~\ref{lem:reduction}, we would get a mechanism that is $(\epsilon+\delta)$-BIC and suboptimal by $(\epsilon+\delta)C$ (Lemma 3 of~\cite{DW12} pins down the loss of truthfulness/revenue due to discretization of the value distribution in multiples of $\delta$, and Theorem~\ref{thm:itemsym} bounds the additional loss due to computational constraints). If we apply the reduction of Lemma~\ref{lem:reduction} on this mechanism, we can turn it into one that is $\epsilon$-BIC and suboptimal by $O(\sqrt{\epsilon}+\sqrt{\delta})C$. The reason we even bother running the approximately-BIC to BIC reduction when we don't get a truly BIC mechanism in the end is because, while keeping our algorithm efficient, $\epsilon$ can be made as small as $1/\poly(n,m)$, while $\delta$ needs to stay a fixed constant. So after our reduction we obtained a qualitatively stronger result, namely one whose distance from truthfulness can be made arbitrarily small in polynomial time.
\end{prevproof}


\section{Accommodating Budget Constraints}\label{app:budgets}
In this section we show a simple modification to our solutions that allows them to accommodate budget constraints as well. To do this, we simply cite an observation from~\cite{DW12}. There, it is observed that if the solution concept is interim individual rationality, then the LP that finds the revenue-optimal reduced form can be trivially modified to accommodate budget constraints. In Figure~\ref{fig:LPbudgets} we show how to modify our LP from Figure~\ref{fig:LPMDMD} to accommodate budget constraints.

\begin{figure}[ht]
\colorbox{MyGray}{
\begin{minipage}{\textwidth} {
\noindent\textbf{Variables:}
\begin{itemize}
\item $p_i(\vec{v}_i)$, for all bidders $i$ and types $\vec{v}_i \in T_i$, denoting the expected price paid by bidder $i$ when reporting type $\vec{v}_i$ over the randomness of the mechanism and the other bidders' types.
\item $\pi_{ij}(\vec{v}_i)$, for all bidders $i$, items $j$, and types $\vec{v}_i \in T_i$, denoting the probability that bidder $i$ receives item $j$ when reporting type $\vec{v}_i$ over the randomness of the mechanism and the other bidders' types.
\end{itemize}
\textbf{Constraints:}
\begin{itemize}
\item $\vec{\pi}_i(\vec{v}_i) \cdot \vec{v}_i - p_i(\vec{v}_i) \geq \vec{\pi}_i(\vec{w}_i)\cdot \vec{v}_i - p_i(\vec{w}_i) - \delta $, for all bidders $i$, and types $\vec{v}_i,\vec{w}_i \in T_i$, guaranteeing that the reduced form mechanism $(\vec{\pi},\vec{p})$ is $\delta$-BIC.
\item $\vec{\pi}_i(\vec{v}_i) \cdot \vec{v}_i - p_i(\vec{v}_i) \geq 0$, for all bidders $i$, and types $\vec{v}_i \in T_i$, guaranteeing that the reduced form mechanism $(\vec{\pi},\vec{p})$ is individually rational.
\item $SO(\vec{\pi}) = $``yes,'' guaranteeing that the reduced form $\vec{\pi}$ is in {$F(\mathcal{F},\mathcal{D}')$}.
\item $p_i(\vec{v}_i) \leq B_i$, for all bidders $i$ and types $\vec{v}_i \in T_i$, \textbf{guaranteeing that no bidder $i$ pays more than their budget $B_i$.}
\end{itemize}
\textbf{Maximizing:}
\begin{itemize}
\item $\sum_{i=1}^{m} \sum_{\vec{v}_i \in T_i} \Pr[t_i = \vec{v}_i]\cdot p_i(\vec{v}_i)$, the expected revenue {when played by bidders sampled from the true distribution $\mathcal{D}$}.\\
\end{itemize}}
\end{minipage}}
\caption{A linear programming formulation for MDMDP that accommodates budget constraints.}
\label{fig:LPbudgets}
\end{figure}
It is also shown in~\cite{DW12} that accommodating budget constraints comes at a cost. First, it is shown that without budget constraints, one can turn any interim-IR mechanism into an ex-post IR mechanism with no loss in revenue. However, with budget constraints, there is a potentially large gap between the revenue of the optimal ex-post IR mechanism and the optimal interim IR mechanism. In other words, accommodating budget constraints requires accepting interim IR instead of ex-post IR. Second, the approximately-BIC to BIC reduction of~\cite{DW12} (used in the proof of Theorem~\ref{thm:bounded}) does not respect budget constraints. So to accommodate budgets in Theorem~\ref{thm:bounded} the output mechanism needs to be $\delta$-BIC instead of $\epsilon$-BIC.\footnote{Recall that the runtime required to find and execute the mechanism of Theorem~\ref{thm:bounded} is polynomial in $1/\epsilon$ but exponential in $1/\delta$.} Relative to the ability to naturally accommodate budget constraints, these costs are minor, but we state them in order to correctly quantify the settings our techniques solve.
\bibliographystyle{plain}
\bibliography{costasbib}

\begin{thebibliography}{10}

\bibitem{Alaei11}
Saeed Alaei.
\newblock {Bayesian Combinatorial Auctions: Expanding Single Buyer Mechanisms
  to Many Buyers}.
\newblock In {\em the 52nd Annual IEEE Symposium on Foundations of Computer
  Science (FOCS)}, 2011.

\bibitem{AlaeiFHHM12}
Saeed Alaei, Hu~Fu, Nima Haghpanah, Jason Hartline, and Azarakhsh Malekian.
\newblock {Bayesian Optimal Auctions via Multi- to Single-agent Reduction}.
\newblock In {\em the 13th ACM Conference on Electronic Commerce (EC)}, 2012.
\newblock {Posted to arXiv on March 22, 2012: http://arxiv.org/abs/1203.5099}.

\bibitem{BeiH11}
Xiaohui Bei and Zhiyi Huang.
\newblock {Bayesian Incentive Compatibility via Fractional Assignments}.
\newblock In {\em the Twenty-Second Annual ACM-SIAM Symposium on Discrete
  Algorithms (SODA)}, 2011.

\bibitem{BhattacharyaGGM10}
Sayan Bhattacharya, Gagan Goel, Sreenivas Gollapudi, and Kamesh Munagala.
\newblock {Budget Constrained Auctions with Heterogeneous Items}.
\newblock In {\em the 42nd ACM Symposium on Theory of Computing (STOC)}, 2010.

\bibitem{Border91}
Kim~C. Border.
\newblock Implementation of reduced form auctions: A geometric approach.
\newblock {\em Econometrica}, 59(4):1175--1187, 1991.

\bibitem{BriestCKW10}
Patrick Briest, Shuchi Chawla, Robert Kleinberg, and S.~Matthew Weinberg.
\newblock {Pricing Randomized Allocations}.
\newblock In {\em the Twenty-First Annual ACM-SIAM Symposium on Discrete
  Algorithms (SODA)}, 2010.

\bibitem{CaiD11}
Yang Cai and Constantinos Daskalakis.
\newblock {Extreme-Value Theorems for Optimal Multidimensional Pricing}.
\newblock In {\em the 52nd Annual IEEE Symposium on Foundations of Computer
  Science (FOCS)}, 2011.

\bibitem{CaiDW12}
Yang Cai, Constantinos Daskalakis, and S.~Matthew Weinberg.
\newblock {An Algorithmic Characterization of Multi-Dimensional Mechanisms}.
\newblock In {\em the 44th Annual ACM Symposium on Theory of Computing (STOC)},
  2012.
\newblock Posted to arXiv on Dec 20, 2011: http://arxiv.org/abs/1112.4572.

\bibitem{CaiH12}
Yang Cai and Zhiyi Huang.
\newblock {Simple and Nearly Optimal Multi-Item Auctions}.
\newblock {\em Manuscript}, 2012.

\bibitem{ChawlaHK07}
Shuchi Chawla, Jason~D. Hartline, and Robert~D. Kleinberg.
\newblock {Algorithmic Pricing via Virtual Valuations}.
\newblock In {\em the 8th ACM Conference on Electronic Commerce (EC)}, 2007.

\bibitem{ChawlaHMS10}
Shuchi Chawla, Jason~D. Hartline, David~L. Malec, and Balasubramanian Sivan.
\newblock {Multi-Parameter Mechanism Design and Sequential Posted Pricing}.
\newblock In {\em the 42nd ACM Symposium on Theory of Computing (STOC)}, 2010.

\bibitem{ChawlaMS10}
Shuchi Chawla, David~L. Malec, and Balasubramanian Sivan.
\newblock {The Power of Randomness in Bayesian Optimal Mechanism Design}.
\newblock In {\em the 11th ACM Conference on Electronic Commerce (EC)}, 2010.

\bibitem{CM85}
Jacques Cremer and Richard~P. McLean.
\newblock Optimal selling strategies under uncertainty for a discriminating
  monopolist when demands are interdependent.
\newblock {\em Econometrica}, 53(2):345--361, 1985.

\bibitem{CM88}
Jacques Cremer and Richard~P. McLean.
\newblock Full extraction of the surplus in bayesian and dominant strategy
  auctions.
\newblock {\em Econometrica}, 56(6):1247--1257, 1988.

\bibitem{DW12}
Constantinos Daskalakis and S.~Matthew Weinberg.
\newblock {Symmetries and Optimal Multi-Dimensional Mechanism Design}.
\newblock In {\em the 13th ACM Conference on Electronic Commerce (EC)}, 2012.
\newblock Posted to arXiv on Dec 17, 2011: http://arxiv.org/abs/1112.4006.

\bibitem{DobzinskiFK11}
Shahar Dobzinski, Hu~Fu, and Robert~D. Kleinberg.
\newblock {Optimal Auctions with Correlated Bidders are Easy}.
\newblock In {\em the 43rd ACM Symposium on Theory of Computing (STOC)}, 2011.

\bibitem{HartlineKM11}
Jason~D. Hartline, Robert Kleinberg, and Azarakhsh Malekian.
\newblock {Bayesian Incentive Compatibility via Matchings}.
\newblock In {\em the Twenty-Second Annual ACM-SIAM Symposium on Discrete
  Algorithms (SODA)}, 2011.

\bibitem{HartlineL10}
Jason~D. Hartline and Brendan Lucier.
\newblock {Bayesian Algorithmic Mechanism Design}.
\newblock In {\em the 42nd ACM Symposium on Theory of Computing (STOC)}, 2010.

\bibitem{Hoeffding63}
Wassily Hoeffding.
\newblock Probability inequalities for sums of bounded random variables.
\newblock {\em Journal of the American Statistical Association},
  58(301):13--30, 1963.

\bibitem{KleinbergW12}
Robert Kleinberg and S.~Matthew Weinberg.
\newblock Matroid prophet inequalities.
\newblock In {\em the 44th Annual ACM Symposium on Theory of Computing (STOC)},
  2012.
\newblock Posted to arXiv on Jan 23, 2012: http://arxiv.org/abs/1201.4764.

\bibitem{MR84}
Eric Maskin and John Riley.
\newblock {Optimal Auctions with Risk Averse Buyers.}
\newblock {\em Econometrica}, 52(6):1473--1518, 1984.

\bibitem{Matthews84}
Steven Matthews.
\newblock {On the Implementability of Reduced Form Auctions.}
\newblock {\em Econometrica}, 52(6):1519--1522, 1984.

\bibitem{MR92}
R.~Preston McAfee and Philip~J. Reny.
\newblock Correlated information and mechanism design.
\newblock {\em Econometrica}, 60(2):395--421, 1992.

\bibitem{Myerson81}
Roger~B. Myerson.
\newblock {Optimal Auction Design}.
\newblock {\em Mathematics of Operations Research}, 6(1):58--73, 1981.

\end{thebibliography}

\end{document}